\documentclass[envcountsect,runningheads]{llncs}

\usepackage{graphicx}

\usepackage{mathtools}
\usepackage{setspace}
\usepackage{ebproof}
\usepackage{stmaryrd}
\usepackage{amssymb}
\usepackage{tabularx}
\usepackage{marvosym}
\usepackage{breakcites}
\usepackage{hyperref}
\usepackage[normalem]{ulem}
\usepackage{listings}
\usepackage{xcolor}
\usepackage{float}
\usepackage{todonotes}
\usepackage{thmtools}
\usepackage{thm-restate}

\restylefloat{figure}
\MakeRobust{\ref}

\makeatletter
\newcommand{\labeltext}[2]{%
  \@bsphack
  \csname phantomsection\endcsname 
  \def\@currentlabel{#1}{\label{#2}}%
  \@esphack
}
\makeatother
\begin{document}
%
%
%
%
%
\definecolor{pcolor}{RGB}{156,90,107}
\definecolor{tcolor}{RGB}{0,102,204}
\definecolor{turquoise}{RGB}{0,153,153}
\definecolor{light-blue}{RGB}{51,153,255}
\definecolor{dark-blue}{RGB}{0,0,204}
\definecolor{purple}{RGB}{76,0,153}
\definecolor{red}{RGB}{212,1,1}

\newcommand{\hide}[1]{}
\newif\ifprettysyntax
\prettysyntaxtrue
\ifprettysyntax
\newcommand{\zero}{{\textcolor{pcolor}{\mathsf 0}}}
\newcommand{\fwd}[2]{\textcolor{pcolor}{\mathsf{fwd} }\ #1 \ #2} 
\newcommand{\fwdA}[3]{\textcolor{pcolor}{\mathsf{fwd}}_{#1} \ #2 \ #3} 
\newcommand{\mix}[2]{#1 \ \textcolor{pcolor}{\mathsf{||}} \ #2}
\newcommand{\shareP}{\textcolor{pcolor}{\mathsf{share}}}
\newcommand{\cutP}{\textcolor{pcolor}{\mathsf{cut}}}
\newcommand{\cutPc}{\textcolor{pcolor}{\mathsf{pcut}}}
\newcommand{\cutPi}{\textcolor{pcolor}{\mathsf{cut!}}}
\newcommand{\mixP}{\textcolor{pcolor}{\mathsf{par}}}
\newcommand{\mixD}{\textcolor{pcolor}{\mathsf{||}}}
\newcommand{\srecc}[3]{#1\left< #2 \right> #3}

\newcommand{\cutD}[1]{\textcolor{pcolor}{\mathsf{|}} #1 \textcolor{pcolor}{\mathsf{|}}}
\newcommand{\cut}[3]{#2   \ \textcolor{pcolor}{\mathsf{|}} #1 \textcolor{pcolor}{\mathsf{|}}   \ #3}
\newcommand{\cutB}[4]{#2.#3 \  \textcolor{pcolor}{\mathsf{|}}!#1\textcolor{pcolor}{\mathsf{|}} \ #4}
\newcommand{\cutp}[3]{ \cutPc \;\{ #2   \ \textcolor{pcolor}{\mathsf{|}} #1 \textcolor{pcolor}{\mathsf{|}}   \ #3 \}}
\newcommand{\cuti}[3]{ \cutP \;\{ #2   \ \textcolor{pcolor}{\mathsf{|}} #1 \textcolor{pcolor}{\mathsf{|}}   \ #3 \}}
\newcommand{\cutBi}[4]{\cutPi \;\{ #2.#3 \  \textcolor{pcolor}{\mathsf{|}}!#1\textcolor{pcolor}{\mathsf{|}} \ #4  \}}
\newcommand{\share}[3]{\textcolor{pcolor}{\mathsf{share}} \ #1 \ \{  #2 \ \textcolor{pcolor}{||} \ #3 \}}
\newcommand{\Blang}[0]{\mathsf{CLLB}}
\newcommand{\withm}[3]{\with_{#1\in #2}{#3}}
\newcommand{\oplusm}[3]{\oplus_{#1\in #2}{#3}}
\newcommand{\labl}[1]{\# \mathsf{#1}}
\newcommand{\pcasem}[4]{\textcolor{pcolor}{\mathsf{case}} \ #1 \  \{ 
|{\labl{#2}\in #3}{:} #4 \}} 
\newcommand{\nsum}[2]{#1 \  \textcolor{pcolor}{+} \ #2}
\newcommand{\pone}[1]{\textcolor{pcolor}{\mathsf{close}} \ #1}
\newcommand{\pbot}[2]{\textcolor{pcolor}{\mathsf{wait}} \ #1;#2}
\newcommand{\pparl}[3]{\textcolor{pcolor}{\mathsf{recv}} \ #1(#2);#3}
\newcommand{\potimes}[4]{\textcolor{pcolor}{\mathsf{send}} \ #1(#2.#3);#4}  
\newcommand{\fout}[3]{\textcolor{pcolor}{\mathsf{send}} \ #1(#2);#3}
\newcommand{\pcase}[3]{\textcolor{pcolor}{\mathsf{case}} \ #1 \  \{ |\#\mathsf{inl}{:} #2 \ | \#\mathsf{inr}{:} #3\}} 

\newcommand{\ass}{{:}}
\newcommand{\E}{\mathcal{E}}
\newcommand{\B}{\mathcal{B}}
\newcommand{\G}{\mathcal{G}}
\newcommand{\F}{\mathcal{F}}
\newcommand{\C}{\mathcal{C}}
\newcommand{\D}{\mathcal{D}}
\newcommand{\gcase}[2]{\textcolor{pcolor}{\mathsf{case}} \ #1 \  \{ #2\}} 
\newcommand{\cleft}[2]{\#\textcolor{pcolor}{\mathsf{inl}}\ #1;#2}
\newcommand{\clab}[3]{\#\textcolor{pcolor}{\mathsf{#3}}\ #1;#2}
\newcommand{\cright}[2]{\#\textcolor{pcolor}{\mathsf{inr}}\ #1;#2}
\newcommand{\choice}[3]{\mathsf{#1} \ #2; #3}
\newcommand{\pbang}[3]{\textcolor{pcolor}{\mathsf{!}}#1(#2);#3} 
\newcommand{\pwhy}[2]{\textcolor{pcolor}{\mathsf{?}}#1; #2}
\newcommand{\pcopy}[3]{\textcolor{pcolor}{\mathsf{call}} \ #1(#2); #3} 
\newcommand{\ncell}[3]{\textcolor{pcolor}{\mathsf{cell}} \ #1(#2.#3)}
\newcommand{\nfree}[1]{\textcolor{pcolor}{\mathsf{release}} \ #1} 
\newcommand{\rd}[3]{\textcolor{pcolor}{\mathsf{read}} \ #1(#2);#3}
\newcommand{\nwrt}[4]{\textcolor{pcolor}{\mathsf{wrt}} \ #1(#2.#3);#4} 
\newcommand{\fwrt}[3]{\textcolor{pcolor}{\mathsf{wrt}} \ #1 \ #2;#3}
\newcommand{\fwdE}[2]{\textcolor{pcolor}{\mathsf{fwd}^!} \ #1 \ #2}
\newcommand{\fwdEA}[3]{\textcolor{pcolor}{\mathsf{fwd}^!}_{#1} \ #2 \ #3}
\newcommand{\cellE}[3]{\textcolor{pcolor}{\mathsf{cell}_!} \ #1 \ #2 \ #3 }
\newcommand{\pexists}[3]{\textcolor{pcolor}{\mathsf{sendty}} \ #1(#2); #3}
\newcommand{\pforall}[3]{\textcolor{pcolor}{\mathsf{recvty}} \ #1(#2);#3} 
\newcommand{\lock}[2]{\textcolor{pcolor}{\mathsf{lock}} \ #1;#2} 
\newcommand{\unlock}[2]{\textcolor{pcolor}{\mathsf{unlock}} \ #1;#2} 
\newcommand{\shareL}[3]{\textcolor{pcolor}{\mathsf{shareL}} \ #1 \ \{  #2 \ \textcolor{pcolor}{||} \ #3  \}}
\newcommand{\shareR}[3]{\textcolor{pcolor}{\mathsf{shareR}} \ #1 \ \{   #2 \ \textcolor{pcolor}{||} \ #3  \}}
\newcommand{\mlet}[2]{\textcolor{pcolor}{\mathsf{let}} \ #1 \ #2}
\newcommand{\mletB}[2]{\textcolor{pcolor}{\mathsf{let!}} \ #1 \ #2}
\newcommand{\mif}[3]{\textcolor{pcolor}{\mathsf{if}} \ (#1)\{#2\}\{#3\}}
\else 
\newcommand{\zero}{{\mathbf 0}}
\newcommand{\fwd}[2]{#1 \leftrightarrow #2} 
\newcommand{\mix}[2]{#1  \parallel   #2}
\newcommand{\cut}[3]{#2  \parallel_{#1}  #3}
\newcommand{\cutB}[4]{#2.#3   \parallel_{!#1} #4}
\newcommand{\share}[3]{{#1.\mathsf{share}}[#2\;|\;#3]}
\newcommand{\nsum}[2]{#1  + #2}
\newcommand{\pone}[1]{\dual{#1}.\mathsf{close}}
\newcommand{\pbot}[2]{#1.\mathsf{close};#2} 
\newcommand{\pparl}[3]{#1(#2){;}#3}
\newcommand{\potimes}[4]{\dual{#1}(#2.#3);#4}
\newcommand{\fout}[3]{\dual{#1}\langle #2 \rangle{;} #3}
\newcommand{\pcase}[3]{#1.\mathsf{case}[#2, #3]} 
\newcommand{\cleft}[2]{\dual{#1}.\mathsf{inl};#2}
\newcommand{\cright}[2]{\dual{#1}.\mathsf{inr};#2} 
\newcommand{\pbang}[3]{!#1(#2.#3)} 
\newcommand{\pwhy}[2]{?#1;#2}
\newcommand{\pcopy}[3]{\dual{#1}?(#2);#3} 
\newcommand{\ncell}[3]{#1.\mathsf{cell}(#2. #3)}
\newcommand{\rd}[3]{#1.\mathsf{rd}(#2){;}#3}
\newcommand{\nwrt}[4]{#1.\mathsf{wr}(#2.#3){;}#4} 
\newcommand{\fwrt}[3]{#1.\mathsf{wr}\langle #2 \rangle; #3}
\newcommand{\fwdE}[2]{#1 \leftrightarrow_! #2}
\newcommand{\cellE}[3]{#1.\mathsf{cell}(\fwdE{#2}{#3})}
\newcommand{\pexists}[3]{\dual{#1}\langle #2 \rangle; #3}
\newcommand{\pforall}[3]{#1(#2);#3} 
\newcommand{\lock}[2]{#1.\mathsf{lock};#2} 
\newcommand{\unlock}[2]{#1.\mathsf{unlk};#2} 
\newcommand{\shareL}[3]{#1.\mathsf{shareL}[#2 \;|\;  #3]} 
\newcommand{\shareR}[3]{#1.\mathsf{shareR}[#2 \;|\;  #3]} 
\fi
\newcommand{\deff}{\triangleq}
\newcommand{\LC}[1]{{\color{red}#1}}
\newcommand{\PR}[1]{{\color{blue}#1}}
\newcommand{\tto}[0]{\Rightarrow}
\newcommand{\sep}[0]{\;|\;}
\newcommand{\dual}[1]{\overline{#1}}
\newcommand{\subs}[3]{\{#1/#2\}#3} 
\newcommand{\texists}[2]{\exists #1. #2}
\newcommand{\tforall}[2]{\forall #1.#2} 
\newcommand{\with}{\mathbin{\binampersand}}
\colorlet{dgreen}{green!40!black}
\newcommand{\cboxwhy}[1]{{\textcolor{dgreen}{\boxwhy}} #1} 
\newcommand{\clboxwhy}[1]{{\textcolor{red}{\lboxwhy}} #1} 
\newcommand{\cboxbang}[1]{{\textcolor{dgreen}{\boxbang}} #1}
\newcommand{\clboxbang}[1]{{\textcolor{red}{\lboxbang}} #1}
\newcommand{\lboxbang}{\makebox[0pt][l]{$\boxslash$}\raisebox{0ex}{\hspace{0.27em}{\scriptsize $!$}}\;}
\newcommand{\lboxwhy} {\makebox[0pt][l]{$\boxslash$}\raisebox{0ex}{\hspace{0.21em}{\scriptsize $?$}}\;}
\newcommand{\boxbang}{\makebox[0pt][l]{$\boxempty$}\raisebox{0ex}{\hspace{0.27em}{\scriptsize $!$}}\;}
\newcommand{\boxwhy} {\makebox[0pt][l]{$\boxempty$}\raisebox{0ex}{\hspace{0.21em}{\scriptsize $?$}}\;}
\newcommand{\inter}[3]{ \mathcal I_{#1}(#2, #3)}
\newcommand{\ptypes}[0]{PaT}
%
\lstset{	basicstyle=\linespread{0.5}\sffamily\small,
columns=fullflexible,
    literate=*{?}{{\textcolor{pcolor}{?}}}{1}
    		 {!}{{\textcolor{pcolor}{!}}}{1}
    		{~}{{\fontfamily{ptm}\selectfont \textcolor{pcolor}{\textasciitilde} }}1,
  }
\lstdefinelanguage{CLASS} {
morekeywords={[1] type}, sensitive=true,
morekeywords ={[2] proc}, sensitive = true, 
morekeywords={[3] close, wait, case, of, offer, choice, recv, send, state, statel, usage, usagel, affine, coaffine, call,
take, put, release, lint, colint, let, print, println, 
sendty, recvty, cell, usage, free, rd, wrt, lock, unlk, par, cut, fwd, share}, sensitive=false,
morekeywords={[4] \~{}}, sensitive=true,
keywordstyle={[1] \color{pcolor}\bfseries},
keywordstyle={[2] \color{pcolor}\bfseries},
keywordstyle={[3] \color{pcolor}\bfseries},
keywordstyle={[4] \color{pcolor}},
morecomment=[l]{--},
morecomment=[s]{/*}{*/},
morestring=[b]", }
\lstset{language=CLASS}
%
\newcommand{\nf}[1]{\llbracket #1 \rrbracket} 
\newcommand{\nfwhy}[1]{\llbracket #1 \rrbracket_?}
\newcommand{\obs}[2]{#1\downarrow_{#2}}
\newcommand{\nshare}[2]{\text{nshare}(#1, #2)}
\newcommand{\offers}{\vartriangleright} 
\newcommand{\linctxt}[1]{\lpred{#1}}
\newcommand{\unrctxt}[1]{{\lpred{#1}^!}}
\newcommand{\true}{\text{true}}
\newcommand{\false}{\text{false}}
\newcommand{\cuteval}[1]{\lto{#1}} 
\newcommand{\ndexp}[1]{\lto{#1}_\epsilon}
\newcommand{\shexp}{\to_\epsilon}
\newcommand{\dzero}{T\zero}
\newcommand{\dfwd}[2]{T\text{fwd} \ #1#2}
\newcommand{\dmix}[2]{T\text{mix} \ #1 #2}
\newcommand{\dcut}[3]{T\text{cut} \ (#1)#2#3}
\newcommand{\dshare}[3]{Tshare \ #1#2#3}
\newcommand{\dsum}[2]{T\text{sum} \ #1#2} 
\newcommand{\dcutB}[4]{Tcut! \ #1(#2)#3#4}
\newcommand{\dcopy}[3]{Tcopy \ #1(#2)#3}
\newcommand{\done}[1]{T\one \ #1}
\newcommand{\dbot}[2]{T\bot \ #1#2}
\newcommand{\dotimes}[4]{T\otimes \ #1(#2)#3#4}
\newcommand{\dparl}[3]{T\parl \ #1(#2)#3}
\newcommand{\dbang}[3]{T! \ #1(#2)#3}
\newcommand{\dwhy}[2]{T? \ #1#2}
\newcommand{\dcell}[3]{Tcell \ #1(#2)(#3)}
\newcommand{\dfree}[2]{Tfree \ #1#2} 
\newcommand{\drd}[3]{Tread \ #1(#2)#3} 
\newcommand{\dwrt}[4]{Twrt \ #1(#2)#3#4}
\newcommand{\doplusl}[2]{T\oplus_l \ #1#2}
\newcommand{\doplusr}[2]{T\oplus_r \ #1#2}
\newcommand{\dand}[3]{T\& \ #1#2#3}
\newcommand{\tot}{\xrightarrow[]{*}}
\newcommand{\toc}{\to_{\mathsf{c}}}
\newcommand{\toct}{\xrightarrow[]{*}_{\mathsf{c}}}
\newcommand{\equivc}{\equiv_{\mathsf{c}}}
\newcommand{\toS}{\mathrel{\mathrlap{\rightarrow}\mkern1mu\rightarrow}}
\newcommand{\toSt}{\mathrel{\mathrlap{\xrightarrow[]{*}}\mkern1mu\rightarrow}}
\newcommand{\fto}{\rightarrowdbl}
\newcommand{\ftoS}{\xrightarrowdbl[]{*}}
\newcommand{\lto}[1]{\xrightarrow[]{#1}}
\newcommand{\llto}[1]{\xRightarrow[]{#1}}
\newcommand{\ltos}[1]{\xrightarrow[s]{#1}}
\newcommand{\ltoS}[1]{\xRightarrow[]{#1}}
\newcommand{\lpred}[1]{\llbracket #1 \rrbracket} 
\newcommand{\lpreds}[1]{\lpred{#1}_s}
\newcommand{\lpredv}[1]{\lpred{#1}_{\to}}
\newcommand{\primeD}[1]{\mathcal P_r(#1)} 
\newcommand{\interE}[5]{\mathcal I^{\rightleftarrows}(#1,#2,#3,#4,#5)}
\newcommand{\rightarrowdbl}{\rightarrow\mathrel{\mkern-14mu}\rightarrow}
\newcommand{\xrightarrowdbl}[2][]{%
  \xrightarrow[#1]{#2}\mathrel{\mkern-14mu}\rightarrow
}
\newcommand{\one}{{\mathbf1}}
\newcommand{\parl}{\mathbin{\bindnasrepma}}
\newcommand{\aotimes}[3]{\dual{#1}(#2.#3)} 
\newcommand{\aclose}[1]{#1.\text{close}}
\newcommand{\ain}[2]{#1\langle #2 \rangle}
\newcommand{\arin}[2]{#1!\langle #2 \rangle}
\newcommand{\awhy}[1]{?#1}
\newcommand{\afree}[1]{#1.\text{free}}
\newcommand{\ard}[2]{#1.\text{read}\langle #2 \rangle}
\newcommand{\awrt}[2]{#1.\text{write}\langle #2 \rangle}
\newcommand{\aleft}[1]{#1.\text{inl}} 
\newcommand{\aright}[1]{#1.\text{inr}} 
%
\newcommand{\predclose}{\mathcal{P}_\text{close}}
\newcommand{\predcom}[2]{\mathcal{P}_\text{com}(#1, #2)} 
\newcommand{\predcomb}[3]{\mathcal{P}_\text{com!}(#1, #2, #3)}
\newcommand{\predfree}{\mathcal{P}_\text{free}}
\newcommand{\predread}[3]{\mathcal{P}_\text{read}(#1, #2, #3)}
\newcommand{\predwrite}[3]{\mathcal{P}_\text{write}(#1, #2, #3)}
\newcommand{\enlto}[1]{\xrightarrow[+]{#1}} 
\newenvironment{centermath}
 {\begin{center}$\displaystyle}
 {$\end{center}}
\newcommand{\muCLL}{$\mu\mathsf{CLL}$}
\newcommand{\piCLL}{$\mu\mathsf{CLL}$}
\newcommand{\CLL}{\mathsf{CLL}}
\newcommand{\LSS}{$\pi$SS} 
\newcommand{\LSSm}{$\pi$SS$_{\exists, \forall}$} 
\newcommand{\LLocks}{$\pi$SSL}
\newcommand{\emp}[1]{#1.\mathsf{emp}}
\newcommand{\acq}[3]{#1.\mathsf{acq}(#2);#3} 
\newcommand{\rel}[4]{#1.\mathsf{rel}(#2. #3);#4} 
\newcommand{\frel}[3]{#1.\mathsf{rel}\langle #2 \rangle;#3}
\newcommand{\res}[3]{#1.\mathsf{res}(#2.#3)} 
\newcommand{\fres}[2]{#1.\mathsf{res}\langle #2 \rangle} 
\newcommand{\nres}[5]{#1.\mathsf{res}(#2. #3) \mapsto #4.#5} 
\newcommand{\nemp}[3]{#1.\mathsf{empty} \mapsto #2.#3} 
\newcommand{\cboxwhyf}[1]{{\textcolor{dgreen}{\boxwhy_f}} #1} 
\newcommand{\clboxwhyf}[1]{{\textcolor{red}{\lboxwhy_f}} #1} 
\newcommand{\cboxbangf}[1]{{\textcolor{dgreen}{\boxbang_f}} #1}
\newcommand{\clboxbangf}[1]{{\textcolor{red}{\lboxbang_f}} #1}
%
\newcommand{\minus}{\scalebox{0.75}[1.0]{$-$}}
%
%
\newcommand{\fn}[1]{\mathsf{fn}(#1)}
\newcommand{\fv}[1]{\mathsf{fv}(#1)}
%
\newcommand{\CLASS}{$\mathsf{CLASS}$} 
\newcommand{\CLASSs}{ \CLASS{}$\setminus\exists \mu$}
\newcommand{\CLASSc}{$\mathsf{CLASS}_{\mathsf{c}}$} 
\newcommand{\CLLSj}{$\mathsf{CLLSj}$} 
%
\newcommand{\piSSL}{$\pi\mathsf{SSL}$} 
%
\newcommand{\aff}[1]{\wedge #1}
\newcommand{\coaff}[1]{\vee #1}
\newcommand{\cstateE}[1]{\textcolor{red}{\mathsf{S}_\circ} #1}
\newcommand{\cstateF}[1]{\textcolor{dark-blue}{\mathsf{S}_\bullet} #1}
\newcommand{\cusageE}[1]{\textcolor{red}{\mathsf{U}_\circ}  #1}
\newcommand{\cusageF}[1]{\textcolor{dark-blue}{\mathsf{U}_\bullet} #1} 
\newcommand{\stateP}[2]{\textbf{S}_#1 \ #2}
\newcommand{\usageP}[2]{\textbf{U}_#1 \ #2}
\newcommand{\emptyf}{\textcolor{red}{e}} 
\newcommand{\fullf}{\textcolor{dark-blue}{f}} 
%
\newcommand{\affine}[2]{\textcolor{pcolor}{\mathsf{affine}} \ #1; #2} 
\newcommand{\daffine}[4]{\textcolor{pcolor}{\mathsf{affine}}_{#3, #4} \ #1; #2} 
\newcommand{\discard}[1]{\textcolor{pcolor}{\mathsf{discard}} \ #1} 
\newcommand{\use}[2]{\textcolor{pcolor}{\mathsf{use}} \ #1; #2} 
\newcommand{\cell}[3]{\textcolor{pcolor}{\mathsf{cell}} \ #1(#2.#3)} 
\newcommand{\fcell}[2]{\textcolor{pcolor}{\mathsf{cell}} \ #1(#2)} 
\newcommand{\mempty}[1]{\textcolor{pcolor}{\mathsf{empty}} \ #1} 
\newcommand{\free}[1]{\textcolor{pcolor}{\mathsf{release}} \ #1} 
\newcommand{\take}[3]{\textcolor{pcolor}{\mathsf{take}} \ #1(#2);#3} 
\newcommand{\mput}[4]{\textcolor{pcolor}{\mathsf{put}} \ #1(#2.#3);#4}
\newcommand{\fput}[3]{\textcolor{pcolor}{\mathsf{put}} \ #1(#2);#3}
%
%
\newcommand{\ind}{\textcolor{pcolor}{\mathsf{rec}}} 
\newcommand{\type}{\textcolor{pcolor}{\mathsf{type}}} 
\newcommand{\coind}{\textcolor{pcolor}{\mathsf{corec}}} 
\newcommand{\trec}[2]{\mu #1. \ #2} 
\newcommand{\tcorec}[2]{\nu #1. \ #2} 
\newcommand{\unfold}[2]{\textcolor{pcolor}{\mathsf{unfold}_\mu} \ #1; #2} 
\newcommand{\unfoldn}[2]{\textcolor{pcolor}{\mathsf{unfold}_\nu} \ #1; #2} 
\newcommand{\ncorec}[5]{\textcolor{pcolor}{\mathsf{corec}} \ #2(#1, #3) ; #5 \ [#4]} 
\newcommand{\corec}[4]{\textcolor{pcolor}{\mathsf{corec}} \ #2(#1, #3); #4} 
\newcommand{\mloops}[3]{\textcolor{pcolor}{\mathsf{loop}} \ #2(#1); #3} 
\newcommand{\rvar}[2]{#1(#2)} 
\newcommand{\println}[2]{\textcolor{pcolor}{\mathsf{println}} \ #1; #2}
\newcommand{\enc}[1]{\llbracket #1 \rrbracket}
%
%
\newcommand{\impl}[1]{\emph{\text{#1}}}
%
\newcommand{\Lcell}[3]{\textcolor{pcolor}{\mathsf{cell}} \ #1(#2.#3)} 
\newcommand{\Lempty}[3]{\textcolor{pcolor}{\mathsf{empty}} \ #1 (#2.#3)} 
\newcommand{\lpredE}[1]{\mathcal L \llbracket #1 \rrbracket}
\newcommand{\cred}[1]{R[#1]}
\newcommand{\srec}[5]{#1\!\left<#2,#3,#4\right>\!#5} 
\newcommand{\srecs}[4]{#1\!\left<#2,#3\right>\!#4} 
\newcommand{\srecss}[4]{#1\!\left<#2\right>^s\!#4} 
\newcommand{\sreccs}[4]{#1\!\left<#2\right>^s\!#3} 
\newcommand{\ou}[1]{\overline{#1}} 

\newcommand{\vdashB}[0]{\vdash^{\mathsf{B}}}
\newcommand{\tob}[0]{\Rightarrow^\mathsf{B}}
\newcommand{\tor}[0]{\to^\mathsf{Br}}
\newcommand{\tobp}[0]{\Rightarrow^\mathsf{B}_p}
\newcommand{\tobn}[0]{{\stackrel{}{\to}{\!\!\!\;}^{\mathsf{B}n}}\;}
\newcommand{\toba}[0]{{\stackrel{}{\to}{\!\!\!\;}^{\mathsf{B}a}}\;}
\newcommand{\tobaxm}[1]{{\stackrel{#1}{\Rightarrow}{\!\!\!\;}^{\mathsf{B}a}}\;}
\newcommand{\tobax}[1]{{\stackrel{#1}{\to}{\!\!\!\;}^{\mathsf{B}a}}\;}
\newcommand{\tobanm}[0]{{\stackrel{}{\Rightarrow}{\!\!\!\;}^{\mathsf{Bap}}}\;}
\newcommand{\toban}[0]{{\stackrel{}{\to}{\!\!\!\;}^{\mathsf{Bap}}}\;}
\newcommand{\tobb}[0]{{\to^\mathsf{B}}}
\newcommand{\tobz}[1]{{\stackrel{#1}{\to}{\!\!\!\;}^{\mathsf{B}}\;}}
\newcommand{\tobx}[1]{{\stackrel{#1}{\to}{\!\!\!\;}^{\mathsf{B}p}\;}}
\newcommand{\tobnx}[1]{{\stackrel{#1}{\to}{\!\!\!\;}^{\mathsf{B}n}\;}}
\newcommand{\tobxz}[1]{{\stackrel{#1}{\Rightarrow}{\!\!\!\;}^{\mathsf{B}p}}}
\newcommand{\clos}[0]{\textcolor{blue}{\mathsf{clos}}}
\newcommand{\closB}[0]{\textcolor{blue}{\mathsf{clos!}}}
\newcommand{\checkm}[0]{\textcolor{blue}{\checkmark}}
\newcommand{\nilm}[0]{\textcolor{blue}{\mathsf{nil}}}

\newcommand{\mypara}[1]{\noindent\textbf{#1.}}
%
\title{The Session Abstract Machine \\ (Extended Version)
  \vspace{-10pt}}
%
%

\author{Luís Caires\inst{1} \and Bernardo Toninho\inst{2}}
\institute{Técnico Lisboa / INESC-ID 
  \quad \email{luis.caires@tecnico.ulisboa.pt} \and NOVA FCT / NOVA
  LINCS \quad \email{btoninho@fct.unl.pt}}

\authorrunning{L. Caires, B. Toninho}

\maketitle              
\begin{quote}\small
  {\bf Abstract.}
  We build on a fine-grained analysis of session-based interaction as
  provided by the linear logic typing disciplines to introduce the
  SAM, an abstract machine for mechanically executing
  session-typed processes.   A remarkable feature of
  the SAM's design is its ability to naturally segregate and coordinate
  sequential with concurrent session behaviours.
  In particular, implicitly sequential parts of session programs may
  be efficiently executed by deterministic sequential application of SAM
  transitions, amenable to compilation, and without concurrent synchronisation
  mechanisms.
  We provide an intuitive discussion of the SAM structure and its
  underlying design, and state and prove its correctness for executing
  programs in a session calculus corresponding to full classical
  linear logic $\CLL$. We also discuss extensions and applications of
  the SAM to the execution of linear and session-based programming
  languages.
\end{quote}


\section{Introduction}
\vspace{-4pt}

In this work, we build on the linear logic based foundation for
session
types~\cite{DBLP:conf/concur/CairesP10,cairesmscs16,wadler2014propositions}
to construct SAM, an abstract machine specially designed 
for executing session processes typed by (classical) linear logic $\CLL$.   
Although motivated by the session type discipline, which originally emerged in the realm of concurrency and distribution~\cite{DBLP:conf/concur/Honda93,DBLP:conf/esop/HondaVK98,DBLP:journals/acta/GayH05,DBLP:journals/csur/HuttelLVCCDMPRT16}, a basic motivation for designing
the SAM was to provide an efficient deterministic execution model for the implicitly
sequential session-typed program idioms that often proliferate  
in concurrent session-based programming.
It is well-known that in a world of fine-grained concurrency, building on many process-based encodings of concepts such as (abstract) data types, functions, continuations, and effects \cite{DBLP:books/daglib/0098267,DBLP:journals/jfp/Vasconcelos05,Toninho12fossacs,DBLP:conf/esop/ToninhoCP13,DBLP:conf/esop/CairesP17,DBLP:conf/esop/ToninhoY18,DBLP:conf/coordination/PfenningP23}, large parts of the code turn out to be inherently sequential, further justifying the foundational and practical relevance of our results. A remarkable feature of the
SAM's design is therefore its
potential to efficiently coordinate sequential with full-fledged
concurrent behaviours in session-based programming.

Leveraging early work relating linear logic with the semantics of
linear and concurrent
computation~\cite{DBLP:journals/tcs/Abramsky93,bellin.scott:linear-logic,abramsky1996interaction},
the proposition-as-types (PaT)
interpretation~\cite{DBLP:journals/cacm/Wadler15} of linear logic
proofs as a form of well-behaved session-typed nominal calculus has
motivated many developments since its inception
~\cite{DBLP:conf/esop/CairesPPT13,balzer2017manifest,
  DBLP:conf/esop/ToninhoY18,DBLP:conf/ppdp/ToninhoCP21}. We believe
that, much how the $\lambda$-calculus is deemed a canonical typed
model for functional (sequential) computation with pure values, the
session calculus can be accepted as a fairly canonical typed model for
stateful concurrent computation with linear resources, well-rooted in
the trunk of ``classical'' Type Theory.  The PaT interpretation of
session processes also establishes a bridge between more classical
theories of computation and process algebra via logic. It also
reinstates Robin Milner's view of computation as
interaction~\cite{milner1993elements},
``data-as-processes''~\cite{DBLP:books/daglib/0098267} and
``functions-as-processes''~\cite{DBLP:journals/mscs/Milner92}, now in
the setting of a tightly typed world, based on linear logic, where
types may statically ensure key properties like deadlock-freedom,
termination, and correct resource usage in stateful
programs.  Session calculi are
motivating novel programming language design, bringing up new insights
on typeful programming~\cite{Cardelli91} with linear and behavioral
types, e.g.,~\cite{DBLP:journals/lmcs/DasP22,DBLP:conf/esop/RochaC23,DBLP:conf/ecoop/ChenBT22,balzer2017manifest}.
Most systems of typed session calculi have been formulated in process
algebraic form~\cite{DBLP:conf/concur/Honda93,DBLP:conf/esop/HondaVK98,DBLP:journals/acta/GayH05}, or on top of concurrent
$\lambda$-calculi with an extra layer of communication channels
(e.g.,~\cite{gay2010linear}), logically inspired systems such as the
those discussed in this
paper (e.g., ~\cite{DBLP:conf/concur/CairesP10,cairesmscs16,wadler2014propositions,DBLP:conf/fossacs/DardhaG18,DBLP:journals/pacmpl/KokkeMP19,qian2020client,DBLP:journals/pacmpl/FruminDHP22,DBLP:conf/esop/RochaC23})
are defined by a logical proof / type system where proof rules are
seen as witnesses for the typing of process terms, proofs are read as
processes, structural equivalence is proof conversion and computation
corresponds to cut reduction.  These formulations provide a
fundamental semantic foundation to study the model's expressiveness
and meta-theory, but of course do not directly support the concrete
implementation of programming languages based on them.

Although several programming language implementations of
nominal calculi based languages have been proposed for some time
(e.g.~\cite{DBLP:conf/birthday/PierceT00}), with some introducing
abstract machines as the underlying technology
(e.g.,~\cite{DBLP:phd/ethos/Turner96,DBLP:conf/ppdp/LopesSV99}), we
are not aware of any prior design proposal for an abstract machine for
reducing session processes exploiting deep properties of a source
session calculus, as e.g., the
CAM~\cite{DBLP:journals/scp/CousineauCM87}  the
LAM~\cite{DBLP:journals/tcs/Lafont88}, or the
KM~\cite{DBLP:journals/lisp/Krivine07}, which also explore the
Curry-Howard correspondences, may reclaim to be, respectively for
call-by-value cartesian-closed structures, linear logic, and the
call-by-name $\lambda$-calculus.

The SAM reduction strategy explores a form of ``asynchronous''
interaction that essentially expresses that, for processes typed by
the logical discipline, sessions are always pairwise causally
independent, in the sense that immediate communication on some session
is never blocked by communication on a different session. This
property is captured syntactically by prefix commutation equations,
valid commuting conversions in the underlying logic: adding equations
for such laws explicitly to process structural congruence keeps
observational equivalence of $\CLL$ processes untouched~\cite{Perez12esop}.
Combined with insights related to focalisation and polarisation in
linear logic~\cite{DBLP:journals/logcom/Andreoli92,DBLP:conf/fossacs/PfenningG15,DBLP:conf/tlca/Laurent99}, we realize that all communication in any
session may be operationally structured as the exchange
of bundles of positive actions from sender to receiver,
where the roles sender/receiver flip whenever the session type swaps
polarity. Communication may then be mediated by message buffers, first
filled up by the sender (``write-biased'' scheduling), and at a later
time emptied by the receiver.
Building on these observations and on key properties of linear logic
proofs leveraged in well-known purely structural proofs of
progress~\cite{DBLP:conf/concur/CairesP10,cairesmscs16,DBLP:conf/esop/RochaC23}, we identify a sequential and deterministic
reduction strategy for $\CLL$ typed processes, based on a form of
co-routining where continuations are associated to session queues, and
``context switching'' occurs whenever polarity flips.  That such
strategy works at all, preserving all the required correctness
properties of the $\CLL$ language does not seem immediately obvious,
given that each processes may sequentially perform multiple actions on
many different sessions, meaning that multiple context switches must
be interleaved. The bulk of our paper is then devoted to establishing
all such properties in a precise technical sense.
We believe that the SAM may provide a principled foundation
for safe execution environments for programming
languages combining functional, imperative and concurrent idioms based on session and linear types, as witnessed in practice
for Rust~\cite{rust-lang}, (Linear) Haskell 
\cite{DBLP:conf/haskell/LindleyM16,DBLP:journals/pacmpl/BernardyBNJS18,DBLP:conf/haskell/KokkeD21}, Move~\cite{move-lang},
and in research languages~\cite{rocha2021propositions,DBLP:journals/pacmpl/JacobsB23,DBLP:journals/lmcs/DasP22}. To
further substantiate these views we have developed an
implementation of the SAM, integrated in a language for realistic session-based
shared-state programs~\cite{artifactesop24}.


\mypara{Outline and Contributions}
In Section~\ref{sec:language} we briefly review the session-typed
calculus $\CLL$, which exactly corresponds to (classical) Linear Logic
with mix. In Section~\ref{sec:coreSAM} we discuss the motivation and
design principles of the core SAM, gradually presenting its structure
for the language fragment corresponding to session types without the
exponentials, which will be introduced later.  Even if the core SAM
structure and transition rules are fairly simple, the proofs of
correctness are more technically involved, and require progressive
build up. Therefore, we first bridge between $\CLL$ and SAM via a
intermediate logical language $\Blang$, introducing explicit queues in
cuts, presented in Section~\ref{sec:blang}. We show preservation
(Theorem~\ref{theorem:type-preservation}) and progress
(Theorem~\ref{theorem:progress}) for $\Blang$, and prove that there is
two way simulation between $\Blang$ and $\CLL$ via a strong
operational correspondence (Theorem~\ref{teo:cll-b}). Given this
correspondence, in Section~\ref{sec:soundcore} we state and prove the
adequacy of the SAM for executing $\CLL$ processes, showing soundness
wrt. $\Blang$ (Theorem~\ref{soundcllb}) and $\CLL$
(Theorem~\ref{soundcll}), and progress / deadlock absense
(Theorem~\ref{samprogress}).  In Section~\ref{sec:sam-mix} modularly
extend the previous results to the exponentials and mix, and revise
the core SAM by introducing explicit environments, stating the
associated adequacy results (Theorem~\ref{fulsamsound} and
Theorem~\ref{fullsamprogress}). We also discuss how to accommodate
concurrency, and other extensions in the SAM. We conclude by a
discussion of related work and additional remarks.
Appendix~\ref{appendix} includes additional lemmas and proofs.



\vspace{-10pt}
\section{Background on $\CLL$, the core language and type system }\label{sec:language}

We start by revisiting the language and type system of $\CLL{}$, and its operational semantics. The system is based on a PaT interpretation of classical linear logic (we follow the presentations of~\cite{cairesmscs16,caires2017linearity,rocha2021propositions}).
\begin{definition}[Types] Types $A,B$ are defined by 
\vspace{6pt}

	$
	\begin{array}{lllllllllllll}
	A, B ::= 
	 &\one  &\sep \bot    &\sep  A\parl B   & \sep A \otimes B &\sep 
	 \withm{\ell}{L}{A_\ell}  &\sep  \oplusm{\ell}{L}{A_\ell}   &\sep !A & \sep  ?A \\
	\end{array}
	$
\end{definition}

Types  
comprise of the units ($\one$, $\bot$), multiplicatives ($\otimes$,
$\parl$), additives ($\oplusm{\ell}{L}{A_\ell}$,
$\withm{\ell}{L}{A_\ell}$) and  exponentials ($!$, $?$).  
We adopt here a labeled version of the additives, where 
the linear logic sum type $A_{\labl{inl}}\oplus A_{\labl{inr}}$ is defined by $\oplusm{\ell}{
\{\labl{inl},\labl{inr}\}}{A_\ell}$.
The \emph{positive types} are $\one$, $\otimes$, $\oplus$, and $!$,
while the \emph{negative types} are $\bot$, $\parl$, $\with$ and $?$.
We abbreviate $\dual A\parl B$ by $A\multimap B$.
We write $A^{+}$ (resp. $A^{-}$) to assert that $A$ is a positive (resp. negative) type.
%
Type \emph{duality} $\dual{A}$ corresponds to negation: 
\[
\begin{array}{llllllllllllllll}
\dual{\one} &=&\bot  &&&\hspace{15pt} \dual{A \otimes B} &=&\dual{A} \parl \dual{B}   &&&\hspace{15pt} \dual{\oplusm{\ell}{L}{A_\ell}} &=& \withm{\ell}{L}{\dual{A}_\ell} &\hspace{15pt}
\dual{!A} &=&?\dual{B}  
\end{array}
\]
Duality captures the symmetry of behaviour in binary process
interaction, as manifest in the cut rule.

\begin{definition}[Processes]
  The syntax of processes $P,Q$ is given by:
  $$
  \begin{array}{lcl}
  P, Q & ::= & \zero \mid \mix{P}{Q} \mid \fwd{x}{y} \mid
               \cuti{x{:}A}{P}{Q} \mid \pone{x} \mid \pbot{x}{P}\\
        &       \mid &
  \pcasem{x}{\ell}{L}{P_\ell}
 \mid \cright{x}{P} \mid
                       \potimes{x}{y}{P}{Q} \mid \pparl{x}{z}{P}\\
    &\mid & \pbang{x}{y}{P}
    \mid \pwhy{x}{P} \mid \cutBi{x:A}{y}{P}{Q} \mid \pcopy{x}{z}{Q}
    \end{array}
  $$

\end{definition}
\emph{Typing judgements} have the form $P \vdash \Delta; \Gamma$, where $P$ is a process and
the \emph{typing context} $\Delta;\Gamma$ is
dyadic~\cite{DBLP:journals/logcom/Andreoli92,Benton94amixed,DBLP:conf/lics/Pfenning95,DBLP:conf/concur/CairesP10}:
both $\Delta$ and $\Gamma$ assign types to names, the context $\Delta$
is handled linearly (no implicit contraction or weakening) while the
exponential context $\Gamma$ is unrestricted.
The type system exactly corresponds, via a propositions-as-types 
correspondence, to the canonical proof
system of Classical Linear Logic with Mix.
When a cut type annotation is easily inferred, we may omit it and
write $\cuti{x}{P}{Q}$.  
	The typing rules of $\CLL{}$ are given in
        Fig.~\ref{fig:typing-muCLL}.

        The process $\zero$ denotes the inactive process, typed in the empty
linear context (rule [T$\zero$]). $\mix{P}{Q}$ denotes independent
parallel composition of processes $P$ and $Q$ (rule [Tmix]),
whereas $\cuti{x{:}A}{P}{Q}$ denotes interfering parallel composition of $P$ and
$Q$, where $P$ and $Q$ share exactly one channel name $x$, typed as
$A$ in $P$ and $\dual{A}$ in $Q$ (rule [Tcut]).
The
construct $\fwd{x}{y}$ captures forwarding between dually typed names
$x$ and $y$ (rule [Tfwd]), which operationally consists in (globally) renaming
$x$ for $y$.

Processes $\pone{x}$ and $\pbot{x}{P}$ denote session termination and
the dual action of waiting for session termination, respectively
(rules [T$\one$] and [T$\bot$]).
The constructs $\pcasem{x}{\ell}{L}{P_\ell}$ and $\labl{l}\ x;P$
denote label input and output, respectively, where the input construct
pattern matches on the received label to select the
process continuation that is to run.
Process $\potimes{x}{y}{P_1}{P_2}$ and $\pparl{x}{z}{Q}$ codify
the output of (fresh) name $y$ on channel $x$ and the corresponding
input action, where the received name will be substituted for $z$ in
$Q$ (rules [T$\otimes$] and [T$\parl$]).
Typing ensures that the names used in $P_1$ and $P_2$ are disjoint.

Processes $\pbang{x}{y}{P}$, $\pwhy{x}{Q}$ and $\pcopy{x}{z}{Q}$
embody replicated servers and client processes. Process
$\pbang{x}{y}{P}$ consists of a process that waits for inputs on $x$,
spawning a replica of $P$ (depending on no linear sessions -- rule [T$!$]). Process $\pwhy{x}{Q}$ and $\pcopy{x}{z}{Q}$
allow for replicated servers to be activated and subsequently used as (fresh)
linear sessions (rules [T$?$] and [Tcall]).
Composition of exponentials is achieved by the $\cutBi{x:A}{y}{P}{Q}$
process, where $P$ cannot depend on linear sessions and so may be
safely replicated.

\begin{figure}[t]
{\small$$
\begin{array}{c}
\begin{prooftree}
	\infer[no rule]0{\labeltext{[T$\zero$]}{[Tzero]}}
	\infer1[\text{[T$\zero$]} ]{\zero \vdash \emptyset;\Gamma}
\end{prooftree}
\quad 
\begin{prooftree}
	\infer[no rule]0{\labeltext{[Tmix]}{[Tmix]}}
	\infer[no rule]1{P \vdash \Delta'; \Gamma \quad Q \vdash \Delta; \Gamma}
	\infer1[\text{[Tmix]}]{\mix{P}{Q} \vdash \Delta', \Delta; \Gamma} 
\end{prooftree} 
\\ 
\begin{prooftree}
	\infer[no rule]0{\labeltext{[Tfwd]}{[Tfwd]}}
\infer1[\text{[Tfwd]}]{\fwd{x}{y} \vdash x:\dual A, y:A; \Gamma}
\end{prooftree} 
\quad
\begin{prooftree}
	\infer[no rule]0{\labeltext{[Tcut]}{[Tcut]}}
	\infer[no rule]1{P \vdash \Delta', x:A; \Gamma \quad Q \vdash \Delta, x:\dual A; \Gamma} 
	\infer1{\cuti{x:A}{P}{Q} \vdash \Delta', \Delta; \Gamma} 
\end{prooftree} \text{[Tcut]}
\vspace{4pt}\\
\begin{prooftree}
	\infer[no rule]0{\;}
	\infer1[\text{[T$\one$]}]{\pone{x} \vdash x:\one; \Gamma} 
\end{prooftree} 
\quad 
\begin{prooftree}
	\infer[no rule]0{Q \vdash \Delta; \Gamma}
	\infer1[\text{[T$\bot$]}]{\pbot{x}{Q} \vdash \Delta, x:\bot; \Gamma} 
\end{prooftree} 
\vspace{4pt}\\
\begin{prooftree}
	\infer[no rule]0{P_\ell \vdash \Delta, x: A_\ell 
	;\Gamma \quad (\mathit{all}\ \ell\in L)}
	\infer1[\text{[T$\with$]}]{
	\pcasem{x}{\ell}{L}{P_\ell} \vdash \Delta, x: \withm{\ell}{L}{A_\ell}; \Gamma} 
\end{prooftree} 
\quad
\begin{prooftree}
	\infer[no rule]0{Q \vdash \Delta', x:A_{\labl{l}}; \Gamma\quad \labl{l}\in L} 
	\infer1[\text{[T$\oplus$]}]{\labl{l}\ x;Q \vdash \Delta', 
	x:\oplusm{\ell}{L}{A_\ell} ; \Gamma} 
\end{prooftree}
\hide{ 
\hspace{0.3cm}
\begin{prooftree}
	\infer[no rule]0{Q_2 \vdash \Delta', x:B; \Gamma}
	\infer1[\text{[T$\oplus_r$]}]{\cright{x}{Q_2} \vdash \Delta', x:A \oplus B; \Gamma} 
\end{prooftree}} 
\vspace{4pt} \\
\begin{prooftree}
	\infer[no rule]0{P_1 \vdash \Delta_1, y: A; \Gamma \quad P_2 \vdash \Delta_2, x:  B; \Gamma}
	\infer1[\text{[T$\otimes$]}]{\potimes{x}{y}{P_1}{P_2} \vdash \Delta_1, \Delta_2, x: A \otimes  B; \Gamma} 
\end{prooftree} 
\vspace{4pt}\\
\begin{prooftree}
	\infer[no rule]0{Q \vdash \Delta, z:A, x: B; \Gamma} 
	\infer1[\text{[T$\parl$]}]{\pparl{x}{z}{Q} \vdash \Delta, x:A \parl B; \Gamma} 
\end{prooftree}
\vspace{4pt}\\
\begin{prooftree}
	\infer[no rule]0{P \vdash y: A; \Gamma} 
	\infer1[\text{[T!]}]{\pbang{x}{y}{P} \vdash x:! A; \Gamma}
\end{prooftree} 
\quad 
\begin{prooftree}
	\infer[no rule]0{Q \vdash \Delta; \Gamma, x:A} 
	\infer1[\text{[T?]}]{\pwhy{x}{Q} \vdash
          \Delta, x:?A; \Gamma} 
\end{prooftree} 
\vspace{4pt}\\
\begin{prooftree}
	\infer[no rule]0{\labeltext{[Tcut!]}{[Tcut!]}}
	\infer[no rule]1{P \vdash y:A; \Gamma \quad  Q \vdash\Delta; \Gamma, x:\dual A} 
	\infer1[\text{[Tcut!]}]{\cutBi{x:A}{y}{P}{Q} \vdash \Delta; \Gamma} 
\end{prooftree} 
\quad 
\begin{prooftree}
	\infer[no rule]0{Q \vdash \Delta, z:A; \Gamma, x:A} 
	\infer1[\text{[Tcall]}]{\pcopy{x}{z}{Q} \vdash \Delta; \Gamma, x:A} 
\end{prooftree}
\end{array} 
$$}
\vspace{-10pt}
\caption{\label{fig:typing-muCLL}Typing Rules of $\CLL$.}
\end{figure}

%


We call {\em action} any process that is either a forwarder or realizes an  
introduction rule, and denote by $\mathcal{A}$ the set of all actions, by $\mathcal{A}(x)$ the set of action with subject $x$ (the subject of an action is the
channel name in which it interacts~\cite{DBLP:books/daglib/0098267}).
An action is deemed \emph{positive} (resp. \emph{negative}) if its
associated type is positive (resp. negative) in the sense of \emph{focusing}. The set of positive (resp. negative) actions is denoted by $\mathcal{A}^+$ (resp. $\mathcal{A}^-$). We sometimes use, e.g., $\mathcal A$ or 
$\mathcal{A}^+(x)$ to denote a process in the set.
%
The $\CLL$ operational semantics is given by a \emph{structural congruence} relation $\equiv$ that captures static identities on processes, corresponding to commuting conversions in the logic, and a \emph{reduction} relation $\to$ that captures  process interaction, and corresponds to cut-elimination steps.
 \begin{figure}[t]
{\small$$
\begin{array}{ll}
 \fwd{x}{y} \equiv \fwd{y}{x} & \text{[fwd]} 
 \vspace{6pt}\\
  \cuti{x:A}{P}{Q} \equiv \cuti{x:\dual A}{Q}{P}  & \text{[com]} 
\vspace{6pt}\\
\mix{P}{\zero} \equiv P \quad \mix{P}{Q} \equiv \mix{Q}{P}   \quad
 \mix{P}{(\mix{Q}{R})} \equiv \mix{(\mix{P}{Q})}{R}   \quad  \ & \text{[par]} 
\vspace{6pt}\\
\hide{
\mix{P}{(\mix{Q}{R})} \equiv \mix{(\mix{P}{Q})}{R} & \text{[MM]} 
\vspace{6pt}\\
}
\cuti{x}{P}{(\mix{Q}{R})} \equiv \mix{(\cuti{x}{P}{Q})}{R}& \text{[CM]} 
\vspace{6pt}\\
  \cuti{x}{P}{(\cuti{y}{Q}{R})} \equiv \cuti{y}{(\cuti{x}{P}{Q})}{R}& \text{[CC]}
\vspace{6pt}\\
 \cuti{z}{P}{(\cutBi{x}{y}{Q}{R})} \equiv \cutBi{x}{y}{Q}{(\cuti{z}{P}{R})}&  \text{[CC!]} 
\vspace{6pt}\\
 \cutBi{x}{y}{Q}{(\mix{P}{R})} \equiv \mix{P}{(\cutBi{x}{y}{Q}{R})}& \text{[C!M]} 
\vspace{6pt}\\
 \cutBi{x}{y}{P}{(\cutBi{z}{w}{Q}{R})} \equiv \cutBi{z}{w}{Q}{(\cutBi{x}{y}{P}{R})} & \text{[C!C!]} 
\vspace{6pt}\\
    \cutBi{x}{y}{P}{(\cut{*}{Q}{R})} \equiv \cut{*}{\cutBi{x}{y}{P}{Q}}{\cutBi{x}{y}{P}{R}}  & \text{[C!*]} 
 \vspace{8pt}\\
a(x);\cut{*}{Q}{R} \equiv a(x);{(\cut{*}{Q}{R})} & \text{[C$+$*]} 
 \vspace{6pt}\\
a_1(x);a_2(y);P \equiv  a_2(y);a_1(x);P
 & \text{[C$i$]} 
\hide{
\potimes{x}{y}{P}{Q} * R \equiv \potimes{x}{y}{P}{(Q * R)} & \text{[C$\otimes$*]} 
 \vspace{6pt}\\

\potimes{x}{y}{P}{\potimes{z}{w}{Q}{R}} \equiv  
\potimes{z}{w}{Q}{\potimes{x}{y}{P}{R}} & \text{[C$\otimes$]} 
}
\end{array}
$$}
{\small Provisos: in [CM] 
$x \in \fn{Q}$; in [CC] 
$x,y \in \fn{Q}$; in [CC!], [C!M] 
$x \notin \fn{P}$; in [C!C!], $x \notin \fn{Q}$ and $z \notin \fn{P}$.
In [C$i$], $x\neq y$ and $bn(a_1(x))\cap bn(a_2(y))=\emptyset$}
 \caption{\label{fig:equiv} Structural congruence $P \equiv Q$. }
\end{figure}
\begin{definition}[$P \equiv Q$]\label{def:equiv}
Structural congruence $\equiv$ is the least congruence on processes closed under $\alpha$-conversion and the $\equiv$-rules in Fig.~\ref{fig:equiv}. 
%
\end{definition}
The definition of $\equiv$ reflects expected static laws, along the lines of the structural congruences / conversions in \cite{DBLP:conf/concur/CairesP10,Wadler12icfp}. 
The binary operators forwarder, cut, and mix are commutative. The set of processes modulo $\equiv$ is a commutative monoid with operation the parallel composition $(\mix{-}{-})$ and identity given by inaction $\zero$ ([par]). 
Any static constructs commute, as expressed by the laws [CM]-[C!sC!]. The unrestricted cut distributes over all the static constructs  by law [C*], where 
$\cut{*}{-}{-}$ stands for either a mix, linear or unrestricted cut. 
The laws [C$+ *$] and [C$+$] denote sound proof equivalences in linear logic and  bring explicit the independence of 
linear actions (
 %
 noted $a(x)$), in different sessions $x$~\cite{Perez12esop}. These
 conversions are not required to obtain deadlock freedom. However,
 they are necessary for full cut elimination (e.g., see~\cite{Wadler12icfp}), and expose more redexes, thus
more non-determinism in the choice of possible reductions. Perhaps surprisingly, this extra flexibility is important to allow the deterministic sequential evaluation strategy for $\CLL$ programs adopted by the SAM to be expressed.

\begin{definition}[Reduction $\to$]\label{def:to} 
	Reduction $\to$ is defined by the rules of Fig.~\ref{fig:reduction}. 
\end{definition}

We denote by $\Rightarrow$ the reflexive-transitive closure of $\to$. Reduction includes the set of principal cut conversions, i.e. the 
 redexes for each pair of interacting constructs.
It is closed by structural congruence ([$\equiv$]),  in rule [cong] we consider that $\mathcal C$ is a static context, i.e. a process context in which the single hole is covered only by the static constructs mix or cut. 
The forwarding behaviour is implemented by name substitution [fwd]~\cite{TLDI12}. All the other reductions act on a principal cut between two dual actions, and eliminate it on behalf of cuts involving their subprocesses.
\begin{figure}[t]
{\small$$
\begin{array}{ll}
\labeltext{[fwd]}{[to-fwd]} 
\cuti{y}{\fwd{x}{y}}{P} \to \subs{x}{y}{P} &  \text{[fwd]}
\vspace{6pt}\\
\labeltext{[$\one\bot$]}{[to-close]} 
\cuti{x}{\pone{x}}{\pbot{x}{P}} \to P  & \text{[$\one\bot$]} 
\vspace{6pt}\\
\labeltext{[$\otimes \parl$]}{[to-send]} 
\cuti{x}{\potimes{x}{y}{P}{Q}}{\pparl{x}{z}{R}} \to \cut{x}{Q}{(\cut{y}{P}{\subs{y}{z}{R}})}  & \text{[$\otimes \parl$]} 
\vspace{6pt}\\
\labeltext{[$\with \oplus_l$]}{[to-left]} 
\cuti{x}{\pcasem{x}{\ell}{L}{P_{\labl{\ell}}}}{\labl{l}\ {x};{R}} \to \cuti{x}{P_{\labl{l}}}{R} \ & \text{[$\with \oplus_l$]} 
\vspace{8pt}\\
\hide{\labeltext{[$\with \oplus_r$]}{[to-right]} 
\cuti{x}{\pcase{x}{P}{Q}}{\cright{x}{R}} \to \cuti{x}{Q}{R}  & \text{[$\with \oplus_r$]}
\vspace{6pt}\\}
\labeltext{[!?]}{[to-act]} 
\cuti{x}{\pbang{x}{y}{P}}{\pwhy{x}{Q}} \to \cutBi{x}{y}{P}{Q} & \text{[!?]} \vspace{6pt}\\
\labeltext{[call]}{[to-call]} 
\cutBi{x}{y}{P}{\pcopy{x}{z}{Q}} \to \cuti{z}{\subs{z}{y}{P}}{(\cutBi{x}{y}{P}{Q})} & \text{[call]}
\end{array}
\vspace{-10pt}
$$}
 \caption{\label{fig:reduction} Reduction $P \to Q$. }
\end{figure}
\hide{Reduction rules for the basic session constructs that interpret
Second Order Linear Logic and recursion are the expected ones~\cite{DBLP:conf/concur/CairesP10,cairesmscs16,wadler2014propositions}, along predictable
lines. }
$\CLL$ satisfies basic safety properties~\cite{DBLP:conf/concur/CairesP10} listed below, and
also confluence, and termination~\cite{rocha2021propositions,DBLP:conf/esop/RochaC23}. In particular we have:
\begin{theorem}[Type Preservation]\label{theorem:type-preservationCLL} 
	Let $P \vdash \Delta; \Gamma$. (1) If $P \equiv Q$, then $Q \vdash \Delta; \Gamma$. (2) If $P \to Q$, then $Q \vdash \Delta; \Gamma$. 
\end{theorem}
A process $P$ is \emph{live} if and only if $P = \mathcal C[Q]$, for some static context $\mathcal C$ (the hole lies within the scope of static constructs mix and cut) and $Q$ is an active process (a process with a topmost action prefix).
\begin{theorem}[Progress]\label{theorem:progressCLL}
	Let $P \vdash \emptyset; \emptyset$ be live. Then $P\to Q$ for some $Q$.
\end{theorem} 


\hide{
We comment the rules concerning affinity.
The interaction between an affine session and an use operation is defined by reduction rule \ref{[to-use]}, where a cut on $a:\aff{A}$ between $\daffine{a}{P}{\vec{b}}{\vec{c}}$ and $\use{a}{Q}$ reduces to a cut on $a:A$ between the continuations $P$ and $Q$. 
The reduction between an affine session and a discard operation is defined
by \ref{[to-discard]}. A cut between $\daffine{a}{P}{\vec{b}}{\vec{c}}$ and $\discard{a}$ reduces to a mix-composition of discards (for the coaffine sessions $\vec{b}$) and releases (for the cell usages $\vec{c}$) cf.~\cite{asperti2002intuitionistic, caires2017linearity}). 
In the corner case
where $\vec{c}$ and $\vec{a}$ are empty, the left-hand side of \ref{[to-discard]} simply degenerates to inaction $\zero$ (the identity of mix).

The reductions for the mutable state operations are fairly self-explanatory. 
In rule \ref{[to-release]}, a cut between a full mutex cell cell and a release operation reduces to 
a process that discards the affine cell contents, 
cf. rule \ref{[to-discard]}. 
In rule \ref{[to-take]}, a cut on $c:\cstateF A$ between a full cell and a take operation reduces to a process with two cuts, both composed with the continuation $\subs{a}{a'}{Q}$ of the take. The outer cut on $a:\aff{A}$ composes with the stored affine session, which was successfully acquired by the take operation. The inner cut on $c:\cstateE A$ composes with the reference cell $c$, which has became empty in the reductum. 
Finally, in rule \ref{[to-put]}, a cut on session $c:\cstateE A$ between an empty cell and a put operation reduces to a cut on session $c:\cstateF A$ between a full cell, that now stores the session that was put, and the continuation of the put process.
Notice that the locking/unlocking behaviour of cells is simply modelled by rewriting of the process terms, from cell to empty and back,
 as typical in process calculi.
\hide{
For example, rule [$\otimes \parl$] applies to a cut on session $x:A \otimes B$ between send and receive and reduces to a process expression with two cuts. The inner cut on $y:A$ connects the continuation $\subs{y}{z}{R}$ of the receiver with the provider $P$ of the sent channel, whereas the outer cut on $x:B$ connects $\subs{y}{z}{R}$ with the continuation $Q$ of the send process.}
}

\hide{
\begin{figure}
$$
\begin{array}{c}
\begin{prooftree}
	\infer[no rule]0{\labeltext{[Tcorec]}{[Tcorec]}}
	\infer[no rule]1{P \vdash_{\eta'} \Delta, z:A; \Gamma \quad \eta' = \eta, \rvar{X}{z,\vec{w}} \mapsto \Delta, z:Y; \Gamma}
	\infer1[\text{[Tcorec]}]{\ncorec{z}{X}{\vec{w}}{x,\vec{y}}{P} \vdash_\eta \subs{\vec{y}}{\vec{w}}{\Delta}, x: \tcorec{Y}{A}; \subs{\vec{y}}{\vec{w}}{\Gamma}} 
\end{prooftree}
\vspace{2pt}\\
\begin{prooftree}
	\infer[no rule]0{\labeltext{[Tvar]}{[Tvar]}}
	\infer[no rule]1{\eta = \eta', \rvar{X}{x, \vec{y}} \mapsto \Delta, x:Y; \Gamma} 
	\infer1[\text{[Tvar]}]{\rvar{X}{z, \vec{w}} \vdash_\eta \subs{\vec{w}}{\vec{y}}{\Delta}, z:Y;\subs{\vec{w}}{\vec{y}}{\Gamma}}
\end{prooftree}
\vspace{2pt}\\ 
\begin{prooftree}
	\infer[no rule]0{\labeltext{[T$\mu$]}{[Tmu]}}
	\infer[no rule]1{P \vdash_\eta \Delta, x:\subs{\trec{X}{A}}{X}{A}; \Gamma}
	\infer1[\text{[T$\mu$]}]{\unfold{x}{P} \vdash_\eta \Delta, x:\trec{X}{A};\Gamma} 
\end{prooftree} 
\quad
\begin{prooftree}
	\infer[no rule]0{\labeltext{[T$\nu$]}{[Tnu]}}
	\infer[no rule]1{P \vdash_\eta \Delta, x:\subs{\tcorec{X}{A}}{X}{A}; \Gamma}
	\infer1[\text{[T$\nu$]}]{\unfoldn{x}{P} \vdash_\eta \Delta, x:\tcorec{X}{A};\Gamma} 
\end{prooftree} 
\end{array}
$$
\vspace{-10pt}
\caption{Typing Rules II: Induction and Coinduction.\label{fig:typing-induction}}
\end{figure} 
\begin{figure}[t]
$$
\begin{array}{c}
\begin{prooftree}
	\infer[no rule]0{\labeltext{[Taffine]}{[Taffine]}}
	\infer[no rule]1{P \vdash_\eta a:A, \vec{b} : \coaff{\vec{B}}, \vec{c} : \cusageF{\vec{C}}; \Gamma }
	\infer1[\text{[Taffine]}]{\daffine{a}{P}{\vec{b}}{\vec{c}} \vdash_\eta a:\aff{A}, \vec{b} : \coaff{\vec{B}}, \vec{c} : \cusageF{\vec{C}}; \Gamma} 
\end{prooftree} 
\vspace{2pt} \\ 
\begin{prooftree}
	\infer[no rule]0{\labeltext{[Tdiscard]}{[Tdiscard]}}
	\infer1[\text{[Tdiscard]}]{\discard{a} \vdash_\eta  a:\coaff{A}; \Gamma}
\end{prooftree}
\quad 
\begin{prooftree}
	\infer[no rule]0{\labeltext{[Tuse]}{[Tuse]}}
	\infer[no rule]1{Q \vdash_\eta \Delta, a: A; \Gamma} 
	\infer1[\text{[Tuse]}]{\use{a}{Q} \vdash_\eta \Delta, a:\coaff{A}; \Gamma}
\end{prooftree} 
\end{array}
$$
\vspace{-10pt}
\caption{Typing Rules III: Affinity. \label{fig:typing-affinityx}}
\end{figure} %
}
\hide{
Coinductive types are introduced by rule \ref{[Tcorec]}. It types corecursive processes  $\ncorec{z}{X}{\vec{w}}{x,\vec{y}}{P}$, with parameters $z,\vec{w}$ bound in $P$, that are instantiated with the arguments $x, \vec{y}$ (free in the process term). 
By convention, the coinductive behaviour, of type $\tcorec{Y}{A}$, of a corecursive process is always offered in the first argument $z$. 
 According to \ref{[Tcorec]}, to type the body $P$ of a corecursive process, the map $\eta$ is extended with a coinductive hypothesis
binding the process variable $X$ to the typing context $\Delta, z:Y; \Gamma$, so that when typing the body $P$ of the corecursion we can appeal to $X$, which intuitively stands for $P$ itself, and recover its typing invariant. 
 Crucially, the type variable $Y$ is free only in $z:A$. This causes
corecursive calls to be always applied to names $z'$  that hereditarily descend from the initial corecursive argument $z$, a necessary condition for strong normalisation (Theorem~\ref{theorem:SN}), and morally corresponds to only allowing corecursive calls on ``smaller'' argument sessions (of inductive type).

Rule~\ref{[Tvar]}
types a corecursive call $X(z, \vec{w})$ by looking up in $\eta$ for the corresponding binding and renaming the parameters with the arguments of the call. Inductive and coinductive types are explicitly unfolded with \ref{[Tmu]} and \ref{[Tnu]}.
 
To simplify the presentation in program examples, we omit explicit unfolding actions, and write inductive and coinductive type definitions with equations of the form $ \ind  \ A = f(A)$ and $\coind \ B  = f(B)$ instead of $A = \trec{X}{f(X)}$ and $B = \tcorec{X}{f(X)}$, respectively. Similarly, we write corecursive process definitions  as $Q(x, \vec{y}) =f(Q(-))$ instead of $Q(x,\vec{y}) = \ncorec{z}{X}{\vec{w}}{x,\vec{y}}{f(X(-))}$, while of course respecting the constraints imposed by typing rules \ref{[Tvar]} and \ref{[Tcorec]}. 

\vspace{-10pt}
\subsubsection{Affinity}
Affinity is important to model discardable linear resources, and
plays
an important role in \CLASS{}.
An affine session can either be used as a linear session or discarded.
The typing rules for the affine modalities are in Fig.~\ref{fig:typing-affinityx}.  Affine sessions are introduced by rule \ref{[Taffine]}
that promotes a linear $a:A$ to an affine session $a:\aff{A}$. It types $\daffine{a}{P}{\vec{b}}{\vec{c}}$, which provides an affine session at $a$ and continues as $P$, and follows the structure of 
a standard promotion rule. 

A session $a$ may be promoted to affine if it only depends on resources that can be disposed,
i.e. resources that satisfy some form of weakening capability, namely: 
coaffine sessions $b_i$ of type $\coaff{B_i}$, that can be discarded; 
full cell usages $c_i$ of type with $\cusageF{C_i}$, that can be released; and
unrestricted sessions in $\Gamma$, which are implicitly $?$-typed.
\hide{
(we may compare
\ref{[Taffine]} with the standard linear logic promotion rule for $!A$, where likewise the replicated session may only depend on $?$-typed names.)
$$
\frac{P\vdash y:A,?\Gamma}{!x(y);P\vdash x:{!}A,?\Gamma}
$$
}
The dependencies  of an affine object on coaffine or full cell objects  are explicitly annotated as $\vec{b},\vec{c}$ in the process term, to instrument the operational semantics, but we often omit them and simply write $\affine{a}{P}$. 

The coaffine endpoint $\coaff{A}$ of an affine session, dual of $\aff{\dual A}$, has two operations: use and discard. Rule \ref{[Tuse]}
\hide{
$$
\begin{prooftree}
	\infer[no rule]0{Q \vdash_\eta \Delta, a: A; \Gamma} 
	\infer1[\text{[Tuse]}]{\use{a}{Q} \vdash_\eta \Delta, a:\coaff{A}; \Gamma}
\end{prooftree} 
$$}types a process $\use{a}{Q}$ that uses a coaffine session $a$ and continues as $Q$, it is a dereliction rule.
\ref{[Tdiscard]} types the process $\discard{a}$ that discards a coaffine session $a$, it is a weakening rule.

\vspace{-12pt}
\subsubsection{Shared Mutable State}
Shared state is introduced in \CLASS{} by 
typed constructs that model mutex memory cells,
and associated cell operations allowing its use by client code, defined by the tying rules in Fig.~\ref{fig:typing-statex}.

 At any moment a cell may be either \emph{full} or \emph{empty}, akin to the MVars of Concurrent Haskell~\cite{jones1996concurrent}.
A full cell on $c$, written $\ncell{c}{a}{P}$, is typed $\cstateF{A}$ by rule \ref{[Tcell]}. Such cell stores an \emph{affine} session of type $\aff A$, implemented at $a$ by $P$.  
All objects stored in cells are required to be affine, so that memory cells may always 
be safely disposed without causing memory leaks.
An empty cell on $c$, of type $\cstateE{A}$, and written $\mempty{c}$,  is typed by rule \ref{[Tempty]}. 
 \begin{figure}[t]
$$
\begin{array}{c}
\begin{prooftree}
	\infer[no rule]0{\labeltext{[Tcell]}{[Tcell]}}
	\infer[no rule]1{P \vdash_\eta \Delta, a:\aff{A}; \Gamma} 
	\infer1[\text{[Tcell]}]{\ncell{c}{a}{P} \vdash_\eta \Delta, c:\cstateF  A; \Gamma}
\end{prooftree} 
\quad
\begin{prooftree}
	\infer[no rule]0{\labeltext{[Trelease]}{[Trelease]}}
	\infer1[\text{[Trelease]}]{\free{c} \vdash_\eta  c:\cusageF A; \Gamma}
\end{prooftree}
\vspace{2pt}\\
\begin{prooftree}
	\infer[no rule]0{\labeltext{[Tempty]}{[Tempty]}}
	\infer1[\text{[Tempty]}]{\mempty{c} \vdash_\eta  c:\cstateE  A; \Gamma}
\end{prooftree}
\quad
\begin{prooftree}
	\infer[no rule]0{\labeltext{[Ttake]}{[Ttake]}}
	\infer[no rule]1{Q \vdash_\eta \Delta, a:\coaff{A}, c: \cusageE{A}; \Gamma}
	\infer1[\text{[Ttake]}]{\take{c}{a}{Q} \vdash_\eta \Delta, c: \cusageF{ A};\Gamma}	
\end{prooftree} 
\vspace{2pt}\\ 
\begin{prooftree}
		\infer[no rule]0{\labeltext{[Tput]}{[Tput]}}
		\infer[no rule]1{Q_1 \vdash_\eta \Delta_1, a:\aff \dual A; \Gamma} 
		\infer[no rule]0{Q_2 \vdash_\eta \Delta_2, c: \cusageF A; \Gamma} 
		\infer2[\text{[Tput]}]{\mput{c}{a}{Q_1}{Q_2} \vdash_\eta \Delta_1, \Delta_2, c: \cusageE A; \Gamma} 
\end{prooftree}
\end{array}
$$
\vspace{-10pt}
\caption{Typing Rules IV: Reference Cells.\label{fig:typing-statex}}
\end{figure} 

Client processes manipulate cells via \emph{take}, \emph{put} and \emph{release}
operations. These operations apply to names of cell usage types -
 $\cusageF{A}$ (full cell usage) and $\cusageE{A}$ (empty cell usage) - which are dual types of $\cstateF{ \dual A}$ and $\cstateE{ \dual A}$, respectively. 
At any given moment, a client thread owning a $\cusageF{A}$-typed usage to a cell may execute a \emph{take} operation, typed by rule \ref{[Ttake]}. 
The \emph{take} operation $\take{c}{a}{Q}$
waits to acquire the cell mutex $c$, and reads its contents 
 into
 parameter $a$,
 the linear
(actually coaffine, of type $\coaff{A}$) usage for the object stored in the cell; the cell becomes empty, and execution continues as $Q$. 

It is responsibility of the taking
thread to put some value back in the empty cell, thus releasing the lock, causing the cell to transition to the full state. 
The \emph{put} operation $\mput{c}{a}{Q_1}{Q_2}$ is typed by \ref{[Tput]},
the stored object $a$, implemented by $Q_1$, is required to be affine, as specified in the premise $a:\aff \dual{A}$. 

Hence a cell flips from full to empty and back; \ref{[Ttake]} uses the cell $c$ at $\cusageF{A}$ type, and its continuation (in the premise) at $ \cusageE{A} $ type, symmetrically \ref{[Tput]} uses the cell $c$ at $\cusageE{A}$ type, and its continuation (in the premise) at $ \cusageF{A} $ type.

The $\free{c}$ operation allows a thread to manifestly drop
its cell usage $c$. Release is typed by \ref{[Trelease]} (cf. coweakening~\cite{ehrhard2018introduction});
a usage may only be released in the unlocked state $\cusageF{A}$.
When, for some cell $c$, all the owning threads release their usages, which eventually happens in well-typed programs, the cell $c$ gets disposed, and its (affine) contents safely discarded.

Our memory cells cells are linear objects, 
with a linear mutable payload, which are never duplicated 
by reduction or conversion rules. 
However, in \CLASS{}, multiple cell usages may be shared between concurrent threads, which compete to take and use it in interleaved critical sections. Such aliased usages be passed around and duplicated dynamically, changing the sharing topology at runtime. 

Sharing of cell usages is logically expressed in our system by the
typing rules in Fig.~\ref{fig:typing-sharex}.
\begin{figure}[t]
$$
\begin{array}{c}
\begin{prooftree}
	\infer[no rule]0{\labeltext{[Tsh]}{[Tsh]}}
	\infer[no rule]1{P \vdash_\eta \Delta', c:\cusageF  A; \Gamma \quad Q \vdash_\eta \Delta, c:\cusageF A; \Gamma} 
	\infer1[\text{[Tsh]}]{\share{c}{P}{Q} \vdash_\eta \Delta', \Delta, c:\cusageF A; \Gamma} 
\end{prooftree}
\hide{
\quad
\begin{prooftree}
	\infer[no rule]0{\labeltext{[Tsum]}{[Tsum]}}
	\infer[no rule]1{P \vdash_\eta \Delta; \Gamma \quad Q \vdash_\eta \Delta; \Gamma}
	\infer1[\text{[Tsum]}]{\nsum{P}{Q} \vdash_\eta \Delta; \Gamma} 
\end{prooftree}
}
\vspace{2pt}\\
\begin{prooftree}
	\infer[no rule]0{\labeltext{[TshL]}{[TshL]}}
	\infer[no rule]1{P \vdash_\eta \Delta', c:\cusageE  A; \Gamma \quad Q \vdash_\eta \Delta, c:\cusageF A; \Gamma} 
	\infer1[\text{[TshL]}]{\share{c}{P}{Q} \vdash_\eta \Delta', \Delta, c:\cusageE A; \Gamma} 
\end{prooftree}
\vspace{2pt}\\
\begin{prooftree}
	\infer[no rule]0{\labeltext{[TshR]}{[TshR]}}
	\infer[no rule]1{P \vdash_\eta \Delta', c:\cusageF  A; \Gamma \quad Q \vdash_\eta \Delta, c:\cusageE A; \Gamma} 
	\infer1[\text{[TshR]}]{\share{c}{P}{Q} \vdash_\eta \Delta', \Delta, c:\cusageE A; \Gamma} 
\end{prooftree}
\end{array}
$$
\vspace{-10pt}
\caption{ \label{fig:typing-sharex}Typing Rules V: State Sharing.}
\end{figure}
Co-contraction, introduced in Differential Linear Logic DiLL~\cite{ehrhard2018introduction},
allows finite multisets of linear resources to safely
interact in cut-reduction, resolving concurrent sharing into nondeterminism, as required here to soundly model memory cells and their linear concurrent usages.
 \hide{
 , and was already explored, with a different interpretation, in~\cite{rocha2021propositions}, but we here target to discipline concurrent take operations (which corresponds to DiLL co-dereliction) on a shared cell reference.
DiLL is a ``quantitative" version of linear logic, where exponentials $!_dA$ and $?_d\dual A$ are interpreted as typing finite multisets of linear  objects of respectively type $A$ and $\dual A$.
While the usual contraction, dereliction apply to $?_dA$,
dual principles of co-contraction, co-dereliction and co-weakening apply to $!_dA$. Proofs of type $!_dA$ and $?_d\dual A$ interact via cut-elimination
 matching and distributing offers and uses, reducing to a failure term in case there is a mismatch in the number of
clients and servers. Mon-determinism arises  
 Non-determinism arises because  
 
A cell usage can be shared by an arbitrary number of concurrent threads, as expressed by rules  \ref{[Tsh]}, \ref{[TshL]} and \ref{[TshR]}, that type the process construct $\share{c}{P}{Q}$, where $c$ is being shared between two concurrent threads $P$ and $Q$.    
}
Rule \ref{[Tsh]} 
interprets cocontraction with the construct $\share{c}{P}{Q}$,
and types sharing of the cell usage $c:\cusageF{A}$ between the concurrent threads $P$ and 
$Q$.

Contrary to cut, $\share{c}{P}{Q}$ is \emph{not} a binding operator for $c$. The shared usage $c:\cusageF{A}$ is \emph{free} in the conclusion of 
the typing rule, permitting $c$ to be shared 
among an arbitrary number of threads, 
by nested iterated use of \ref{[Tsh]}.
In \ref{[Tsh]}, $P$ and $Q$ only share the single mutex cell $c$, since the linear context is split
multiplicatively, just like \ref{[Tcut]} wrt. binary sessions. 
This condition comes from the DiLL typing discipline, and is important to ensure deadlock freedom.

While \ref{[Tsh]} types sharing of a full (unlocked) cell usage of
type $\cusageF A$, the symmetric rules \ref{[TshR]} and \ref{[TshR]} type sharing of an empty (locked) cell usage of
type $\cusageE A$. 
We may verify that for every cell $c$ in a well-typed process, at most one unguarded operation to $c$ may be using type $\cusageE{A}$, all the remaining unguarded operations to $c$ must be using type $\cusageF{A}$.
This implies that, at runtime, only one thread may own the lock for a given (necessarily empty) cell, and execute a \emph{put} to it, which will bring the cell back to full and release its lock, other threads must be either attempting to \emph{take}, or \emph{release} the reference. 

Working together, the sharing typing rules ensure 
that in any well-typed cell sharing tree, at most one single
thread at any time may be actively using  a cell (in the locked empty state) and put to it,
thus guaranteeing mutual exclusion, while satisfying Progress (Theorem~\ref{theorem:progress}) which in turn ensures deadlock absence, even in the presence of the crucially blocking behaviour of the take operation.
}

 
\vspace{-16pt}
\section{A Core Session Abstract Machine}
\label{sec:coreSAM}

In this section we develop the key insights that guide the
construction of our session abstract machine (SAM) and introduce its
operational rules in an incremental fashion. We omit the linear logic
exponentials for the sake of clarity of presentation, postponing their
discussion for Section~\ref{sec:sam-mix}.

One of the main observations that drives the design of the SAM is the nature
of proof dynamics in (classical) linear logic, and thus of process
execution dynamics in the $\CLL$ system of
Section~\ref{sec:language}. The proof dynamics of linear logic are
derived from the computational content of the cut elimination proof,
which defines a proof simplification strategy that removes
(all) instances of the cut rule from a proof. However, the strategy induced by cut
elimination is \emph{non-deterministic} insofar as multiple simplification steps
may apply to a given proof.
Transposing this observation to $\CLL$ and other related systems, we
observe that their operational semantics is does not prescribe a rigid
evaluation order for processes. For instance, in the process
$\cuti{x}{P}{Q}$, reduction is allowed in both $P$ and $Q$. 
{This is of course
in line with reduction in process calculi (e.g.,~\cite{DBLP:books/daglib/0098267}). However,} in
logical-based systems this amounts to \emph{don't care}
non-determinism since, regardless of the evaluation order, confluence
ensures that the same outcomes are produced (in
opposition to \emph{don't know} non-determinism which breaks
confluence and is thus disallowed in purely logical systems). 
The design of the SAM arises from attempting to fix a purely
sequential reduction strategy for $\CLL$ processes, such that only
\emph{one} process is allowed to execute at any given point in time,
in the style of coroutines. To construct such a strategy, we forego
the use of purely synchronous communication channels, which require a
handshake between two concurrently executing processes, and so
consider session channels as a kind of \emph{buffered} communication
medium (this idea has been explored in the context of linear logic
interpretations of sessions in~\cite{DeYoung2012}), or queue, where one process
can asynchronously write messages so that another may, subsequently,
read.  To ensure the correct directionality of communication, the
queue has a write endpoint (on which a process may only write) and a
read endpoint (along which only reads may be performed), such that at
any given point in time a process can only hold one of two endpoints
of a queue. Moreover, our design takes inspiration from insights
related to polarisation and focusing in linear logic, grouping
communication in sequences of positive (i.e.~write) actions.

\begin{figure}[t]
{\small\[
\begin{array}{llll}
S & ::= & (P,  H) & \mbox{State}
\vspace{4pt}\\
  H & ::= & \mathit{SessionRef} \to \mathit{SessionRec} \quad\quad\;&\mbox{Heap}
\vspace{4pt}\\
  \mathit{R} & ::= & \srecs{x}{q}{ P}{y} & \mbox{Session Record}
  \vspace{4pt}\\
  \mathit{q} & ::= & \nilm \; | \; \mathit{Val} @ q & \mbox{Queue}
\vspace{4pt}\\
\mathit{Val} & ::= & \checkmark & \mbox{Close token}\\
 & | & \labl{l} & \mbox{Choice label}\\
  & | & \mathsf{clos}(x,P) & \mbox{Process Closure}
\end{array}
\]}
\caption{core SAM Components\label{fig:sam_syntax}\vspace{-10pt}}
\end{figure}
{Allowing session channels to buffer message sequences,} we may then model process
execution by alternating between writer processes (that inject
messages into the respective queues) and corresponding reader
processes. Thus, the SAM must maintain a \emph{heap} that tracks the
queue contents of each session (and its endpoints), as well as the
suspended processes.  The construction of the core of the SAM is given in
Figure~\ref{fig:sam_syntax}. An execution state is simply a pair
consisting of \emph{the} running process $P$ and the heap $H$.
For technical reasons that are made clear in Sections~\ref{sec:blang}
and~\ref{sec:soundcore}, the process language used in the SAM differs
superficially from that of $\CLL$, but for the purposes of this
overview we will use $\CLL$ process syntax. Later we show the two
languages are equivalent in a strong sense.
%
%

A heap
is a mapping between session identifiers and \emph{session
  records} of the form $\srecs{x}{q}{ Q}{y}$, denoting a session
with write endpoint $x$ and read endpoint $y$, with queue contents $q$
and a suspended process $Q$, holding one of the two endpoints.
If $Q$ holds the read endpoint then it is suspended waiting for the
process holding the write endpoint to fill the queue with data for it
to read.
If $Q$ holds the write endpoint, then $Q$ has been suspended
\emph{after} filling the queue and is now waiting for the reader
process on $y$ to empty the queue.

We adopt the convention of placing the write endpoint on the
left and the read endpoint on the right. In general, session records
in the SAM support a form of coroutines through their contained
processes, which are {called on and returned from}
multiple times over the course of the execution of the machine.
A queue can either be empty ($\nilm$) or 
holding a sequence of values. A value is either a close session token
($\checkmark$), identifying the last output on a session;
a choice label $\labl{l}$ or a process closure $\mathsf{clos}(x,P)$,
used to model session send and receive. We overload the $@$ notation
to also denote concatenation of queues.

\mypara{Cut}
We begin by considering how to execute a cut of the form
$\cuti{x:A^{+}}{P}{Q}$ 
where $x$ is a positive type (in the sense of polarized
logic~\cite{DBLP:journals/mscs/Girard91}) in $P$. A positive type
corresponds to a type denoting an
\emph{output} (or write) action, whereas a negative type denotes
an \emph{input} (or read) action.
We maintain the invariant that
in such a cut, $P$ holds the write endpoint and $Q$ the read
endpoint. This means that the next action performed by $P$ on the
session will be to push some value onto the queue and, dually, the
next action performed by $Q$ on the session will be to read a value
from the queue. In general, the holder of the write and read endpoint
can change throughout execution.

Given the choice of either scheduling $P$ or $Q$, we are effectively
\emph{forced} to schedule 
$P$ \emph{before} $Q$. Given that the cut introduces the
(unique) session 
that is shared between the two processes, the only way for $Q$ to
exercise its read capability on the session successfully is to wait
for $P$ to have exercised (at least some of) its write capability. If
we were to schedule $Q$ before $P$, the process might attempt to read a
value from an empty queue, resulting in a stuck state of the SAM.
Thus, the SAM execution rule for cut is:
\[
(\cuti{x:A^{+}}{P}{Q}, H) \Mapsto 
(P, H[\srecs{x}{\nilm}{ \subs{y}{x}Q }{y}  ]) \qquad \text{\rm[SCut]}
\]
The rule states that $P$ is the process that is to be scheduled, adding
the session record $\srecs{x}{\nilm}{Q}{y}$ to the heap, which effectively
suspends the execution of $Q$ until $P$ has exercised some of its
write capabilities on the new session.
Note that, in general, both $P$ and $Q$ can interact along many
different sessions as both readers and writers before exercising any
action on $x$ (resp. $y$). However, they alone
hold the freshly created endpoints $x$ and $y$ and so the next value
sent along the session must come from $P$ and $Q$ is its intended
receiver.

\mypara{Channel Output}
To execute an output of the form $\potimes{x}{z}{R}{Q}$ in the SAM we
simply lookup the session record for $x$ and add to the queue a
\emph{process closure} containing $R$ (which interacts along $z$),
continuing with the execution of $Q$: 
\[
( \potimes{x}{z}{R}{Q}, H[\srecs{x}{q}{P}{y}]) \Mapsto 
(Q, H[\srecs{x}{q @ \clos(z,R)}{P}{y}]) \qquad \text{\rm[S$\otimes$]}
\]
Note that the SAM \emph{eagerly} continues to execute $Q$
instead of switching to $P$, the holder of the read endpoint of the
queue. This allows for the running process to
perform all available writes before a context switch occurs.

\mypara{Session Closure}
Executing a $\pone{\!\!}$ follows a similar spirit, but no
continuation process exists and so execution switches to the process $P$
holding the \emph{read} endpoint $y$ of the queue:
\[
( \pone{x}, H[\srecs{x}{q}{P}{y}]) \Mapsto 
(P, H[\srecs{x}{q @ {\checkmark}}{\zero}{y}]) \qquad \text{\rm[S$\one$]} 
  \]
The process $P$ will eventually read the termination mark from the queue
(triggering the deallocation of the session record from the heap): 
\[(\pbot{y}{ P}, 
H[\srecs{x}{{\checkmark}}{\zero}{y}]) \Mapsto 
(P, H) \qquad \text{\rm[S$\bot$]}
\]
Note the requirement that $\checkmark$ be the final element
of the queue.

\mypara{Negative Action on Write Endpoint}
As hinted above for the case of executing a $\cutP$, the SAM has a
kind of \emph{write bias} insofar as the process chosen to execute in
a cut is that which holds the write endpoint for the newly created
session. Since $\CLL$ processes use channels bidirectionally, the role
of writer and reader on a channel (and thus the holder of the write
and read endpoints of the queue) may be exchanged during
execution. For instance, a process $P$ may wish to send a value $v$
to $Q$ and then receive a response on the same
channel. However, when considering a queue-based semantics, the
execution of the input action \emph{must not} obtain the value
$v$, intended for $Q$. Care is therefore needed to ensure that $v$ is received by the
holder of the read endpoint of the queue \emph{before} $P$ is allowed
to execute its input action (and so taking over the read endpoint).
This notion is captured by the following rule, where $\mathcal{A}^-$
denotes any process performing a negative polarity action (i.e.,~a
$\textcolor{pcolor}{\mathsf{wait}}$,
$\textcolor{pcolor}{\mathsf{recv}}$,
$\textcolor{pcolor}{\mathsf{case}}$ or, as we discuss later, a
{$\fwd{x}{y}$ when $x$ is a write endpoint with a negative polarity type}):
\[
(\mathcal A^-(x), 
H[\srecs{x}{{q}}{Q}{y}]) \Mapsto 
(Q, H [\srecs{x}{{q}}{\mathcal A^-(x)}{y}]) \qquad \text{\rm[S$- $]}
\]
If the executing process is to perform a negative polarity action on a write
endpoint $x$, the SAM context switches to $Q$, the holder of the read
endpoint $y$ of the session, and suspends the previously running
process. This will now allow for $Q$ to perform the appropriate inputs
before execution of the action $\mathcal{A}^-$ resumes.

\mypara{Channel Input}
The rules for $\textcolor{pcolor}{\mathsf{recv}}$ actions are as
follows:
\[
  \begin{array}{l}
 (\pparl{y}{w\ass +}{Q}, H[\srecs{x}{\clos(z,R)@q }{P}{y}]) \Mapsto 
(Q, H[\srecs{w}{\nilm}{R}{z}][\srecss{x}{q }{P}{y}]) \;\; \text{\rm[S$\parl_+$]}
\\
( \pparl{y}{w\ass -}{Q}, H[\srecs{x}{\clos(z,R)@q }{P}{y}]) \Mapsto 
(R, H[\srecs{z}{\nilm}{Q}{w}][\srecss{x}{q }{P}{y}])  \;\; \text{\rm[S$\parl_-$]}
  \end{array}
\]
where
$\srecss{x}{q }{P}{y} \deff \mbox{if }(q = \nilm)\mbox{ then
}\srecs{y}{q }{P}{x}\mbox{ else } \srecs{x}{q }{P}{y}$.  The execution
of an input action requires the corresponding queue to contain a
process closure, denoting the process that interacts along the
received channel $w$. In order to ensure that no inputs attempt to
read from an empty queue, we must \emph{branch} on the polarity of the
communicated session (written $w{:}+$ and $w{:}-$ in the rules above): if
the session has a \emph{positive} type, then $Q$ must take the
\emph{write} endpoint $w$ of the newly generated queue (since $Q$ uses the
session with a dual type) and thus we execute $Q$ and
{allocate a session record in the heap for the new session}\hide{extend the heap with a session record for the new session},
with read endpoint $z$; if the
exchanged session has a \emph{negative} type, the converse holds and
$Q$ must take the \emph{read} endpoint of the newly generated
queue. In this scenario, we must execute $R$ so that it may exercise
its write capability on the queue and suspend $Q$ in the new session
record.

In either case, the session record for the original session is updated
by removing the received message from the queue. Crucially, since
processes are well-typed, if the resulting queue is empty then it must
be the case that $Q$ has no more reads to perform on the session, and
so we \emph{swap} the read and write endpoints of the session. This
swap serves two purposes: first, it enables $Q$ to perform writes if
needed; secondly, and more subtly, it allows for the process,
say, $P$, that holds the other endpoint of the queue to be resumed to
perform its actions accordingly. To see how this is the case, consider
that such a process will be suspended (due to rule \text{\rm[S$- $]})
attempting to perform a negative action on the write endpoint of the
queue. After the swap, the endpoint of the suspended process now
matches its intended action. Since $Q$ now holds the write endpoint,
it will perform some number of positive actions on the session which
end either in a $\textcolor{pcolor}{\mathsf{close}}$, which context
switches to $P$, or until it attempts to perform a negative action on
the write endpoint, triggering rule \text{\rm[S$- $]} and so context
switching to $P$.

\mypara{Choice and Selection}
The treatment of the additive constructs in the SAM is straightforward:
\[
  \begin{array}{ll}
( \#\mathsf{l}\; x;Q, H[\srecs{x}{q}{P}{y}]) \Mapsto 
(Q, H[\srecs{x}{q @ \#\mathsf{l}}{P}{y}]) & \quad\text{\rm[S$\oplus$]}
\\
( \pcasem{y}{\ell}{L}{Q_\ell}, H[\srecs{x}{\#\mathsf{l}@q }{P}{y}]) \Mapsto 
(Q_{\labl{l}}, H[\srecss{x}{q }{P}{y}])
& \quad\text{\rm[S$\with$]}
  \end{array}
  \]
Sending a label $\#\mathsf{l} $ simply adds the $\#\mathsf{l} $ to the corresponding queue and
proceeds with the execution, whereas executing a
$\textcolor{pcolor}{\mathsf{case}}$ reads a label from
the queue and continues execution of the appropriate branch. Since
removing the label may empty the queue, we perform the same adjustment
as in rules \text{\rm[S$\parl_+$]} and \text{\rm[S$\parl_-$]}.

\mypara{Forwarding}
Finally, let us consider the execution of a forwarder (we overload the $@$ notation
to also denote concatenation of queues):
\[
( \fwd{x^-}{y^+}, H[\srecs{z}{q_1}{Q}{x}][\srecs{y}{q_2}{P}{w}]) \Mapsto 
(P, H[\srecs{z}{q_2 @ {q_1}}{Q}{w}]) \qquad  \text{\rm[Sfwd]} 
\]
A forwarder denotes the merging of two sessions $x$ and $y$.
Since the forwarder holds the read and write endpoints $x$ and $y$,
respectively, $Q$ has written (through $z$) the contents of $q_1$,
whereas the previous steps of the currently running process have
written $q_2$. Thus, $P$ is waiting to read $q_2 @ q_1$, justifying
the rule above.

The reader may then wonder about other possible configurations of the
SAM heap and how they interact with the forwarder. Specifically, what
happens if $y$ is of a positive type but a read endpoint of a queue,
or, dually, if $x$ is of a negative type but a write endpoint. The
former case is \emph{ruled out} by the SAM since the heap satisfies
the invariant that any session record of the form
$\srecs{x\ass A}{q}{P}{y\ass A}\in H$ is such that $A$ must be of
negative polarity or $P$ is the inert process (which cannot be
forwarded).  The latter case is possible and is handled by rule
\text{\rm[S$- $]}, since such a forward $\fwd{x^-}{y^+}$ stands
for a process that wants to perform a \emph{negative polarity} action
on a \emph{write} endpoint (or a positive action on a read endpoint).

\vspace{-10pt}
\subsection{On the Write-Bias of the SAM}
Consider the following $\CLL$ process:

\vspace{6pt}
\centerline{$
   P \triangleq \cuti{a:\one \otimes \one }{P_1}{\subs{a}{b}Q_1}
$}
$$
  \begin{array}{ll}
    P_1 \triangleq \potimes{a}{y}{P_2}{P_3} &  \qquad  Q_1 \triangleq \pparl{b}{x}{Q_2}\\
    P_2 \triangleq \pone{y} &  \qquad  Q_2 \triangleq \pbot{x}{Q_3}\\
    P_3 \triangleq \pone{a} &  \qquad  Q_3 \triangleq \pbot{b}{\zero}\\[1em]
  \end{array}
$$
Let us walk through the execution trace of $P$:
  \begin{tabbing}
   (1) $(P , \emptyset) \Mapsto$ \` \text{by \rm[SCut]}\\
   (2) $(P_1 , \srecs{a}{\nilm}{Q_1}{b}) \Mapsto$  \` \text{by \rm[S$\otimes$]}\\
    (3) $(P_3, \srecs{a}{\clos(y,P_2)}{Q_1}{b}) \Mapsto$ \` \text{by
      \rm[S$\one$]} \\
    (4) $(Q_1, \srecs{a}{\clos(y,P_2) @\checkmark }{\zero}{b}) \Mapsto$ \`
    \text{by \rm[S$\parl_-$]}\\
    (5) $(P_2 , \srecs{y}{\nilm}{Q_2}{x} , \srecs{a}{\checkmark
    }{\zero}{b} ) \Mapsto$ \` \text{by \rm[S$\one$]} \\
    (6) $(Q_2, \srecs{y}{\checkmark}{\zero}{x} , \srecs{a}{\checkmark
    }{\zero}{b}) \Mapsto$ \` \text{by \rm[S$\bot$]} \\
    (7) $(Q_3, \srecs{a}{\checkmark }{\zero}{b}) \Mapsto$ \` \text{by
      \rm[S$\bot$]}\\
    (8) $(\zero , \emptyset)$ 
\end{tabbing}

The SAM begins in the state on line (1) above, executing the
cut. Since the type of $a$ is positive, we execute $P_1$, and allocate
the session record, suspending $Q_1$, resulting in the state on line
(2). Since $P_1$ is a write action on a write endpoint, we proceed via
the \rm[S$\otimes$] rule, resulting in the SAM configuration in line
(3), executing $P_3$ and adding a closure containing
$P_2$ to the session queue with write endpoint $a$. Executing $P_3$
(3), a $\textcolor{pcolor}{\mathsf{close}}$ action, requires adding
the $\checkmark$ to the queue and context switching to the 
process $Q_1$, now ready to receive the sent value. The applicable
rule is now (4) \rm[S$\parl_-$], and so execution will context switch
to $P_2$ after creating the session record for the new session with
endpoints $y$ and $x$. $P_2$ will execute and the machine ends up in
state (6) followed by (7), which consume the appropriate $\checkmark$
and deallocate the session records.

Note how after executing the send action of $P_1$ we eagerly execute
the positive action in $P_3$ rather than context switching to
$Q_1$. While in this particular process it would have been safe to
execute the negative action in $Q_1$, switch to $P_2$ and then back to
$Q_2$, we would now need to somehow context switch to $P_3$
\emph{before} continuing with the execution of $Q_3$, or execution
would be stuck. However, the relationship between $P_3$ and $Q_2$ is
unclear at best. Moreover, if the continuation of $Q_1$ were of the
form $ \pbot{b}{\pbot{x}{\zero}}$, the context switch after the
execution of $P_2$ would have to execute $P_3$, or the machine
would also be in a stuck state.


\vspace{-10pt}
\subsection{Illustrating Forwarding}
To better illustrate the way in which $\fwd{x^-}{y^+}$ effectively
stands for a negative action, consider the following $\CLL$ process
(to simplify the execution trace we assume the existence of output and
input of integers typed as $\mathsf{int} \otimes A$ and
$\dual{\mathsf{int}}\parl A$, respectively, eliding the need
for process closures in this example):

\vspace{-0.3cm}
{\small\[
   P \triangleq \cuti{b:\dual{\mathsf{int}} \parl \dual{\mathsf{int}} \parl \one }{P_1}{\subs{b}{c}
     \cuti{a:{\mathsf{int}}\otimes\dual{\mathsf{int}} \parl
       \one}{Q_1}{\subs{a}{d} R_1}}
 \]}
\vspace{-0.5cm}
 {\small
\[  
  \begin{array}{lll}
    P_1 \triangleq \pparl{b}{x}{P_2}
    &  \qquad  Q_1 \triangleq \textcolor{pcolor}{\mathsf{send}}\,a(1);{Q_2}
    & \qquad R_1 \triangleq \pparl{d}{y}{R_2}\\
    P_2 \triangleq \pparl{b}{z}{P_3}
    &  \qquad  Q_2 \triangleq \textcolor{pcolor}{\mathsf{send}}\,c(3);{Q_3}
    & \qquad R_2 \triangleq \textcolor{pcolor}{\mathsf{send}}\,d(2);{R_3}\\
    P_3 \triangleq \pone{b} &  \qquad Q_3 \triangleq \fwd{a}{c}
    & \qquad R_3 \triangleq \pbot{d}{\zero}\\[1em]
  \end{array}
\]}

If we consider the execution of $P$ we observe:
 {\small\begin{tabbing}
   (1) $(P , \emptyset) \Mapsto$ \` \text{by \rm[SCut]}\\
  (2) $( \cuti{a}{Q_1}{\subs{a}{d}R_1} , \srecs{c}{\nilm}{P_1}{b}) \Mapsto$ 
  \` \text{by \rm[SCut]}\\
  (3) $ ( Q_1 , \srecs{a}{\nilm}{R_1}{d} ,\srecs{c}{\nilm}{P_1}{b} )
  \Mapsto$ \` \text{by \rm[S$\otimes$]}\\
  (4) $(Q_2 ,  \srecs{a}{1}{R_1}{d} , \srecs{c}{\nilm}{P_1}{b} )
  \Mapsto$ \` \text{by \rm[S$\otimes$]}\\
  (5) $(\fwd{a}{c} ,  \srecs{a}{1}{R_1}{d} , \srecs{c}{3}{P_1}{b} )
  \Mapsto$ \` \text{by \rm[S$-$]}\\
  (6) $(R_1,  \srecs{a}{1}{Q_3 }{d} , \srecs{c}{3}{P_1}{b})
  \Mapsto$ \` \text{by \rm[S$\parl$]}\\
  (7) $(R_2 , \srecs{d}{\nilm}{Q_3 }{a},\srecs{c}{3}{P_1}{b} )
  \Mapsto$ \` \text{by \rm[S$\otimes$]}\\
  (8) $(R_3 , \srecs{d}{2}{Q_3 }{a},\srecs{c}{3}{P_1}{b} )
  \Mapsto$ \` \text{by \rm[S$-$]}\\
  (9) $(\fwd{a}{c} , \srecs{d}{2}{R_3 }{a},\srecs{c}{3}{P_1}{b} )
  \Mapsto$ \` \text{by \rm[Sfwd]}\\
  (10) $(P_1 , \srecs{d}{3@ 2}{R_3}{b}) \Mapsto$ \` \text{by
    \rm[S$\parl$]}\\
  (11) $(P_2 , \srecs{b}{2}{R_3}{d}) \Mapsto$ \` \text{by
    \rm[S$\parl$]}\\
    (12) $(P_3 , \srecs{b}{\nilm}{R_3}{d}) \Mapsto$ \` \text{by
    \rm[S$\one$]}\\
  (13) $(R_3, \srecs{b}{\checkmark}{R_3}{d}) \Mapsto$ \` \text{by
    \rm[S$\bot$]}\\
  (14) $(\zero , \emptyset)$
\end{tabbing}}

The first four steps of the execution of $P$ allocate the two session
records and the writes by $Q_1$ and $Q_2$ takes place. We are now in
configuration (5), where $Q_3 = \fwd{a^-}{c^+}$ is to execute and $a$ is a
write endpoint of a queue assigned a negative type ($\dual{\mathsf{int}}\parl
\one$). This forwarder stands for a process performing a
negative action on a write endpoint (i.e.,~$P_1$) and so context
switching is required, rule \rm[S$-$] applies
and the SAM context switches to $R_1$, suspending $Q_3$ until the
forward can be performed.
After $R_1$ receives (6) and the queue endpoints $a$ and $d$ are
swapped (7), $R_2$ executes and then rule  \rm[S$-$] applies (8),
context switching back to $Q_3$. Since the queue endpoints are now
flipped, rule \rm[Sfwd] now applies (9), collapsing the two session
records (via queue concatenation) and proceeding with the execution of
$P_1$, $P_2$, $P_3$ and $R_3$
(10-14). Note the correct ordering in which the sent values are
dequeued, where $3$ is read before $2$, as intended.


\mypara{Discussion}
The core execution rules for the SAM are summarized in
Figure~\ref{fig:rules}.  At this point, the reader may wonder just how
reasonable the SAM's evaluation strategy is. Our evaluation strategy
is devised to be a deterministic, sequential strategy, where exactly
one process is executing at any given point in time,
supported by a queue-based buffer structure for channels
and a heap for session records.
Moreover,
taking inspiration from focusing and polarized logic, we adopt a
write-biased stance and prioritize (bundles of) write actions over
reads, where suspended processes hold the read endpoint of queues
while waiting for the writer process to fill the
queue, and hold write endpoints of queues \emph{after} filling them,
waiting for the reader process to empty the queue. 

While this latter point seems like a reasonable way to ensure that
inputs never get stuck, it is not immediately obvious that the
strategy is sound wrt the more standard (asynchronous) semantics of
$\CLL$ and related languages, given that processes are free to act on
multiple sessions. Thus, the write-bias of the cut rule (and
the overall SAM) does
not necessarily mean that the process that is chosen to execute will
immediately perform a write action on the freshly cut session $x$. In
general, such a process may perform multiple write or read actions on many
other sessions before performing the write on $x$, meaning that
multiple context switches may occur. Given this, it is not 
obvious that this strategy is adequate insofar as preserving the
correctness properties of $\CLL$ in terms of soundness, progress and type
preservation. The remainder of this paper is devoted to establishing
this correspondence in a precise technical sense.

\begin{figure}[t]
{\small$$
\begin{array}{lll}
(\cuti{x:A^{+}}{P}{Q}, H) \Mapsto 
(P, H[\srecs{x}{\nilm}{\subs{y}{x}Q}{y} ] ) & \text{\rm[SCut]}
\vspace{6pt}
\\
( \fwd{x}{y}, H[\srecs{z}{q_1}{Q}{x}][\srecs{y}{q_2}{P}{w}]) \Mapsto 
(P, H[\srecs{z}{q_2 @ {q_1}}{Q}{w}])&  \text{\rm[Sfwd]} 
\vspace{6pt}
\\
( \pone{x}, H[\srecs{x}{q}{P}{y}]) \Mapsto 
(P, H[\srecs{x}{q @ {\checkmark}}{\zero}{y}])&  \text{\rm[S$\one$]} 
\vspace{6pt}
\\
(\pbot{y}{ P}, 
H[\srecs{x}{{\checkmark}}{\zero}{y}]) \Mapsto 
(P, H) &  \text{\rm[S$\bot$]} 
\vspace{6pt}
\\
(\mathcal A^-(x), 
H[\srecs{x}{{q}}{Q}{y}]) \Mapsto 
(Q, H [\srecs{x}{{q}}{ \mathcal A^-(x)}{y} ]) & \text{\rm[S$- $]}
\vspace{6pt}
\\
( \potimes{x}{z}{R}{Q}, H[\srecs{x}{q}{P}{y}]) \Mapsto 
(Q, H[\srecs{x}{q @ \clos(z,R)}{P}{y}]) & \text{\rm[S$\otimes$]}
\vspace{6pt}
\\
( \pparl{y}{w:+}{Q}, H[\srecs{x}{\clos(z,R)@q }{P}{y}]) \Mapsto 
(Q, H[\srecs{w}{\nilm}{R}{z}][\srecss{x}{q }{P}{y}])
& \text{\rm[S$\parl$]}
\vspace{6pt}
\\
( \pparl{y}{w:-}{Q}, H[\srecs{x}{\clos(z,R)@q }{P}{y}]) \Mapsto 
(R, H[\srecs{z}{\nilm}{Q}{w}][\srecss{x}{q }{P}{y}])
& \text{\rm[S$\parl$]}
\vspace{6pt}
\\
( \#\mathsf{l}\; x;Q, H[\srecs{x}{q}{P}{y}]) \Mapsto 
(Q, H[\srecs{x}{q @ \#\mathsf{l}}{P}{y}]) & \text{\rm[S$\oplus$]}
\vspace{6pt}
\\
( \pcasem{y}{\ell}{L}{Q_\ell}, H[\srecs{x}{\#\mathsf{l}@q }{P}{y}]) \Mapsto 
(Q_{\labl{l}}, H[\srecss{x}{q }{P}{y}])
& \text{\rm[S$\with$]}
\vspace{6pt}
\\
N.B.: \srecss{x}{q }{P}{y} \deff \mbox{if }(q = \nilm)\mbox{ then }\srecs{y}{q }{P}{x}\mbox{ else }
\srecs{x}{q }{P}{y}
\end{array}
$$}
\caption{The core SAM Transition Rules}\label{fig:rules}
\end{figure}







 

\section{ $\Blang$: A Buffered Formulation of $\CLL$}\label{sec:blang}

There is a substantial gap between the language $\CLL$, presented in an abstract algebraic style, and its operational semantics, defined by equational and rewriting systems, and an abstract machine as the SAM, a deterministic state machine manipulating several low level structures. Therefore, even if the core SAM structure and transition rules are fairly simple, proving its correctness is more challenging and technically involved, and require progressive build up. 
Therefore, we first bridge between $\CLL$ and SAM via a intermediate logical language $\Blang$, which extends $\CLL$ with a buffered cut construct.
\begin{figure}[t]
  {\small
    $$
\begin{array}{c}
\begin{prooftree}
	\infer[no rule]0{\labeltext{[Tcut]}{[Tcut]}}
	\infer[no rule]1{P \vdashB  \Delta', x:\dual A; \Gamma \quad Q \vdash^{\mathsf{B}}  \Delta, y: A; \Gamma} 
	\infer1{\cuti{x:\dual{A}\; [\nilm]\; \ou{y}:A}{P}{Q} \vdash^{\mathsf{B}}  \Delta', \Delta; \Gamma} 
\end{prooftree}\quad \mbox{(A positive)}\quad \text{[TcutB]  }
\vspace{6pt}\\
\begin{prooftree}
	\infer[no rule]0{\labeltext{[Tcut]}{[Tcut]}}
	\infer[no rule]1{\cuti{\ou{x}:\one \; [{q}]\;  y:B}{\pone{x}}{Q} \vdashB \Delta; \Gamma} 
	\infer1{\cuti{\ou{x}: \emptyset [q @ \checkm] y:B}{\zero}{Q} \vdashB \Delta; \Gamma} 
\end{prooftree} \text{[Tcut-$\one$]}
\vspace{6pt}\\
\begin{prooftree}
	\infer[no rule]0{\labeltext{[Tcut]}{[Tcut]}}
	\infer[no rule]1{\cuti{\ou{x}:T{\otimes} A \;[q]\; y:B}{\potimes{x}{y}{R}{P}}{Q} \vdashB \Delta; \Gamma} 
	\infer1{\cuti{\ou{x}:A \;[q @ \clos(y,R)   ] \; y:B}{P}{Q} \vdashB \Delta; \Gamma} 
\end{prooftree} \text{[Tcut-$\otimes$]}
\vspace{6pt}\\
\begin{prooftree}
	\infer[no rule]0{\labeltext{[Tcut]}{[Tcut]}}
	\infer[no rule]1{\cuti{\ou{x}:\oplusm{\ell}{L}{A_\ell} \;[q]\; y:B}
	{\labl{l}\ x;P}{Q} \vdashB \Delta; \Gamma} 
	\infer1{\cuti{\ou{x}:A_{\labl{l}} \;[q @ \labl{l} ] \; y:B}{P}{Q}
	 \vdashB \Delta; \Gamma} 
\end{prooftree} \text{[Tcut-$\oplus$]}
\vspace{6pt}\\
\begin{prooftree}
	\infer[no rule]0{\labeltext{[Tcut]}{[Tcut]}}
	\infer[no rule]1{\cuti{\ou{x}:!A \;[q]\; y:B}
	{\pbang{x}{z}{P}}{Q} \vdashB \Delta; \Gamma} 
	\infer1{\cuti{\ou{x}:\emptyset \;[q @ \closB(z,P) ] \; y:B}{\zero}{Q} \vdashB \Delta; \Gamma} 
\end{prooftree} \text{[Tcut$!$]}
\end{array}
$$}
\caption{Additional typing rules for $\Blang$.\label{fig:B-rules}}
\end{figure}
$$
\cuti{a:A\;[ q ]\;b:B}{P}{Q}
$$
The buffered cut construct models interaction via a ``message queue" with two polarised endpoints $a$ and $b$,
held respectively by the processes $P$ and $Q$. A polarised endpoint has the form $x$ or $\ou{x}$. The endpoint marked $\overline{x}$ is the only allowing writes,
the unmarked $y$ is the only one allowing reads, exactly one of the two endpoints is marked.
The endpoints types $A,B$  are of course
related but do not need to be exact duals, the type of the writer endpoint may be
advanced in time wrt the type of the reader endpoint, reflecting the messages already enqueued but not yet consumed. If the queue is empty, we have $A = \dual{B}$. Thus a buffered cut with empty queue corresponds to the basic cut of $\CLL$.
$$
\begin{array}{llllll}
\cuti{x:A}{P}{Q}
\equiv
\cuti{\ou x:{A}\;[ \nilm ]\;y:\dual A}{P}{\subs{y}{x}Q} &
\quad \mbox{($A+$)}  \\
\end{array}
$$
The queue $q$ stores values $V$ defined by
$$
\begin{array}{llllllllll} 
\mathit{V} & ::= &  \checkm & \mbox{(Close token)}
 & | & \;\;\labl{l} &  \mbox{(Selection Label)}\\
 & | & \clos(x,P) &  \mbox{(Linear Closure)} 
 & | & \;\;\closB(x,P) &  \mbox{(Exponential Closure)}
\end{array}
$$
$$
\begin{array}{llllllllll} 
\mathit{q} & ::= & \nilm \; | \; \mathit{V} \; | \; V @ q & & \mbox{(Queue)}
\end{array}
$$
We use $@$ to also denote (associative) concatenation operation of queues, with unit $\nilm$.
Enqueue and dequeue operations occur respectively on the lhs and rhs.

The type system $\Blang$ is obtained from $\CLL$ by replacing [TCut] with the  typing rules (and symmetric ones) in Fig.~\ref{fig:B-rules}.
We distinguish the type judgements as $P\vdash \Delta; \Gamma$ for $\CLL$ and $P \vdashB \Delta; \Gamma$ for $\Blang$.
 The [TCutB] rule sets the endpoints mode based in the cut type polarity,
 applicable whenever the queue is empty. The remaining rules relate
 queue contents with their corresponding (positive action)
 processes. For instance, rule [Tcut-$\otimes$] can be read bottom-up
 as stating that typing processes mediated by a queue containing a process
 closure $\clos(y,R)$ amounts to typing the process that will emit the
 session $y$ (bound to $R$), interacting with the queue with the
 closure removed. 
Rules [Tcut-$\oplus$] and [Tcut!] apply a similar principle to the
other possible queue contents. In [Tcut-$\one$] and [Tcut!] the write endpoint is  typed $\emptyset$, as the sender has terminated ($\zero$).
 
  \begin{figure}[t]
{\small$$
\begin{array}{ll}
\cuti{a:A[q]b:B}{Q}{P} \equiv^\mathsf{B} \cuti{b:B[q]a:A}{Q}{P}&\text{[comm]} 
\vspace{6pt}\\
\cuti{x[q]y}{P}{(\mix{Q}{R})} \equiv^\mathsf{B} \mix{(\cuti{x[q]y}{P}{Q})}{R}& \text{[CM]} 
\vspace{6pt}\\
  \cut{x[q]z}{P}{(\cuti{y[p]w}{Q}{R})} \equiv^\mathsf{B} \cuti{y[p]w}{(\cuti{x[q]z}{P}{Q})}{R}& \text{[CC]}
\vspace{6pt}\\
 \cuti{z[q]w}{P}{(\cutBi{x}{y}{Q}{R})} \equiv^\mathsf{B} \cutBi{x}{y}{Q}{(\cuti{z[q]w}{P}{R})}&  \text{[CC!]} 
\vspace{6pt}\\
    \cutBi{x}{y}{P}{(\cuti{z[q]w}{Q}{R})} \equiv^\mathsf{B} \\
    \quad\quad\quad\quad\quad\quad
    \cuti{z[q]w}{
    (\cutB{x}{y}{P}{Q})}{(\cutBi{x}{y}{P}{R}})  & \text{[D-C!]}
\end{array}
$$}
\vspace{-10pt}
 \caption{\label{fig:equivB} Additional structural congruence rules for 
 $\Blang$.}
\end{figure} 
Structural congruence for $\mathsf{B}$ (noted $\equiv^\mathsf{B}$) is
obtained by extending $\equiv$ with commutative conversions for the
buffered cut, listed in Fig.~\ref{fig:equivB}. The following provisos
apply: [CM] $y \in \fn{Q}$; in [CC] $y,z \in \fn{Q}$; in [CC!]
$x \notin \fn{P}$.  Accordingly, reduction for $\mathsf{B}$ (noted
$\to^\mathsf{B}$) is obtained by replacing the $\to$ rules [fwd],
[$\one\bot$], [$\otimes\parl]$ and [$\oplus\with]$ by the rules in
Fig.~\ref{fig:redB}. Essentially each principal cut reduction rule of
$\CLL$ is replaced by a pair of ``positive'' ($\to_p$) / ``negative''
($\to_n$) reduction rules that allow processes to interact
asynchronously via the queue, that is, positive process actions
(corresponding to positive types) are non-blocking. For example, the
rule [$\otimes$] for send appends a session closure to the tail of the
queue (rhs) and the rule for receive pops a session closure from the head
of the queue (lhs).  Notice that positive rules are enabled only if the
relevant endpoint is in write mode ($\ou{x}$), and negative rules are
enabled only if the relevant endpoint is in read mode ($y$).  In
[$\parl$] above the target cuts endpoint polarities depends on the
types of the composed processes. To uniformly express the appropriate
marking of endpoint polarities we define some convenient
abbreviations:
\begin{definition}[Setting polarities]
~$$
\begin{array}{llllll}
\cuti{a:A[\nilm]b:B}{Q}{P}^\mathsf{p} \deff 
\mbox{if }${+A}$  \mbox{ then }\cuti{\ou a:A[\nilm]b:B}{Q}{P} \\
\hspace{5.5cm} \mbox{ else }\cuti{ a:A[\nilm]\ou b:B}{Q}{P}\\
\cuti{\ou a:A[q]b:B}{Q}{P}^\mathsf{p} \deff 
\cuti{\ou a:A[q]b:B}{Q}{P} \quad (q\neq \nilm)\\
\end{array}
$$
\end{definition} 

\begin{figure}[t]
{\small$$
\begin{array}{ll}
\labeltext{[fwd]}{[to-fwd]} 
\cuti{\ou z \; [q_1] \; x}{Q}{\cut{\ou y\; [q_2]\; w}{\fwd{x}{y}}{P}} 
\to^{\mathsf{B}} \cuti{\ou z\; [q_2 @ q_1]\; w}{Q}{P} \quad\quad \;&  \text{[fwdp]}
\vspace{6pt}\\
\labeltext{[$\one\bot$]}{[to-close]} 
\cuti{\ou{x}{}\; [q]\; y}{\pone{x}}{Q} \to^\mathsf{B} \cuti{\ou{x} \; [q@\checkm] \; y}{\zero}{Q}   & \text{[$\one$]} 
\vspace{6pt}\\
\labeltext{[$\one\bot$]}{[to-close]} 
\cuti{\ou{x}\;[\mathsf{\checkmark}]\;y}{\zero}{\pbot{y}{P}} \to^\mathsf{B} P  & \text{[$\bot$]} 
\vspace{6pt}\\
\labeltext{[$\otimes \parl$]}{[to-send]} 
\cuti{\ou x \;[q] \;y}{\potimes{x}{z}{P}{Q}}{R} \to^\mathsf{B} \cuti{\ou x\;[q@\mathsf{clos}(z,P)] |\; y}{Q}{R}  & \text{[$\otimes$]}
\vspace{6pt}\\
\labeltext{[$\otimes \parl$]}{[to-recv]} 
\cuti{\ou x \; [\mathsf{clos}(z,P)@q ] \; y}{Q}{\pparl{y}{w}{R}} \to^\mathsf{B} \\\quad\quad\quad\quad\quad\quad\quad\quad\quad\cuti{\ou x\;[q]\;y}{Q}{\cuti{ z\;[\nilm]\; w}{P}{R}^\mathsf{p}}^\mathsf{p}  & \text{[$\parl$]}
\vspace{6pt}\\
\labeltext{[$\otimes \parl$]}{[to-oplusl]} 
\cuti{\ou x \;[q] \;y}{\clab{x}{P}{l}}{R} \to^\mathsf{B} \cuti{\ou x\;[q@\#\mathsf{l}] |\; y}{Q}{R}  & \text{[$\oplus$]}
\vspace{6pt}\\
\labeltext{[$\otimes \parl$]}{[to-recv]} 
\cuti{\ou x \; [\mathsf{l}@q ] \; y}{Q}{\pcasem{y}{\ell}{L}{P_\ell} } \to^\mathsf{B} 
\cuti{ x\;[q]\;y}{Q}{P_{\labl{l}}}^\mathsf{p}  & \text{[$\with$]}
\vspace{6pt}\\
\labeltext{[$\otimes \parl$]}{[to-!]} 
\cuti{\ou x \;[q] \;y}{\pbang{x}{z}{P}}{Q} \to^\mathsf{B} \cuti{\ou x\;[q@\closB (z,P)] |\; y}{\zero}{Q}  & \text{[$!$]}
\vspace{6pt}\\
\labeltext{[$\otimes \parl$]}{[to-?]} 
\cuti{\ou x \;[\closB (y,P)] \;y}{\zero}{\pwhy{y}{Q}} \to^\mathsf{B} 
\cutBi{x}{y}{P}{Q}  & \text{[$?$]}
\end{array}
$$}
 \caption{\label{fig:redB} Reduction $P \to^\mathsf{B} Q$.}
\end{figure} 
The following definition then formalizes the  intuition  given above about how to
encode processes of 
$\CLL$ into processes of $\Blang$.
\begin{definition}[Embedding]
Let $P\vdash\Delta;\Gamma$.
$P^\dagger$ is the $\mathsf{B}$ process such that
$$
\begin{array}{llllll}
(\cuti{x:A}{P}{Q})^\dagger
\deff
\cuti{x:{A}\;[ \nilm ]\;y:\dual A}{P^\dagger}{(\subs{y}{x}Q)^\dagger}^\mathsf{p} \\
\end{array}
$$
 homomorphically defined in the remaining constructs. Clearly $P^\dagger \vdash^{\mathsf{B}} \Delta; \Gamma$.

\end{definition}

\subsection{Preservation and Progress for $\Blang$ }
\label{sec:safetyCLLB}

In this section, we prove basic safety properties of $\Blang$: Preservation (Theorem~\ref{theorem:type-preservation}) and Progress (Theorem~\ref{theorem:progress}). To reason about type derivations involving buffered cuts, we formulate some auxiliary 
inversion principles that allow
us, by aggregating sequences of application of
 [TCut-$*$] rules of $\Blang$,  to talk in a uniform way about typing of values in queues and typing
of processes connected by queues. To assert typing of queue values $c$ we use judgments the
form $\Gamma;\Delta \vdash  c\ass E$, where $E$ is a either a type or a one hole type context,
defined by
$$
E ::= \Box \; | \; T \; | \; T\parl E \; | \; \withm{\ell}{L}{E_\ell} 
$$
where in $\withm{\ell}{L}{E_\ell}$ only branch type
 $E_{\labl{l}}$ for some selected label $\labl{l}\in L$ is a one hole context (to plug the continuation type);
 only the branch chosen by the selected label in a queue is relevant to type  next queue values.
We identify the selected branch in the type by tagging it
with the corresponding label $\labl{l}$ thus
$\withm{\ell}{L}{E_\ell} [\labl{l}]$.
We then introduce the following typing rules for queue values.
\begin{definition}[Typing of Queue Values]

  {\small$$
\begin{array}{cccccccccccc}
\displaystyle
\frac{}{\Gamma; \vdash_{\mathtt{val}} \checkm:\bot} &  \quad\quad &
\displaystyle
\frac{P \vdashB \Delta,z:T; \Gamma}{\Gamma;\Delta\vdash_{\mathtt{val}} \clos(z,P):\dual T\parl E}
\vspace{4pt}\\
\displaystyle
\frac{\withm{\ell}{L}{E_\ell} [\labl{l}]}{\Gamma;\Delta\vdash_{\mathtt{val}} \labl{l}:\withm{\ell}{L}{E_\ell} }
&  \quad\quad &
\displaystyle
\frac{ P\vdashB z:A; \Gamma}{\Gamma;\vdash_{\mathtt{val}} 
\closB(z,P):?\dual A}
\end{array}
$$}
\end{definition}
Given a sequence of $k$ one hole queue value types $E_i$ and a type $A$, we
denote by $\mathbf{E}_k;A$ the type $E_1[E_2[...E_k[A]]]$. 
Queue value types allow us to talk in a uniform way about the type a receiver processes compatible with the types of enqueued values, as characterized by the following Lemma \ref{lemma:queue-prf} and Lemma \ref{lemma:queue-full}.

\begin{restatable}[Non-full]{lemma}{queueprf}\label{lemma:queue-prf} 
For $P\neq\zero$ the rule below is admissible and invertible:
$$
\frac{
P \vdashB \Delta_P,x{:}A;\Gamma
\quad
Q \vdashB \Delta_Q, y{:}B;\Gamma
\quad q = \overline{c_k}
\quad B=\mathbf{E}_k;\dual{A}
\quad
\Gamma;\Delta_i \vdash c_i{:} E_i
	\quad {-}B	} 
{\cuti{\ou{x}{:}A \;[q] \; y{:}B}{P}{Q}
	\vdashB \Delta_P,\Delta_Q, \Delta_1,...,\Delta_k;\Gamma} 
$$
\hide{
Let $\cut{\ou{x}:A \;[\overline{c_k}] \; y:B}{P}{Q} \vdashB \Delta; \Gamma$,
with $c_i = \mathsf{clos}(z_i;R_i;\Delta_{R_i})$ where $0\leq k$. Then 

(1)
$\cut{\ou{x}:T_1\otimes \ldots \otimes T_k \otimes A \;[\nilm] \; y:B}{s_1;\cdots ; s_k;P} {Q}\vdash \Delta; \Gamma $ and $s_i = \pi(c_i)$.

(2) 
$\Delta=\Delta_{R_1},\ldots, \Delta_{R_k},\Delta'$ and 
$R_i\vdash \Delta_{R_i},z_i:T_i$ and $B=\dual T_1\parl\cdots \parl \dual T_k \parl \dual A$.
}
\end{restatable}

\hide{
\begin{proof}
To check that the rule is admissible, directly derive the
conclusion using $k$ applications of [TCut$\otimes$] or [TCut$\oplus$].
To check inversion, we proceed by
induction in the derivation of
$\cuti{\ou{x}:A \;[\overline{c_k}] \; y:B}{P}{Q} \vdashB \Delta; \Gamma$.
 
  (Case  [TCutB]) We have $k=0$ so $q=\nilm$, $A$ positive and $B$ negative. 

 (Case  [TCut-$\one$]) 
 Not applicable, since $c_k\neq \checkmark$.
 
(Case [TCut-$\otimes$])
We have $\cuti{\ou{x}{:}A \;[\overline{c_{k-1}} @ \mathsf{clos}(z_k;R_k)   ] \; y{:}B}{P}{Q} \vdashB \Delta; \Gamma$ derived from
$\cuti{\ou{x}:T_k{\otimes} A \;[  \overline{c_{k-1}} ]\; y:B}{\potimes{x}{z_k}{R_k}{P}}{Q} \vdashB \Delta; \Gamma$.
By i.h. $\potimes{x}{z_k}{R_k}{P}\vdashB \Sigma_{k},x{:}T_k\otimes A$ and
$Q \vdash \Delta_Q, y{:}B$ where $B=\mathbf{E}_{k-1};{\dual { (T_k\otimes {A}) }}$
and 
$\Delta_i \vdashB c_i : E_i $ for $i\leq k-1$.
By inversion ([T$\otimes$])
$P \vdashB \Delta_P,x:A$,
 $\Sigma_{k} = \Delta_k,\Delta_P$ and 
 $R_k \vdashB \Delta_k, z_k\ass T_k$, so
 $\Delta_k \vdashB c_k \ass E_k$ where $E_k=\dual{T_k}\parl \Box $.
We conclude $P\vdashB \Delta_{P},x:A$,
$Q \vdash \Delta_Q, y\ass B$,
$B= \mathbf{E}_k ;\dual{A}$
and 
$\Delta_i \vdash c_i \ass E_i $ for $i\leq k$.

(Case  [TCut-$\oplus$]) 
We have $\cuti{\ou{x}{:}A_{\labl{l}} \;[\overline{c_{k-1}} @ \labl{l}   ] \; y{:}B}{P}{Q} \vdashB \Delta; \Gamma$ derived from
$\cuti{\ou{x}: \oplusm{\ell}{L}{A_\ell} \;[  \overline{c_{k-1}} ]\; y:B}{{\labl{l}\ x;P}}{Q} \vdashB \Delta, \Delta'; \Gamma$.
We proceed by the i.h. as with [T$\otimes$] above, considering
$B=\mathbf{E}_{k-1};{\dual { (\oplusm{\ell}{L}{A_\ell}) }}$
and
$\Delta_k \vdashB c_k : E_k$ where $E_k=\withm{\ell}{L}{E_\ell} $ and $E_{\labl{l}} = \Box$,
so that $B= \mathbf{E}_k ;\dual{A}$
and 
$\Delta_i \vdash c_i : E_i $ for $i\leq k$.

\end{proof}}
\hide{
\begin{lemma}[Queue prefix]\label{lemma:queue-prf} 
The proof rule below is admissible and invertible:
$$
\frac{
P \vdash \Delta_P,x{:}A
\quad
Q \vdash \Delta_Q, y{:}B
\quad q = \overline{c_k}
\quad B=\dual T_1{\parl}{...}{\parl} \dual T_k{\parl} \dual{A}
\quad
\Delta_i \vdash c_i{:} T_i {\parl} ()
	\quad {-}B	} 
{\cuti{\ou{x}{:}A \;[q] \; y{:}B}{P}{Q}
	\vdash \Delta_P,\Delta_Q, \Delta_1,...,\Delta_k} 
$$
\hide{
Let $\cut{\ou{x}:A \;[\overline{c_k}] \; y:B}{P}{Q} \vdash \Delta; \Gamma$,
with $c_i = \mathsf{clos}(z_i;R_i;\Delta_{R_i})$ where $0\leq k$. Then 

(1)
$\cut{\ou{x}:T_1\otimes \ldots \otimes T_k \otimes A \;[\nilm] \; y:B}{s_1;\cdots ; s_k;P} {Q}\vdash \Delta; \Gamma $ and $s_i = \pi(c_i)$.

(2) 
$\Delta=\Delta_{R_1},\ldots, \Delta_{R_k},\Delta'$ and 
$R_i\vdash \Delta_{R_i},z_i:T_i$ and $B=\dual T_1\parl\cdots \parl \dual T_k \parl \dual A$.
}
\end{lemma}
\begin{proof}
To check that the rule is admissible, directly derive the
conclusion using $k$ applications of [TCut-$\otimes$].
To check inversion, we proceed by
induction in the derivation of
$\cuti{\ou{x}:A \;[\overline{c_k}] \; y:B}{P}{Q} \vdash \Delta; \Gamma$.
 
  (Case of [TCutB]) We have $k=0$ so $q=\nilm$, $A$ positive and $B$ negative. 

 (Case of [TCut-$\one$]) Not applicable, since $c_k\neq \checkmark$.

(Case of [TCut-$\otimes$])
We have $\cut{\ou{x}:A \;[\overline{c_{k-1}} @ \mathsf{clos}(z_k;R_k;\Delta_{R_k})   ] \; y:B}{P}{Q} \vdash \Delta; \Gamma$ derived from
$\cut{\ou{x}:T_k{\otimes} A \;[  \overline{c_{k-1}} ]\; y:B}{\potimes{x}{z_k}{R_k}{P}}{Q} \vdash \Delta, \Delta'; \Gamma$.
By i.h. $\potimes{x}{z_k}{R_k}{P}\vdash \Sigma_{k},x:T_k\otimes A$,
$Q \vdash \Delta_Q, y:B$ where $B=\dual T_1\parl\ldots \parl \dual {(T_k\otimes {A})}$
and 
$\Delta_i \vdash c_i : T_i\parl () $ for $i\leq k-1$.
By inversion ([T$\otimes$])
$P \vdash \Delta_P,x:A$,
 $\Sigma_{k} = \Delta_k,\Delta_P$ and 
 $R_k \vdash \Delta_k, z_k:T_k$, so
 $\Delta_k \vdash c_k : T_k\parl () $.
We conclude $P\vdash \Delta_{P},x:A$,
$Q \vdash \Delta_Q, y:B$,
$B=\dual T_1\parl\ldots \parl \dual T_k\parl \dual{A}$
and 
$\Delta_i \vdash c_i : T_i\parl () $ for $i\leq k$.

(Case of [TCut-$\oplus$]) Similar to [T$\otimes$].

\hide{
By i.h.
$\cut{\ou{x}:T_1\otimes \ldots T_k \otimes A \;[\nilm] \; y:B}{s_1;\ldots ;s_{k-1};s_k;P} {Q}\vdash \Delta; \Gamma $.

(2) By (1) and inversion on the typing derivation of  
$s_1;\cdots; s_k;P \vdash x:A,\Delta_R$.
}

\end{proof}
}
Notice that a session type,
as defined by a $\CLL$ proposition, may terminate in either $\one$,
$\bot$ or an exponential type $!A/?A$. We then also have

\begin{restatable}[Full]{lemma}{queuefull}\label{lemma:queue-full} 
The proof rules below are admissible:
$$
\frac{
Q \vdashB \Delta_Q, y:B;\Gamma
\quad
\Gamma;\Delta_i \vdash c_i: E_i 
\quad B = \mathbf{E}_k;\bot
\quad c_k = \checkm
\quad
{-}B
	} 
{\cuti{\ou{x}:\emptyset \;[\overline{c_k}] \; y:B}{\zero}{Q}
	\vdashB \Delta_Q, \Delta_1,\ldots,\Delta_k;\Gamma} 
$$
$$
\frac{
Q \vdashB \Delta_Q, y\ass B;\Gamma
\quad
\Gamma;\Delta_i \vdash c_i\ass E_i 
\quad B = \mathbf{E}_{k-1};C
\quad \Gamma\vdash c_k = \closB(z,R){:}C
\quad
{-}B
	} 
{\cuti{\ou{x}:\emptyset \;[\overline{c_k}] \; y:B}{\zero}{Q}
	\vdashB \Delta_Q, \Delta_1,\ldots,\Delta_k,\Gamma} 
$$
Moreover, one of them must apply for inverting the judgment in the conclusion.
\end{restatable}
\hide{
\begin{proof}
We have
$\cuti{\ou{x}:\emptyset \;[\overline{c_k}] \; y:B}{\zero}{Q} \vdashB \Delta; \Gamma$
derived from [TCut$\one$] (a) or [TCut$!$] (b). In case (a)
$\cuti{\ou{x}:\one \; [{\overline{c_k}}]\;  y:B}{\pone{x}}{Q} \vdashB \Delta; \Gamma$
with $c_k=\checkm$. We conclude by Lemma \ref{lemma:queue-prf} .
 In case (b)
$\cuti{\ou{x}:!A \; [{\overline{c_k}}]\;  y:B}{\pbang{x}{z}{R}}{Q} \vdashB \Delta; \Gamma$ with $c_k=\closB(z,R)$.
By Lemma \ref{lemma:queue-prf}, $\pbang{x}{z}{R}\vdash z:!A;\Gamma$, so $\Gamma; \vdash c_k=\closB(z,R):?\dual A$ and we conclude 
with $B = \mathbf{E}_{k-1};C$ and
 $C=?\dual A$.
\end{proof}
}
\begin{restatable}[Non-empty]{lemma}{queuenonempty}
\label{lemma:queue-non-empty} 
Let $\cuti{\ou{x}\ass A \;[q] \; y\ass B}{P}{Q} \vdashB \Delta; \Gamma$.

If $A$ is negative or void, then $q\neq\nilm$.
\end{restatable}
\begin{proof}
If $A$ is void then immediate by Lemma \ref{lemma:queue-full}.
If $-A$, suppose $q=\nilm$. By Lemma~\ref{lemma:queue-prf}  $k=0$ and  $\dual A = B$. But then $B$ positive, contradiction. So $q\neq\nilm$.
\end{proof}

\begin{restatable}[Preservation]{theorem}{typepreservation}
  \label{theorem:type-preservation} 
	Let $P \vdash^{\mathsf{B}} \Delta; \Gamma$. 
	
	(1) If $P \equiv^\mathsf{B} Q$, then $Q \vdash^{\mathsf{B}} \Delta; \Gamma$. (2) If $P \to^\mathsf{B} Q$, then $Q \vdash^{\mathsf{B}} \Delta; \Gamma$. 
\end{restatable}
\begin{proof}
We verify that rules for  $\stackrel{\mathsf{\tiny B}}{\equiv}$ (Fig.~\ref{fig:equivB})  (resp.  $\stackrel{\mathsf{B}}{\to}$ (Fig.~\ref{fig:redB})) are type preserving. We illustrate with  (2) [fwdp].

(Case [fwdp])
$P = \cuti{\ou z\ass A \; [q_1] \; x\ass \dual B}{Q'}{\cuti{\ou y\ass B\; [q_2]\; w\ass C}{\fwd{x}{y}}{P'}}\vdashB\Delta;\Gamma$.
If $q_2=\nilm$ and $B = \dual{C}$,
hence $ Q=\cuti{\ou z \ass A\; [q_1]\; w\ass C}{Q'}{P'} \vdashB\Delta;\Gamma$.

Otherwise $q_2\neq\nilm$. Let
$F_2 = \cuti{\ou y\ass B\; [q_2]\; w\ass C}{\fwd{x}{y}}{P'}
$ where  $F_2 \vdashB \Delta_2, x\ass \dual B$ and $\Delta =\Delta_1, \Delta_2$,
and $q_2=\overline{c_k}$.
By Lemma \ref{lemma:queue-prf},
$\fwd{x}{y} \vdashB x\ass \dual B, y\ass B$, 
and $C = \mathbf{E}_k ; \dual B$, $\Gamma;\Delta_i\vdash c_i\ass E_i$,
$P\vdash w\ass C,\Delta_P;\Gamma$ and $\Delta_2 = \Delta_1,\ldots,\Delta_k,\Delta_P$.

Let 
$F_1 =  \cuti{\ou z \ass A\; [q_1]\; x\ass \dual B}{Q'}{F_2}$
 with $F_1 \vdashB \Delta_1, x\ass B;\Gamma$ and
$q_1=\overline{d_l}$.
By Lemma \ref{lemma:queue-prf},
$Q' \vdashB\Delta_Q,z\ass A;\Gamma$, $\Gamma;\Delta'_j \vdash d_j\ass F_j$ and $\dual B = \mathbf{F}_l ; \dual A$
and $\Delta_1 = \Delta'_1,\ldots,\Delta'_l,\Delta_Q$.
By Lemma \ref{lemma:queue-prf} we get 
$\cuti{z\ass A \; [ q_2@q_1 ] \; w\ass C}{Q'}{P'}\vdashB \Delta;\Gamma$.
\end{proof}

\hide{
---

Let $q=c_2@ \ldots @ c_k$ where $c_i = \mathsf{clos}(z_i;R_i;\Delta_{R_i})$.

By Lemma \ref{lemma:queue-prf},
$\cut{\ou x:T_1\otimes \cdots \otimes T_k \otimes A \; [\nilm] \; y:B}
{s_1;\cdots; s_k;P}{\pparl{y}{w}{Q}} \vdash\Delta$,
where $R_i\vdash \Delta_{R_i},z_i:T_i$.
By inversion,  $s_1;\cdots s_k;P\vdash x:T_1\otimes \cdots \otimes T_k
\otimes A, \Delta_{R_1},\cdots, {\Delta_{R_k}}$, and $s_2;\cdots ; s_k;P\vdash x:T_2\otimes \cdots \otimes T_k
\otimes A, \Delta_{R_2},\cdots, {\Delta_{R_k}}$. 

By inversion,  $Q \vdash w:\dual T_1, y:B',\Delta_Q$
with $C=T_2\parl \cdots \parl T_k
\parl  \dual A$.

By ]TCutB], 
$\cut{\ou x:T_2\otimes \cdots \otimes T_k \otimes A \; [\nilm] \; y:C}
{s_2;\cdots; s_k;P}{Q} \vdash\Delta, w:T_1$.

By [TCutB] $\cut{z_1:T_1\;[\nilm]\; w:\dual T_1}{R_1}{Q}\vdash y:C,\Delta_Q$.

By [TCut-$\otimes$] 
($k-1$ times), 
$\cut{\ou x:A\;[q]\;y:C}{P}{(\cut{z_1:T_1\;[\nilm]\; w:\dual T_1}{R}{Q})}\vdash \Delta$.

}

A process $P$ is \emph{live} if and only if $P = \mathcal C[Q]$, for some static context $\mathcal C$ (the hole lies within the scope of static constructs mix, cut) and $Q$ is an action process.  We first show that a live process 
either reduces or offers an interaction on a free name.
The observability predicate defined in Fig. \ref{fig:obs} (cf. \cite{sangiorgi-walker:book}) characterises interactions of a process with the environment.

\begin{figure}[t]
{\small$$
\begin{array}{c}
\displaystyle
\begin{prooftree}
	\infer0[\text{[fwd]}]{\obs{\fwd{x}{y}}{x}}
\end{prooftree} 
\quad
\begin{prooftree}
	\infer[no rule]0{s(\mathcal A) = x} 
	\infer1[\text{[$\mathcal A$]}]{\obs{\mathcal A}{x}}
\end{prooftree} 
\quad
\begin{prooftree}
	\infer[no rule]0{P \equiv Q \quad \obs{Q}{x}} 
	\infer1[\text{[$\equiv$]}]{\obs{P}{x}} 
\end{prooftree} 
\quad 
\begin{prooftree}
	\infer[no rule]0{\obs{P}{x}}
	\infer1[\text{[mix]}]{\obs{(\mix{P}{Q})}{x}}
\end{prooftree} 
\vspace{2pt}\\
\begin{prooftree}
	\infer[no rule]0{\obs{P}{x} \quad x\neq y}
	\infer1[\text{[cut]}]{\obs{(\cut{y [q] x}{P}{Q})}{x}} 
\end{prooftree} 
\quad
\begin{prooftree}
	\infer[no rule]0{\obs{Q}{x} \quad x \neq y} 
	\infer1[\text{[cut!]}]{\obs{(\cutB{y}{z}{P}{Q})}{x}}
\end{prooftree} 
\quad
\end{array}
$$}
 \caption{\label{fig:obs} Observability Predicate $\obs{P}{x}$.}
\end{figure}

\begin{restatable}[Liveness]{lemma}{liveness}\label{lemma:liveness} 
	Let $P \vdashB \Delta; \Gamma$ be live. Either $\obs{P}{x}$
	or $P\to^\mathsf{B}$. 
 \end{restatable} 

\begin{proof}
By induction on the derivation for $P \vdash \Delta; \Gamma$, and case analysis on the last typing rule. We illustrate with one rule. 
 
(Case of [TCut$\otimes$])
We have $P=\cuti{\ou{x}:A \;[q  ] \; y:B}{P_1}{P_2} \vdashB \Delta; \Gamma$ 
where $q=r @ \mathsf{clos}(z_k;R_k)  $
derived from
$P_1 \vdashB x:A,\Delta_{P_1}; \Gamma$.
By  i.h. $P_1\to$ or $\obs{P_1}{x_1}$. 
If $x\neq x_1$ then $\obs{P_1}{x_1}$ and $\obs{P}{x_1}$.
If $A$ is positive, by the same reasoning as above for $P_2$
we conclude that $P\to$ or $\obs{P}{w}$.

Otherwise, $A$ is negative.
By  i.h. $P_2\to$ or $\obs{P_2}{z}$. 
If $z\neq y$ then $\obs{P}{x}$.
Otherwise, $\obs{P_2}{y}$.
By Lemma \ref{lemma:queue-prf}, $B = \dual T\parl C$.
By Lemma \ref{lemma:barbs} (2,6), 
either $P_2 \equiv \cut{*}{\fwd{y}{v}}{Q'}$ (a) or
$P_2\equiv \cut{*}{\pparl{y}{w}{Q'}}{Q''}$ (b).
For case (a) we conclude as in [TCutB] above that $P\to$ or $\obs{P}{v}$,
for case (b) $P\to$ by [$\parl$].
\end{proof}

\begin{theorem}[Progress]\label{theorem:progress}
	Let $P \vdash \emptyset; \emptyset$ be a live process. Then, $P\to^\mathsf{B}$.
\end{theorem} 
\begin{proof}
		Follows from Lemma~\ref{lemma:liveness} since $\Delta =\emptyset$. 
\end{proof} 

\subsection{Correspondence between $\CLL$ and $\Blang$}\label{sec:CLLCLLB}

In this section we establish the correspondence between reduction in $\CLL$ and $\Blang$, proving that the two languages simulate each other in a tight sense.
Intuitively, the correspondence shows that $\Blang$ allows some
positive actions to be buffered ahead of reception, while in $\CLL$ a single 
positive action synchronises with the corresponding dual in one step,
or a forward reduction takes place. 

We write a reduction $P\tobb Q$ as $P\tobx{} Q$ if the 
reduced action is positive, $P\tobnx{} Q$ if the 
reduced action is negative (we consider [call] negative), $P\toba{} Q$ if the 
reduced action is a forwarder, and $P\toban{} Q$ if the 
reduced action is positive or a forwarder. We also write
$P\tor{}Q$ for positive action followed by a matching negative action on the same cut
with an initially empty queue.
\begin{restatable}{lemma}{commutemain}
  \label{lemma:commute-main}
The following commutations of reductions hold.
\begin{enumerate}
\item
Let $P_1 \tobx{} S \tobn P_2 $. 
Either $P_1 \tor{} P_2 $, 
or $P_1  \tobn  S'   \tobx{} P_2$ for some $S'$.
\item
Let $P_1 \toba{} S \tobn P_2 $. Then $P_1  \tobn  S'   \toba{} P_2$
for some $S'$.
\item
If $P_1 \toban{} S \tobn P_2 $, either $P_1 \tor{} P_2 $, 
or $P_1  \tobn  S'   \toban{} P_2$ for some $S'$.
\item
Let $P_1 \toban{} N \tobaxm{\epsilon} S \tor{} P_2 $.
Either $P_1 \tobax{\epsilon} N$ or  
$P_1 \tor{} S' \toban{}  P_2 $ for some $S'$.
\end{enumerate}
\end{restatable}
\vspace{-10pt}
\begin{restatable}[Simulation]{lemma}{simtrivial}
\label{lemma:sim-trivial} Let $P\vdash \emptyset;\emptyset$.
If $P \to Q$ then $P^\dagger \tob Q^\dagger$.
\end{restatable}
\begin{proof}
Each cut reduction of $\CLL$ is either simulated 
by two reduction steps of $\mathsf{B}$ in sequence or by a [fwd] reduction.
\end{proof}

The following lemma identifies that in $\Blang$, a sequence of
positive actions (or forwards) followed by a negative action can
always be commuted either by pulling out the negative action first,
followed by the sequence of positive actions and forwards; having the
negative action follow a positive action on the same channel and then
performing the remaining actions; or by first performing a sequence of
forward actions, the output and input on the relevant session and then
the remaining actions.
\vspace{-2pt}
\begin{restatable}[Simulation]{lemma}{simul}
 \label{lemma:simul}
Let $P \vdashB \emptyset; \emptyset$. 
If $P \tobanm{} \tobn{} Q$ then
(1) $ P \tobn R$ and $R \tobanm{} Q$ for some $R$, or;
(2) $ P \tor{} R$ and $R \tobanm{}  {\;}  Q$ for some $R$, or;
(3) $ P \tobaxm{\epsilon} \tor{} R $ and $R \tobanm{} Q$ for some $R$.
%
\end{restatable}
\begin{proof}
By induction on $P (\toban{})^*  P'$.
\end{proof}


\vspace{-4pt}
\begin{restatable}[Operational correspondence $\CLL$-$\Blang$]
  {theorem}{cllb}
  \label{teo:cll-b}
  Let $P \vdash \emptyset; \emptyset$. 
\begin{enumerate}
\item If $P\Rightarrow R$ then $P^\dagger\Rightarrow^\mathsf{B} R^\dagger$.
\item
If $P^\dagger (\tobanm{} \tobn)^* Q$ then there is $R$ such that
$P\Rightarrow R$ and $R^\dagger \tobanm{} Q$.
\end{enumerate}
\end{restatable}

 Due to the progress property for $\Blang$ (Theorem ~\ref{theorem:progress}) and because queues are bounded by the size of positive/negative sections in types, after a sequence of positive or forwarder reductions a negative reduction 
 consuming a queue value must occur.
Theorem~\ref{teo:cll-b}(2) states that every reduction sequence in $\Blang$ is simulated by a reduction sequence in $\CLL$ up to some
anticipated forwarding and buffering of positive actions. Our results imply that every reduction path in $\Blang$ maps to a reduction path in $\CLL$ in which every negative reduction step in the former 
is mapped, in order, to a cut reduction step in the latter.

\vspace{-10pt}
\section{Correctness of the core SAM}\label{sec:soundcore}
We now prove that every execution trace of the core SAM
defined in Fig.~\ref{fig:rules} represents a correct process
reduction sequence $\Blang$ (and therefore of $\CLL$, in the light of
Theorem~\ref{teo:cll-b}), first for the language without
exponentials and mix, which will be treated in Section~\ref{sec:sam-mix}. 
In what follows, we annotate
endpoints of session records with their types (e.g.~as
$\srecs{x{:}A}{q}{P}{y{:}B}$), these annotations are not needed to
guide the operation of the SAM, but convenient for the
proofs; they will be omitted when not relevant or are obvious
from the context. We first define a simple encoding of well-typed $\Blang$
processes to SAM states.

\newcommand{\expa}{\stackrel{\mathsf{cut*}}{\Mapsto}}

\begin{definition}[Encode]
\label{def:encode}
Given $P\vdashB\emptyset$
we define $enc(P)=\C$ as $enc(P,\emptyset) \stackrel{\mathsf{cut}*}{\Mapsto} \mathcal \C$ where $enc(P,H) \expa \mathcal \C$ is defined by the rules
{\small$$
\begin{array}{clll}
\displaystyle
\frac{
enc(P(x),H[\srecs{ x\ass A}{q}{Q}{y\ass B}])
\expa \mathcal C}{
enc(\cuti{\ou x\ass A [q]\; y\ass B]}{P}{Q}, H)\expa \mathcal C}
& \mbox{\rm ($A+$ )}
\vspace{6pt}\\
\displaystyle
\frac{
enc(Q(y),H[\srecs{ x\ass A}{q}{P}{y\ass B}] )
\expa \mathcal C
}{
enc(\cuti{\ou x\ass A\; [q]\; y\ass B ]}{P}{Q}, H)\expa \mathcal C}
\quad \quad & \mbox{\rm ($A-$ or $P=\zero$ )}
\vspace{6pt}\\
enc(\mathcal{A},H) 
\expa 
(\mathcal{A},H) & \mbox{\rm ($A \in \mathcal A$)}
\end{array}
$$}
\end{definition}
Notice that $enc(P)$ maximally applies  the SAM execution rule for cut to $(P,\emptyset)$ until an action is reached. Clearly, for any $P\vdashB \emptyset$, if $enc(P) = \C$ then $P \Mapsto^* \C$. Also, 
if all cuts in a state $C$ have empty queues then there is  
a process $Q$ of $\CLL$
such that $enc(Q^\dagger) = C$.
We then have

\begin{restatable}[Soundness wrt $\Blang$]{theorem}{soundcllb}
\label{soundcllb}
Let $P\vdashB\emptyset$.

\noindent
If $enc(P) \Mapsto D \expa \mathcal{C}$  then there
is $Q$ such that $P \to\cup\equiv Q$ and
$\mathcal{C} = enc(Q)$. 
\end{restatable}
\begin{proof}
Let $enc(P) \Mapsto \mathcal B$.
Let $enc(P) = (\mathcal{A},H)$ for some action $\mathcal{A}$.
Then  $P \equiv E[\mathcal{A}]$
for some cut context $E[-]$. We illustrate one case.
\hide{
(Case of [Sfwd]) 
Let $enc(P) =(\fwd{x}{y}, H[\srecs{z}{q_1}{P_1}{x}][\srecs{y}{q_2}{P_2}{w}]) \Mapsto \B$ where $\B=(P_2, H[\srecs{z}{q_2 @ q_1}{P_1}{w}])$.

Therefore $P \equiv E[\fwd{x}{y}] \equiv E'[\cuti{\ou z \; [q_1]\; x]}{P_1}{
\cuti{\ou y \; [q_1]\; w]}{\fwd{x}{y}}{P_2}
}]$
and $P\to Q = E'[\cuti{\ou z \; [q_2@q_1]\; w]}{P_1}{P_2}]$ and $\B \expa \C$
and $enc(Q) = \C$.}

(Case of [S$\one$]) 
Let $enc(P) =(\pone{x}, H[\srecs{x}{q}{R}{y}]) \Mapsto 
(R, H[\srecs{x}{q @ \checkmark}{\zero}{y}]) = D \expa C$.
Therefore, $P \equiv E[\pone{x}] \equiv E'[\cuti{\ou x \; [q]\; y]}{\pone{x}}{R}]$
and $P\to Q = E'[\cuti{\ou x \; [q@\checkmark]\; y]}{\zero}{R}]$
and $enc(Q) = C$.



\end{proof}

We can then combine soundness with the operational correspondence
between $\CLL$ and $\Blang$ (Theorem~\ref{teo:cll-b}) to obtain an overall soundness result for
the SAM with respect to $\CLL$: 

\vspace{-2pt}
\begin{restatable}[Soundness wrt $\CLL$]{theorem}{soundcll}
  \label{soundcll}
Let $P\vdashB\emptyset$. 

\noindent
1. If $enc(P) \stackrel{*}{\Mapsto} \expa C$ 
there
is $Q$ such that $P \Rightarrow\cup\equiv Q$ and
$\C = enc(Q)$. 

\noindent
2. Let $P\vdash\emptyset$. 
If $enc(P^\dagger) \stackrel{*}{\Mapsto} enc(Q^\dagger)$ then $P\Rightarrow Q$.
\end{restatable}

In Definition~\ref{def:ready} we identify readiness, the fundamental
invariant property of SAM states, key to prove progress of its
execution strategy.  Readiness means that any running process holding
an endpoint of negative type, and thus attempting to execute a
negative action (e.g., a receive or offer action) on it, will always
find an appropriate value (resp. a closure or a label) to be
read in the appropriate session queue. No busy waiting or context
switching will be necessary since the sequential execution semantics
of the SAM enforces that all actions corresponding to a positive
section of a session type have always been enqueued by the ``caller"
process before the "callee" takes over.  As discussed in
Section~\ref{sec:coreSAM} it might not seem obvious whether all such input
endpoints, (including endpoints moved around via send / receive
interactions), always refer to non-empty queues.

Readiness must also be maintained by processes
suspended in session records, even if a suspended process waiting on a read
endpoint will not necessarily have the corresponding queue already
populated. Intuitively, a process $P$ is $(H,N)$-ready if all its ``reads" in the 
input channels (except those in $N$) will be matched 
by values already stored in the corresponding session queue. 

\begin{definition}[Ready]\label{def:ready}
~Process $P$ is $H,N$-ready if for all $y\in \fn{P}\setminus N$ and $\srecs{{x:A}}{q}{R}{{y}}\in H$ then $A$ 
is negative or void.
We abbreviate $H,\emptyset$-ready by $H$-ready.
Heap $H$ is ready if, for
all $\srecs{x}{q}{R}{y}\in H$, the following conditions hold:
\begin{enumerate}
\itemsep 0pt
\item if $R(y)$ then $R$ is $H,\{y\}$-ready 

\item if $R(x)$ then $R$ is $H$-ready  

\item if $\clos(z:-,R)\in q$, $R$ is $H,\{z\}$-ready.

\item if $\clos(z:+,R)\in q$, $R$ is $H$-ready.
\end{enumerate}


%
\noindent 
State $C=(P,H)$ is ready if $H$ is ready and $P$ is 
$H$-ready.
\end{definition}

\vspace{-.2cm}

\begin{restatable}[Readiness]{lemma}{ready}
\label{lemma:ready}~Let $P\vdash\emptyset$ and $(P,\emptyset) \stackrel{*}{\Mapsto} S$.
Then $\mathcal S$ is ready.
\end{restatable}
\begin{proof}
The property trivially holds for $(P,\emptyset) $ for
$\fn{P}=\emptyset$ and $H=\emptyset$.
We proceed by transition induction, assuming that  
$S=(P,H) $ is ready, and $\mathcal S \Mapsto \mathcal 
S'= (P,H')$, we check that $S'$ is ready.
We illustrate one case.
 
(Case of [SCut])
Let $R = \cuti{\ou x:T \; [q] \; y:B}{P}{Q}$
and
$S = (R,H) \Mapsto  S' =
( P, H')$, where
$H'=H[\srecs{x}{\nilm}{Q}{y} ]$ and $T$ positive.
By i.h., $R$ is $H$-ready hence $P(x)$ is $H,\{x\}$-ready
and $Q(y)$ is $H,\{y\}$-ready. Then $P(x)$ is $H'$-ready, and $H'$ is ready. We conclude that $S'$ is ready.
\end{proof}

\vspace{-.2cm}

\begin{restatable}[Progress]{theorem}{samprogress}
\label{samprogress}
Let $P\vdashB\emptyset$ and $P$ live. Then $enc(P)\Mapsto S'$.

\end{restatable}
\begin{proof}
Since $P$ is live then $enc(P) =  S = ({\mathcal A},H)$ for some action $\mathcal A$.
\hide{
(Case of [fwdB]) $\mathcal A = \fwd{x}{y}$.
Wlog, assume that $x{:}A$ is negative, so $y{:}\dual A$ is positive.
 By typing, $x$ and $y$ must be endpoints of different cuts (and session records).
Case $H=H'[\srecs{z}{q_1}{U(z)}{x}][\srecs{y}{q_2}{V(w)}{w}]$.
We have
$S\Mapsto S'$ 
by [Sfwd].
Case $H=H'[\srecs{x}{q_1}{U(z)}{z}][\srecs{y}{q_2}{V(w)}{w}]$.
Then $\fwd{x}{y} = \mathcal A^-(x)$, and  $S\Mapsto S'$ by [S$-$].
}

(Case of [$\bot$]) 
$\mathcal A = \pbot{y}{R'}$. 
If $H=H'[\srecs{y{:}A}{q}{Q}{x}]$,
then $S\Mapsto S'$ by [S$-$].
Otherwise, $H=H'[\srecs{x\ass A}{q}{\mathcal A}{y}]$. 
By Lemma~\ref{lemma:ready}, $\mathcal A$ is $H'$-ready,
so $A$ is negative or void.  By Lemma~\ref{lemma:queue-non-empty}  $q\neq\nilm$. By Lemma~\ref{lemma:queue-full}, $q=\checkmark$. So $S\Mapsto S'$ by [S$\bot$].  
\end{proof}


\vspace{-10pt}
\section{{The SAM for full $\CLL$}}
\label{sec:sam-mix}
\begin{figure}[t]
{\small$$
\begin{array}{llll}
S & ::= & (\E, P,  H) & \mbox{State}
\vspace{6pt}\\
H & ::= & \mathit{Ref} \to \mathit{SessionRec} & \mbox{Heap}
\vspace{6pt}\\
 \mathit{SessionRec} & ::= & \srecs{x}{q}{\E, P}{y}
 \vspace{6pt}\\
\mathit{q} & ::= & \nilm \; | \; \mathit{Val}  @ q & \mbox{Queue}
\vspace{8pt}\\
\mathit{Val} & ::= & \checkmark & \mbox{Close token}\\
 & | & \labl{l} & \mbox{Choice label}\\
  & | & \clos(x, \E, P) & \mbox{Linear Closure} \\  
  & | & \closB(x, \E, P) & \mbox{Exponential Closure} 
\vspace{6pt}\\
   \mathcal{E},  \mathcal{G},\F & ::= & \mathit{Name} \to (\mathit{Ref} 
 \cup \mathit{Val}) \ \quad & \mbox{Environment}
\end{array}
$$}
\label{fig:bigSAM}
\caption{The SAM}
\end{figure}
In this section, we complete our initial presentation of the SAM, 
in particular, we introduce support for the exponentials, allowing
the machine to compute with non-linear values, and a selective concurrency
semantics.
We have delayed the introduction
of an environment structure for the SAM, to make the presentation easier to follow.
However, this was done at the expense of a more abstract formalisation of
the operational semantics, making use of $\alpha$-conversion,
and overloading language syntax names as heap references for allocated session records. 

The SAM actually relies on environment-based implementation of
name management, presented in Fig.~\ref{fig:bigSAM}.
A SAM state is then a triple $(\E,P,H)$ where $\E$ is an environment 
that maps each free name of the code $P$ into either a closure or a heap record endpoint. These heap references are freshly allocated and unique, thus avoiding any clashes and enforcing proper static scoping. Closures, representing suspended linear ($\clos(z,\E,P)$) and exponential behaviour ($\closB(z,\E,P)$), pair the code in its environment, and we expect the following structural safety conditions for name biding in configurations to hold.
\begin{definition}[Closure]

\noindent
A process $P$ is $(\E,N)$-closed
if $\fn{P}\setminus N \subseteq dom(\E)$,
and $\E$-closed if $(\E,\emptyset)$-closed.

\noindent
Environment $\E$ is $H$-closed if for all
$x\in dom(\E)$ if $\E(x)$ is a reference then $x\in H$,
if $\E(x)=\closB(z,\F,R)$ then $\F$ is $H$-closed
and $R$ is $(\F,\{z\})$-closed.

\noindent
Heap $H$ is closed if 
for all $\srecs{x}{q}{\G,Q}{y}\in H$, 
$\G$ is $H$-closed,
$Q$ is $\G$-closed,
and for all $\clos(z,\F,R)\in q$ and $\closB(z,\F,R)\in q$, 
$\F$ is $H$ closed and 
$R$ is $(\F,\{z\})$-closed.
State $(\E,P,H)$ is closed if $H$ is closed, $\E$ is $H$-closed, and $P$ is $\E$-closed.
\end{definition}
\begin{figure}
{\small$$
\begin{array}{lrr}
(\E, \cuti{\ou x:A \; [\nilm] \; y:B}{P}{Q}, H) \Mapsto 
(\G, P, H[\srecs{a}{q}{\F,Q}{b} ] ) \quad\quad\;& \text{\rm[SCut]}\\

a,b = \mathsf{new}, \G=\E\subs{a}{x}, 
 \F=\E\subs{b}{y}
 \vspace{6pt}
 \\
(\E, \pone{x}, H[\srecs{a}{q}{\F,P}{b}]) \Mapsto 
(\F, P, H[\srecs{a}{q @ {\checkmark}}{\emptyset,\zero}{b}])&  \text{\rm[S$\one$]} \\
a = \E (x)
 \vspace{6pt}
 \\
 (\E,\fwd{x}{y}, H[\srecs{c}{q_1}{\G,Q}{a}][\srecs{b}{q_2}{\F,P}{d}]) \Mapsto 
(\F,P, H[\srecs{c}{q_2 @ {q_1}}{\G,Q}{d}])&  \text{\rm[Sfwd]}\\
a = \E (x), b = \E(y)
 \vspace{6pt}
 \\
(\E,\pbot{y}{ P}, 
H[\srecs{a}{{\checkmark}}{\emptyset,\zero}{b}]) \Mapsto 
(\E, P, H) &  \text{\rm[S$\bot$]}\\
b = \E (y)
 \vspace{6pt}
 \\
(\E, \mathcal A^-(x), 
H[\srecs{a}{{q}}{\G,Q}{b}]) \Mapsto 
(\G, Q, H [\srecs{a}{{q}}{\E,\mathcal A^-(x)}{b}]) & \text{\rm[S$-$]}\\
a = \E (x)
 \vspace{6pt}
 \\
(\E, \potimes{x}{z}{R}{Q}, H[\srecs{a}{q}{P}{b}]) \Mapsto \\
\quad\quad\quad\quad 
\quad\quad\quad\quad 
(\E, Q, H[\srecs{a}{q @ \clos(z,\E,R)}{P}{b}]) & \text{\rm[S$\otimes$]}
\\
a = \E (x)
 \vspace{6pt}
 \\
(\E, \pparl{y}{w{:}+}{Q}, H[\srecs{a}{\clos(z,\F,R)@q }{\G,P}{b}]) \Mapsto \\
\quad\quad\quad\quad 
\quad\quad\quad\quad 
(\E', Q, H[\srecs{e}{\nilm}{\F',R}{f}][\srecss{a}{q }{\G,P}{b}])
& \text{\rm[S$\parl+$]}\
\\
e,f = \mathsf{new}, b = \E (y), \E' = \E\subs{e}{w}, \F' = \F\subs{f}{z}
 \vspace{6pt}
 \\
(\E, \pparl{y}{w{:}-}{Q}, H[\srecs{a}{\clos(z,\F,R)@q }{\G,P}{b}]) \Mapsto \\
\quad\quad\quad\quad 
\quad\quad\quad\quad 
(\F', R, H[\srecs{e}{\nilm}{\E', Q}{f}]
[\srecss{a}{q }{\G, P}{b}])
& \text{\rm[S$\parl-$]}\\
e,f = \mathsf{new}, b = \E (y), \F' = \F\subs{e}{z}, \E' = \E\subs{f}{w}
 \vspace{6pt}
 \\
( \E, \#\mathsf{l}\; x;Q, H[\srecs{a}{q}{\G, P}{b}]) \Mapsto 
(\E, Q, H[\srecs{a}{q @ \#\mathsf{l}}{\G, P}{b}]) & \text{\rm[S$\oplus$]}
\\
a = \E (x)
 \vspace{6pt}
 \\
(\E, \pcasem{y}{\ell}{L}{Q_\ell}, H[\srecs{a}{\#\mathsf{l}@q }{\G, P}{b}]) \Mapsto 
(\E, Q_{\labl{l}}, H[\srecss{a}{q }{\G,P}{b}])
& \text{\rm[S$\with$]}\\
b = \E (y)
 \vspace{6pt}
 \\
( \E, \pbang{x}{z}{Q}, H[\srecs{a}{q}{\G, P}{b}]) \Mapsto 
(\G,P, H[\srecs{a}{q @ \clos (z,\E,Q)}{\emptyset, \zero}{b}]) & 
\text{\rm[S$!$]}
\\
a = \E (x)
 \vspace{6pt}
 \\
( \E, \pwhy{y}{Q}, H[\srecs{a}{\clos (z,\F,R)}{\emptyset, \zero}{b}]) \Mapsto 
(\E',Q, H) & \text{\rm[S$?$]}
\\
b = \E (y), \E'=\E\subs{\clos (z,\F,R)}{y}
 \vspace{6pt}
 \\
( \E, \pcopy{y}{w{:}+}{Q}, H) \Mapsto 
(\E',Q, H[\srecs{a}{\nilm}{\F', R}{b}]) & \text{\rm[Scall+]}
\\
a,b = \mathsf{new}, 
\E' = \E\subs{a}{w}, \F' = \F\subs{b}{z}\\
\clos (z,\F,R) = \E (y)
 \vspace{6pt}
 \\
( \E, \pcopy{y}{w{:}-}{Q}, H) \Mapsto 
(\F',R, H[\srecs{a}{\nilm}{\E', Q}{b}]) & \text{\rm[Scall-]}
\\
a,b = \mathsf{new}, 
\E' = \E\subs{b}{w}, \F' = \F\subs{a}{z}\\
\clos (z,\F,R) = \E (y)
 \vspace{6pt}
 \\
 \\
\srecss{a}{q }{\G,P}{b} \deff \mbox{if }(q = \nilm)\mbox{ then }
\srecs{b}{q }{\G,P}{a}\mbox{ else }
\srecs{a}{q }{\G,P}{b}
\end{array}
$$}
\label{fig:final-sam}
\caption{SAM Transition Rules for the complete $\CLL$}
\end{figure}
In Figure~\ref{fig:final-sam} we present the environment-based execution rules for the SAM. 
All  rules except those for exponentials have already been 
essentially presented in Fig.~\ref{fig:rules} and discussed in previous sections. 
The only changes to those rules are due to the
presence of environments, which at all times record the bindings for free names in the code. Overall, we have
\begin{lemma} Let $P\vdashB \emptyset;\emptyset$.
For all $S$ such that $(P,\emptyset,\emptyset)
\stackrel{*}{\Mapsto} S$, $S$ is closed.
\end{lemma}
We discuss the SAM rules for the exponentials. Values of exponential type are represented by exponential closures $\closB(z,\F,R)$. Recall that a session type may terminate in either type $\one$, type $\bot$ or in an exponential type $!A/?A$ (cf. ~\ref{lemma:queue-full}). 
So, the (positive) execution rule  [S$!$] is similar to 
rule [S$\one$]: it enqueues the closure representing the
replicated process, and switches context, since the session terminates (cf. [!] Fig.~\ref{fig:redB}). 
The execution rule [S$?$] is similar to 
rule [S$\parl$]: it pops a closure from the queue (which, in this case,
always becomes empty),
and instead of using it immediately, adds it to the
environment to become persistently available to client code (cf.
reduction rule [S$?$] Fig.~\ref{fig:redB}).
\hide{ 
For any $P\vdashB\emptyset;\emptyset$
we extend Definition~\ref{def:encode} to an environment based mappping $enc(P)=C$ to $enc(P,\emptyset,\emptyset) \stackrel{\mathsf{cut}*}{\Mapsto}  C$ where $enc(P,H,\E) \expa  C$ is defined 
along the lines of  but adding the
expected encoding for the exponential 
cut

$$
\frac{
enc(P,H, \E\subs{\closB(z,\E,R)}{x})
\expa \mathcal C}{
enc(\cutBi{x]}{z}{R}{P}, H, \E)\expa \mathcal C}
$$

}
Any such closure representing
a replicated process
may be called by client code with transition rule [Scall],
which essentially creates a new linear session composed by
cut with the client code, similarly to [S$\parl$].
Rule [SCall] operates with some similarity to 
rule [S$\parl$]: instead of activating a linear closure 
popped from the queue, it activate an exponential closure fetched 
from the environment. 

We extend the $enc$ map to the exponential cut
and environment states $(\E,P,H)$ by adapting Definition~\ref{def:encode}, 
and adding the clause:

\vspace{4pt}
\centerline
{\small$
\displaystyle\frac{
enc(\E\subs{\closB(y,\E,R)}{x}, P),H )
\expa  C
}{
enc(\E,\cutBi{x}{y}{R}{P}, H)\expa  C}
$}
\vspace{4pt}
We now update our meta-theoretical results for the complete SAM. 
\begin{restatable}[Soundness]{theorem}{fulsamsound}
\label{fulsamsound}
Let $P\vdashB\emptyset;\emptyset$. 

\noindent
 If $enc(P) \Mapsto D \expa C$  then there
is $Q$ such that $P \to\cup\equiv Q$ and
$C = enc(Q)$.

\end{restatable}
\vspace{-.2cm}
\begin{restatable}[Progress]{theorem}{fullsamprogress}
\label{fullsamprogress}
Let $P\vdashB\emptyset;\emptyset$ and $P$ live. Then $enc(P)\Mapsto C$.


\end{restatable}
\vspace{-16pt}
\subsection{Concurrent Semantics of Cut and Mix}
\label{ref:mix}
Intuitively, the execution of mix $\mix{P}{Q}$
consists in the parallel execution of (non-interfering) processes $P$ and $Q$.
We may execute $\mix{P}{Q}$ by sequentialising $P$ and $Q$ in some 
arbirary way, and this actually may be useful in some cases. 
\hide{
This would yield simple
execution rules (we get back to environment-free presentation, for clarity).
$$
\frac{
(P, H) \Mapsto (P', H') 
}{
(\mix{P}{Q}, H) \Mapsto (\mix{P'}{Q}, H') 
}\quad\quad \text{\rm[SMixL]}
\quad
\quad
{
(\mix{\zero}{P}, H) \Mapsto (P, H) 
}\quad\quad \text{\rm[SMix$\zero$]}
$$
Introducing the symmetric rule of $\text{\rm[SMixL]}$ would allow the SAM 
to non-determini\-stically select a mix parallel thread to reduce.
}
\hide{
However, it is natural to introduce real concurrency in the SAM, not just for mix, but also for processes sharing binary sessions or other types of resources, such as mutex memory cells (cf.~\cite{DBLP:conf/esop/RochaC23}). To that end, we evolve the SAM from single threaded to multithreaded, introducing states of the form 
$$(\{P_1,P_2,\ldots,P_n\},H) \quad \quad \rm{Thread Pool}$$
exposing a multiset of processes $P_i$ available for execution.  
\hide{We denote the concurrent mix by $\mixP \{ \mix{P}{Q} \} $, and introduce the following SAM rules, where each moment a scheduler picks a thread for one step sequential execution [Srun], to be handled by the rules of the sequential SAM.

\vspace{-0.3cm}
{\small\[
\begin{array}{c}
{
( \mixP \{\mix{P}{Q}\}\uplus T, H) \Mapsto (\{P,Q\}\uplus T, H') 
}\; \text{\rm[Smixp]}
\quad
{
(\zero\uplus T, H) \Mapsto (T, H) 
}\; \text{\rm[Sstp]}
\vspace{6pt}\\
\displaystyle
\frac{
(P, H) \Mapsto (P', H') 
}{
(P\uplus T, H) \Mapsto (P'\uplus T, H') 
}\; \text{\rm[Srun]}

\end{array}
\]}
}
Notice that each individual thread executes locally according to
the SAM sequential transitions from before.  These rules do
not break any of the presented correctness results, since by typing
each $P_i$ has a disjoint heap footprint.
}

However, much more interesting is the accommodation in the SAM of interfering concurrency,
as required to support full-fledged concurrent languages for
session-based programming.
{First, we evolve the SAM from single threaded to multithreaded, where states 
now expose a multiset of processes $P_i$ ready for execution by the
basic SAM sequential transitions: $(\{P_1,P_2,\ldots,P_n\},H)$
and introduce an annotated variant $\cutPc$ of the cut.
It has the same $\CLL/\Blang$ semantics, 
but to be implemented as a fork construct where $P$ and $Q$ spawn concurrently,
 their interaction mediated by an atomic
\emph{concurrent session record} $\srecc{x}{q}{y}$. The type system ensuring that
concurrent channels may be forwarded only to
concurrent channels.
We extend the SAM with transition rule for multisets:}

\vspace{-0.3cm}
{\small $$
\begin{array}{c}
\displaystyle
\frac{
(P, H) \Mapsto (P', H') 
}{
(P\uplus T, H) \Mapsto (P'\uplus T, H') 
}\; \text{\rm[Srun]}\vspace{6pt}\\
(\cutp{\ou x\ass A \;[q]\; y\ass B}{P}{Q}\uplus T, H) \Mapsto 
(\{P,Q\}\uplus T\}, H[\srecc{x}{\nilm}{y}] ) \quad \text{\rm[SCutp]}\vspace{6pt}\\
((\mix{P}{Q})\uplus T, H) \Mapsto 
(\{P,Q\}\uplus T\}, H[\srecc{x}{\nilm}{y}] ) \quad \text{\rm[SMixp]}
\end{array}
$$}

\noindent
Each individual thread executes locally according to the SAM sequential transitions presented before, until an action on a concurrent queue is reached. 
Concurrent process actions on concurrent queues are atomic, and defined as expected. Positive actions always progress by  pushing a value into the queue, while
negative actions will either pop off a value from the queue or block, waiting for a value to become available. We illustrate with the rules for $\one$,$\bot$ typed actions.

\vspace{-.2cm}
{\small\[
\begin{array}{ll}
(\pone{x}, H[\srecc{x}{q}{y}]) \Mapsto 
(\zero, H[\srecc{x}{q @ {\checkmark}}{y}])\quad\quad\;&  \text{\rm[S$\one$c]} 
\vspace{6pt}\\
(\pbot{y}{ P}, 
H[\srecs{x}{{\checkmark}}{y}]) \Mapsto 
(P, H) &  \text{\rm[S$\bot$c]}
\end{array}
\]}

Notice that, as in the case for $\pbot{y}{ P}$ above, any negative action in the thread queue is unable to progress if the corresponding queue is empty. 
It should be clear how to define transition rules for all other pairs of
dual actions. Given an appropriate encoding $enc^c$
of annotated $\Blang$ processes in concurrent SAM states, and as consequence of typing and leveraging the proof scheme for 
progress in $\Blang$ (Theorem~\ref{theorem:progress}),  we have:
\begin{restatable}[Soundness-c]{theorem}{samcsound}
\label{samcsound}
Let $P\vdashB\emptyset;\emptyset$. 

\noindent
 If $enc^c(P) \Mapsto D \expa C$  then there
is $Q$ such that $P \to\cup\equiv Q$ and
$C = enc^c(Q)$.
\end{restatable}
\vspace{-5pt}
\begin{restatable}[Progress-c]{theorem}{samcprogress}
\label{samcprogress}
Let $P\vdashB\emptyset;\emptyset$ and $P$ live. Then $enc_c(P)\Mapsto C$.


\end{restatable}
The extended SAM executes concurrent
session programs, consisting in an arbitrary number of
concurrent threads. Each thread deterministically executes sequential code,
but can at any moment spawn new concurrent threads.
{The whole model is expressed in the common language of (classical) linear logic,
statically ensuring safety, proper resource usage, termination, and deadlock absence by static typing.}

 \hide{
The DILL (dyadic) formulation of $\CLL$ 
considers types $!A$ and $?A$ as linear while in the context $\Delta$.
The bookkeeping reduction rule [?] ``moves'' a linear name from $\Delta$ to the 
exponential context $\Gamma$, where it may be composed with a replicated server
by [TCut!]. In the operational semantics of closed processes this corresponds to
activating a replicable server, implemented in the SAM by 
}

 
\vspace{-10pt}
\section{Concluding Remarks and Related Work}\label{sec:conclusion}
\vspace{-5pt}
We introduce the Session Abstract Machine, or SAM, an
abstract machine 
for executing session processes typed by (classical) linear logic
$\CLL$,
deriving a deterministic, sequential evaluation strategy, where exactly
one process is executing at any given point in time. In the SAM,
session channels are implemented as single queues with a write and a read
endpoint, which are written to, and read by executing
processes. 
Positive actions are non-blocking, giving rise to a degree of asynchrony.
However, processes in a session synchronise at polarity inversions, where 
they alternate execution, according to a fixed co-routining strategy.
Despite its specific strategy, the SAM semantics is sound wrt $\CLL$
and satisfies the correctness properties of logic-based
session type systems. We also present a conservative concurrent
extension of the SAM, allowing the degrees of concurrency to
be modularly expressed at a fine grain, ranging from fully sequential
to fully concurrent execution. Indeed, a practical concern with the
SAM design lies in providing a principled foundation for an
execution environment for multi-paradigm languages,
combining concurrent, imperative and functional programming.
{
The overall SAM design 
as presented here may be uniformly extended to cover any other polarised
language constructs that conservatively extend the PaT paradigm, such as polymorphism,
affine types, recursive and co-recursive types, and shared state~\cite{DBLP:conf/fossacs/PfenningG15,DBLP:conf/esop/RochaC23}.
We have implemented
a SAM-based version~\cite{artifactesop24} of an open-source implementation of
$\CLL$~\cite{artifactEsop23}.}

A machine model provides evidence of the algorithmic feasibility of a
programming language abstract semantics, and illuminates its
operational meaning from certain concrete semantic perspective.  Since
the seminal work of Landin on the SECD~\cite{Landin64}, several
machines to support the execution of programs for a given
programming language have been proposed. The SAM is then proposed
herein in this same spirit of Cousineau, Curien and Mauny's
Categorical Abstract Machine
for the call-by-value $\lambda$-calculus~\cite{DBLP:journals/scp/CousineauCM87}, Lafont's Linear Abstract
Machine for the linear
$\lambda$-calculus~\cite{DBLP:journals/tcs/Lafont88}, and Krivine's
Machine for the call-by-name
$\lambda$-calculus~\cite{DBLP:journals/lisp/Krivine07} ; these works explored Curry-Howard correspondences
to propose provably correct solutions. In~\cite{DBLP:conf/ifl/Danvy04},
Danvy developed a deconstruction of the SECD based on a sequence of
program transformations.  The SAM is also derived from Curry-Howard
correspondences for linear
logic $\CLL$~\cite{cairesmscs16,wadler2014propositions}, and we also rely
on program conversions, via the intermediate buffered language $\Blang$, as a
key proof technique.  We believe that the SAM
is the first proposal of its kind to tackle the challenges of a process
language, while building on several deep properties of its type structure towards a
principled design. Among those,
focusing~\cite{DBLP:journals/logcom/Andreoli92} and
polarisation~\cite{DBLP:conf/tlca/Laurent99,DBLP:journals/tcs/HondaL10,DBLP:conf/fossacs/PfenningG15}\hide{and asynchrony~\cite{}}
played an important role to achieve a deterministic sequential reduction
strategy for session-based programming, perhaps our main initial
motivation. That allows the SAM to naturally and efficiently integrate the
execution of sequential and concurrent session behaviours, and suggests
effective compilation schemes for mainstream virtual machines or compiler frameworks.

The adoption of session and linear types 
is clearly increasing in research 
(e.g.,~\cite{DBLP:conf/sefm/FrancoV13,DBLP:journals/iandc/AlmeidaMTV22,DBLP:conf/esop/PocasCMV23,DBLP:journals/lmcs/DasP22,DBLP:journals/corr/WillseyPP17,DBLP:conf/esop/ToninhoCP13,DBLP:conf/fossacs/PfenningG15,DBLP:conf/esop/RochaC23})
and general purpose 
languages (e.g.,~Haskell~\cite{DBLP:journals/pacmpl/BernardyBNJS18,DBLP:conf/haskell/KokkeD21},
Rust~\cite{DBLP:conf/coordination/LagaillardieNY20,DBLP:conf/ecoop/ChenBT22}
Ocaml~\cite{DBLP:conf/ecoop/ImaiNYY19,DBLP:journals/jfp/Padovani17},
F\#~\cite{DBLP:conf/cc/NeykovaHYA18}, Move~\cite{move-lang}, among many others), which either
require sophisticated encodings of linear typing via type-level
computation or forego of some static correctness properties for
usability purposes.
Such developments typically have as a main focus the realization
of the session typing discipline (or of a particular refinement of
such typing), with the underlying concurrent execution model often
offloaded to existing language infrastructure.

We highlight the work~\cite{DBLP:journals/pacmpl/CastellanY19}, which
studies the relationship between synchronous session types and game
semantics, which are fundamentally asynchronous. Their work proposes
an encoding of synchronous strategies into asynchronous strategies by so-called
call-return protocols. While their focus differs significantly from
ours, the encoding via asynchrony is reminiscent of our own.

We further note the work~\cite{DBLP:conf/csl/Munch-Maccagnoni09} which
develops a polarized variant of the $\overline{\lambda}\mu\tilde{\mu}$
suitable for sequent calculi like that of linear logic.
While we draw upon similar inspirations in the design of the SAM,
there are several key distinctions: the
work~\cite{DBLP:conf/csl/Munch-Maccagnoni09} presents
$\lambda\mu$-calculi featuring values and substitution of terms for
variables (potentially deep within the term structure). Our system,
being based on processes calculus, features neither -– there is no term
representing the outcome of a computation, since computation is the
interactive behavior of processes (cf. game semantics); nor does
computation rely on substitution in the same sense. Another significant
distinction is that our work materializes a heap-based abstract
machine rather than a stack-based machine.
 Finally, our type and term structure is not itself polarized. Instead, we
 draw inspiration from focusing insofar as we extract from focusing
 the insights that drive execution in the SAM. 

%

\hide{
Session-based programming is a programming style based on the pervasive use of
linearity and interaction, corresponding to a lazy form of computation with strict control of resources. Unlike in functional programming, where basic program elements are functions (e.g.,~objects of type $A\to B$ and their application to
objects of type $A$ to compute objects of type $B$), in session based
programming basic program elements are processes that interact
symmetrically by passing values amongst themselves according to a linear session protocol. In a session-calculus a linear function object may be modelled by a session typed process $x$ of type $?[A].![B].\mathsf{end}$~\cite{DBLP:journals/acta/GayH05}  or $T{\multimap} B\otimes \one$~\cite{DBLP:conf/concur/CairesP10} or $\dual T \parl B \otimes 1$~\cite{cairesmscs16,wadler2014propositions} which receives an argument of type $A$ and returns back the result of type $B$. Much more general protocols maybe usefully defined with session typed objects, particularly convenient to define computations that strictly manipulate linear resources, and where disposal and duplication of resources is made explicit (e.g, as in Rust), and where replicable (duplicable) are introduced via exponential types. 

For example, in ~\cite{X}\cite{DBLP:journals/jfp/Vasconcelos05}
}

In future work, we plan to study the semantics of the SAM in terms of
games (and categories), along the lines
of~\cite{DBLP:journals/pacmpl/CastellanY19,DBLP:journals/scp/CousineauCM87,DBLP:journals/tcs/Lafont88}. 
We also plan to
investigate the ways in which the evaluation strategy of the SAM can
be leveraged to develop efficient {compilation} of fine-grained session-based
programming, and its relationship with effect handlers, coroutines and delimited continuations. Linearity plays a key role in programming languages and environments for smart contracts in distributed ledgers \cite{DBLP:journals/lmcs/DasP22,DBLP:journals/pacmpl/SergeyNJ0TH19} manipulating linear resources (assets); it would be interesting to investigate 
how linear abstract machines like the SAM would provide a basis for certifying resource sensitive computing infrastructures~\cite{wood2014ethereum,move-lang}.

\paragraph{Acknowledgments.} This work was supported by NOVA LINCS
    (UIDB/ 04516/ 2020), INESC ID (UIDB/ 50021/ 2020), BIG (Horizon EU 952226 BIG).    




%
%
%

 \bibliographystyle{splncs04}
 \bibliography{bibliography}

\begin{thebibliography}{10}
\providecommand{\url}[1]{\texttt{#1}}
\providecommand{\urlprefix}{URL }
\providecommand{\doi}[1]{https://doi.org/#1}

\bibitem{DBLP:journals/tcs/Abramsky93}
Abramsky, S.: {Computational Interpretations of Linear Logic}. Theoret. Comput.
  Sci.  \textbf{111}(1--2),  3--57 (1993)

\bibitem{abramsky1996interaction}
Abramsky, S., Gay, S.J., Nagarajan, R.: {Interaction categories and the
  foundations of typed concurrent programming}. In: NATO ASI DPD. pp. 35--113
  (1996)

\bibitem{DBLP:journals/iandc/AlmeidaMTV22}
Almeida, B., Mordido, A., Thiemann, P., Vasconcelos, V.T.: Polymorphic lambda
  calculus with context-free session types. Inf. Comput.  \textbf{289}(Part),
  104948 (2022)

\bibitem{DBLP:journals/logcom/Andreoli92}
Andreoli, J.M.: {Logic Programming with Focusing Proofs in Linear Logic}. J.
  Log. Comput.  \textbf{2}(3),  297--347 (1992)

\bibitem{balzer2017manifest}
Balzer, S., Pfenning, F.: Manifest sharing with session types. Proc. ACM
  Program. Lang.  \textbf{1}(ICFP) (2017)

\bibitem{bellin.scott:linear-logic}
Bellin, G., Scott, P.: On the $\pi$-calculus and linear logic. Theoret. Comput.
  Sci.  \textbf{135}(1),  11--65 (1994)

\bibitem{Benton94amixed}
Benton, P.N.: A mixed linear and non-linear logic: Proofs, terms and models.
  In: International Workshop on Computer Science Logic. pp. 121--135. Springer
  (1994)

\bibitem{DBLP:journals/pacmpl/BernardyBNJS18}
Bernardy, J., Boespflug, M., Newton, R.R., Jones, S.P., Spiwack, A.: Linear
  haskell: practical linearity in a higher-order polymorphic language. Proc.
  {ACM} Program. Lang.  \textbf{2}({POPL}),  5:1--5:29 (2018)

\bibitem{move-lang}
Blackshear, S., Cheng, E., Dill, D.L., Gao, V., Maurer, B., Nowacki, T., Pott,
  A., Qadeer, S., Russi, D., Sezer, D., Zakian, T., Zhou, R.: {Move: A Language
  with Programmable Resources}  (2019)

\bibitem{DBLP:conf/esop/CairesP17}
Caires, L., P{\'{e}}rez, J.A.: Linearity, control effects, and behavioral
  types. In: Yang, H. (ed.) Programming Languages and Systems - 26th European
  Symposium on Programming, {ESOP} 2017. Lecture Notes in Computer Science,
  vol. 10201, pp. 229--259. Springer (2017)

\bibitem{caires2017linearity}
Caires, L., P\'{e}rez, J.A.: Linearity, control effects, and behavioral types.
  In: Proceedings of the 26th European Symposium on Programming Languages and
  Systems - Volume 10201. p. 229–259. Springer-Verlag, Berlin, Heidelberg
  (2017)

\bibitem{DBLP:conf/esop/CairesPPT13}
Caires, L., P\'{e}rez, J.A., Pfenning, F., Toninho, B.: Behavioral polymorphism
  and parametricity in session-based communication. In: Proceedings of the 22nd
  European Conference on Programming Languages and Systems. p. 330–349.
  ESOP'13, Springer-Verlag, Berlin, Heidelberg (2013)

\bibitem{DBLP:conf/concur/CairesP10}
Caires, L., Pfenning, F.: Session types as intuitionistic linear propositions.
  In: Gastin, P., Laroussinie, F. (eds.) CONCUR 2010 - Concurrency Theory. pp.
  222--236. Springer Berlin Heidelberg, Berlin, Heidelberg (2010)

\bibitem{TLDI12}
Caires, L., Pfenning, F., Toninho, B.: Towards concurrent type theory. In:
  Proceedings of the 8th ACM SIGPLAN Workshop on Types in Language Design and
  Implementation. p. 1–12. TLDI '12, Association for Computing Machinery, New
  York, NY, USA (2012)

\bibitem{cairesmscs16}
Caires, L., Pfenning, F., Toninho, B.: Linear logic propositions as session
  types. Mathematical Structures in Computer Science  \textbf{26}(3),
  367–423 (2016)

\bibitem{artifactesop24}
Caires, L., Toninho, B.: {The Session Abstract Machine (Artifact)}  (2024).
  \doi{10.5281/zenodo.10459455}

\bibitem{Cardelli91}
Cardelli, L.: {Typeful Programming}. IFIP State-of-the-Art Reports: Formal
  Description of Programming Concepts pp. 431--507 (1991)

\bibitem{DBLP:journals/pacmpl/CastellanY19}
Castellan, S., Yoshida, N.: Two sides of the same coin: session types and game
  semantics: a synchronous side and an asynchronous side. Proc. {ACM} Program.
  Lang.  \textbf{3}({POPL}),  27:1--27:29 (2019)

\bibitem{DBLP:conf/ecoop/ChenBT22}
Chen, R., Balzer, S., Toninho, B.: {Ferrite: {A} Judgmental Embedding of
  Session Types in Rust}. In: Ali, K., Vitek, J. (eds.) 36th European
  Conference on Object-Oriented Programming, {ECOOP} 2022. LIPIcs, vol.~222,
  pp. 22:1--22:28 (2022)

\bibitem{DBLP:journals/scp/CousineauCM87}
Cousineau, G., Curien, P., Mauny, M.: {The Categorical Abstract Machine}. Sci.
  Comput. Program.  \textbf{8}(2),  173--202 (1987)

\bibitem{DBLP:conf/ifl/Danvy04}
Danvy, O.: {A Rational Deconstruction of Landin's {SECD} Machine}. In: Grelck,
  C., Huch, F., Michaelson, G., Trinder, P.W. (eds.) Implementation and
  Application of Functional Languages, 16th International Workshop, {IFL} 2004.
  LNCS, vol.~3474, pp. 52--71. Springer (2004)

\bibitem{DBLP:conf/fossacs/DardhaG18}
Dardha, O., Gay, S.J.: A new linear logic for deadlock-free session-typed
  processes. In: Baier, C., Lago, U.D. (eds.) Foundations of Software Science
  and Computation Structures - 21st International Conference, {FOSSACS} 2018.
  LNCS, vol. 10803, pp. 91--109. Springer (2018)

\bibitem{DBLP:journals/lmcs/DasP22}
Das, A., Pfenning, F.: Rast: {A} language for resource-aware session types.
  Log. Methods Comput. Sci.  \textbf{18}(1) (2022)

\bibitem{DeYoung2012}
DeYoung, H., Caires, L., Pfenning, F., Toninho, B.: Cut reduction in linear
  logic as asynchronous session-typed communication. In: Computer Science Logic
  (2012)

\bibitem{DBLP:conf/sefm/FrancoV13}
Franco, J., Vasconcelos, V.T.: A concurrent programming language with refined
  session types. In: Counsell, S., N{\'{u}}{\~{n}}ez, M. (eds.) Software
  Engineering and Formal Methods - {SEFM} 2013. LNCS, vol.~8368, pp. 15--28.
  Springer (2013)

\bibitem{DBLP:journals/pacmpl/FruminDHP22}
Frumin, D., D'Osualdo, E., van~den Heuvel, B., P{\'{e}}rez, J.A.: A bunch of
  sessions: a propositions-as-sessions interpretation of bunched implications
  in channel-based concurrency. Proc. {ACM} Program. Lang.
  \textbf{6}({OOPSLA2}),  841--869 (2022)

\bibitem{DBLP:journals/acta/GayH05}
Gay, S., Hole, M.: {Subtyping for Session Types in the Pi Calculus}. Acta
  Informatica  \textbf{42}(2-3),  191--225 (2005)

\bibitem{gay2010linear}
Gay, S., Vasconcelos, V.: {Linear Type Theory for Asynchronous Session Types}.
  Journal of Functional Programming  \textbf{20}(1),  19--50 (2010)

\bibitem{DBLP:journals/mscs/Girard91}
Girard, J.: A new constructive logic: Classical logic. Math. Struct. Comput.
  Sci.  \textbf{1}(3),  255--296 (1991)

\bibitem{DBLP:conf/concur/Honda93}
Honda, K.: Types for dyadic interaction. In: Best, E. (ed.) CONCUR'93. pp.
  509--523. Springer Berlin Heidelberg, Berlin, Heidelberg (1993)

\bibitem{DBLP:journals/tcs/HondaL10}
Honda, K., Laurent, O.: An exact correspondence between a typed pi-calculus and
  polarised proof-nets. Theor. Comput. Sci.  \textbf{411}(22-24),  2223--2238
  (2010)

\bibitem{DBLP:conf/esop/HondaVK98}
Honda, K., Vasconcelos, V.T., Kubo, M.: Language primitives and type discipline
  for structured communication-based programming. In: Hankin, C. (ed.)
  Programming Languages and Systems. pp. 122--138. Springer (1998)

\bibitem{DBLP:journals/csur/HuttelLVCCDMPRT16}
H{\"{u}}ttel, H., Lanese, I., Vasconcelos, V.T., Caires, L., et~al.:
  {Foundations of Session Types and Behavioural Contracts}. {ACM} Comput. Surv.
   \textbf{49}(1), ~3 (2016)

\bibitem{DBLP:conf/ecoop/ImaiNYY19}
Imai, K., Neykova, R., Yoshida, N., Yuen, S.: Multiparty session programming
  with global protocol combinators. In: Hirschfeld, R., Pape, T. (eds.) 34th
  European Conference on Object-Oriented Programming, {ECOOP} 2020. LIPIcs,
  vol.~166, pp. 9:1--9:30. Schloss Dagstuhl - Leibniz-Zentrum f{\"{u}}r
  Informatik (2020)

\bibitem{DBLP:journals/pacmpl/JacobsB23}
Jacobs, J., Balzer, S.: Higher-order leak and deadlock free locks. Proc. {ACM}
  Program. Lang.  \textbf{7}({POPL}),  1027--1057 (2023)

\bibitem{rust-lang}
Klabnik, S., Nichols, C.: {The Rust Programming Language}  (2021)

\bibitem{DBLP:conf/haskell/KokkeD21}
Kokke, W., Dardha, O.: {Deadlock-free session types in linear Haskell}. In:
  Hage, J. (ed.) Haskell 2021: Proceedings of the 14th {ACM} {SIGPLAN}
  International Symposium on Haskell. pp. 1--13. {ACM} (2021)

\bibitem{DBLP:journals/pacmpl/KokkeMP19}
Kokke, W., Montesi, F., Peressotti, M.: Better late than never: a
  fully-abstract semantics for classical processes. Proc. {ACM} Program. Lang.
  \textbf{3}({POPL}),  24:1--24:29 (2019)

\bibitem{DBLP:journals/lisp/Krivine07}
Krivine, J.: {A call-by-name Lambda-calculus Machine}. High. Order Symb.
  Comput.  \textbf{20}(3),  199--207 (2007)

\bibitem{DBLP:journals/tcs/Lafont88}
Lafont, Y.: {The Linear Abstract Machine}. Theor. Comput. Sci.  \textbf{59},
  157--180 (1988)

\bibitem{DBLP:conf/coordination/LagaillardieNY20}
Lagaillardie, N., Neykova, R., Yoshida, N.: {Implementing Multiparty Session
  Types in Rust}. In: Coordination Models and Languages {Coordination} 2020.
  Lecture Notes in Computer Science, vol. 12134, pp. 127--136. Springer (2020)

\bibitem{Landin64}
Landin, P.J.: {The Mechanical Evaluation of Expressions}. The Computer Journal,
  Volume 6, Issue 4, January 1964  \textbf{6}(4),  308–320 (1964)

\bibitem{DBLP:conf/tlca/Laurent99}
Laurent, O.: {Polarized Proof-Nets: Proof-Nets for {LC}}. In: Girard, J. (ed.)
  Typed Lambda Calculi and Applications, 4th International Conference, TLCA'99.
  LNCS, vol.~1581, pp. 213--227. Springer (1999)

\bibitem{DBLP:conf/haskell/LindleyM16}
Lindley, S., Morris, J.G.: {Embedding session types in Haskell}. In: Mainland,
  G. (ed.) Proceedings of the 9th International Symposium on Haskell, Haskell
  2016, Nara, Japan, September 22-23, 2016. pp. 133--145. {ACM} (2016)

\bibitem{DBLP:conf/ppdp/LopesSV99}
Lopes, L.M.B., Silva, F.M.A., Vasconcelos, V.T.: A virtual machine for a
  process calculus. In: Nadathur, G. (ed.) Principles and Practice of
  Declarative Programming, International Conference PPDP'99. Lecture Notes in
  Computer Science, vol.~1702, pp. 244--260. Springer (1999)

\bibitem{DBLP:journals/mscs/Milner92}
Milner, R.: Functions as processes. Math. Struct. Comput. Sci.  \textbf{2}(2),
  119--141 (1992)

\bibitem{milner1993elements}
Milner, R.: Elements of interaction: Turing award lecture. Communications of
  the ACM  \textbf{36}(1),  78--89 (1993)

\bibitem{DBLP:books/daglib/0098267}
Milner, R.: Communicating and mobile systems - the Pi-calculus. Cambridge
  University Press (1999)

\bibitem{DBLP:conf/csl/Munch-Maccagnoni09}
Munch{-}Maccagnoni, G.: Focalisation and classical realisability. In:
  Gr{\"{a}}del, E., Kahle, R. (eds.) Computer Science Logic, 23rd international
  Workshop, {CSL} 2009. LNCS, vol.~5771, pp. 409--423. Springer (2009)

\bibitem{DBLP:conf/cc/NeykovaHYA18}
Neykova, R., Hu, R., Yoshida, N., Abdeljallal, F.: A session type provider:
  compile-time {API} generation of distributed protocols with refinements in
  f{\#}. In: Dubach, C., Xue, J. (eds.) Proceedings of the 27th International
  Conference on Compiler Construction, {CC} 2018, February 24-25, 2018, Vienna,
  Austria. pp. 128--138. {ACM} (2018)

\bibitem{DBLP:journals/jfp/Padovani17}
Padovani, L.: A simple library implementation of binary sessions. J. Funct.
  Program.  \textbf{27}, ~e4 (2017)

\bibitem{Perez12esop}
P{\'e}rez, J.A., Caires, L., Pfenning, F., Toninho, B.: Linear logical
  relations and observational equivalences for session-based concurrency.
  Information and Computation  \textbf{239},  254--302 (2014)

\bibitem{DBLP:conf/coordination/PfenningP23}
Pfenning, F., Pruiksma, K.: Relating message passing and shared memory,
  proof-theoretically. In: Jongmans, S., Lopes, A. (eds.) Coordination Models
  and Languages - {COORDINATION} 2023. LNCS, vol. 13908, pp. 3--27. Springer
  (2023)

\bibitem{DBLP:conf/lics/Pfenning95}
Pfenning, F.: Structural cut elimination. In: Proceedings of the 10th Annual
  IEEE Symposium on Logic in Computer Science. p.~156. LICS '95, IEEE Computer
  Society, USA (1995)

\bibitem{DBLP:conf/fossacs/PfenningG15}
Pfenning, F., Griffith, D.: {Polarized Substructural Session Types}. In: Proc.
  of {FoSSaCS} 2015. LNCS, vol.~9034, pp. 3--22. Springer (2015)

\bibitem{DBLP:conf/birthday/PierceT00}
Pierce, B.C., Turner, D.N.: Pict: a programming language based on the
  pi-calculus. In: Plotkin, G.D., Stirling, C., Tofte, M. (eds.) Proof,
  Language, and Interaction, Essays in Honour of Robin Milner. pp. 455--494.
  The {MIT} Press (2000)

\bibitem{DBLP:conf/esop/PocasCMV23}
Po{\c{c}}as, D., Costa, D., Mordido, A., Vasconcelos, V.T.: System
  f\({}^{\mbox{{\(\mu\)}}}\) \({}_{\mbox{{\o}mega}}\) with context-free session
  types. In: Wies, T. (ed.) Programming Languages and Systems - 32nd European
  Symposium on Programming, {ESOP} 2023. LNCS, vol. 13990, pp. 392--420.
  Springer (2023)

\bibitem{qian2020client}
Qian, Z., Kavvos, G., Birkedal, L.: Client-server sessions in linear logic.
  Proceedings of the ACM on Programming Languages  \textbf{5}(ICFP),  1--31
  (2021)

\bibitem{rocha2021propositions}
Rocha, P., Caires, L.: {Propositions-as-types and Shared State}. Proceedings of
  the ACM on Programming Languages  \textbf{5}(ICFP),  1--30 (2021)

\bibitem{DBLP:conf/esop/RochaC23}
Rocha, P., Caires, L.: Safe session-based concurrency with shared linear state.
  In: Wies, T. (ed.) Programming Languages and Systems - 32nd European
  Symposium on Programming, {ESOP} 2023. LNCS, vol. 13990, pp. 421--450.
  Springer (2023)

\bibitem{artifactEsop23}
Rocha, P., Caires, L.: Safe session-based concurrency with shared linear state
  (artifact)  (January 2023). \doi{10.5281/zenodo.7506064}

\bibitem{sangiorgi-walker:book}
Sangiorgi, D., Walker, D.: PI-Calculus: A Theory of Mobile Processes. Cambridge
  University Press, USA (2001)

\bibitem{DBLP:journals/pacmpl/SergeyNJ0TH19}
Sergey, I., Nagaraj, V., Johannsen, J., Kumar, A., Trunov, A., Hao, K.: {Safer
  smart contract programming with Scilla}. Proc. {ACM} Program. Lang.
  \textbf{3}({OOPSLA}),  185:1--185:30 (2019)

\bibitem{Toninho12fossacs}
Toninho, B., Caires, L., Pfenning, F.: {Functions as Session-Typed Processes}.
  In: FoSSaCS'12. No.~7213 in LNCS (2012)

\bibitem{DBLP:conf/esop/ToninhoCP13}
Toninho, B., Caires, L., Pfenning, F.: Higher-order processes, functions, and
  sessions: A monadic integration. In: Felleisen, M., Gardner, P. (eds.)
  Programming Languages and Systems. pp. 350--369. Springer (2013)

\bibitem{DBLP:conf/ppdp/ToninhoCP21}
Toninho, B., Caires, L., Pfenning, F.: A decade of dependent session types. In:
  Veltri, N., Benton, N., Ghilezan, S. (eds.) {PPDP} 2021: 23rd International
  Symposium on Principles and Practice of Declarative Programming. pp.
  3:1--3:3. {ACM} (2021)

\bibitem{DBLP:conf/esop/ToninhoY18}
Toninho, B., Yoshida, N.: On polymorphic sessions and functions: A tale of two
  (fully abstract) encodings. ACM Trans. Program. Lang. Syst.  \textbf{43}(2)
  (Jun 2021)

\bibitem{DBLP:phd/ethos/Turner96}
Turner, D.N.: {The polymorphic Pi-calculus : theory and implementation}. Ph.D.
  thesis, University of Edinburgh, {UK} (1996)

\bibitem{DBLP:journals/jfp/Vasconcelos05}
Vasconcelos, V.T.: Lambda and pi calculi, {CAM} and {SECD} machines. J. Funct.
  Program.  \textbf{15}(1),  101--127 (2005)

\bibitem{Wadler12icfp}
Wadler, P.: Propositions as sessions. In: Proceedings of the 17th ACM SIGPLAN
  International Conference on Functional Programming. p. 273–286. ICFP '12,
  Association for Computing Machinery, New York, NY, USA (2012)

\bibitem{wadler2014propositions}
Wadler, P.: {Propositions as Sessions}. Journal of Functional Programming
  \textbf{24}(2-3),  384--418 (2014)

\bibitem{DBLP:journals/cacm/Wadler15}
Wadler, P.: {Propositions as Types}. Commun. {ACM}  \textbf{58}(12),  75--84
  (2015)

\bibitem{DBLP:journals/corr/WillseyPP17}
Willsey, M., Prabhu, R., Pfenning, F.: Design and implementation of concurrent
  {C0}. In: Cervesato, I., Fern{\'{a}}ndez, M. (eds.) Proceedings Fourth
  International Workshop on Linearity, {LINEARITY} 2016. {EPTCS}, vol.~238, pp.
  73--82 (2016)

\bibitem{wood2014ethereum}
Wood, G.: {Ethereum: A Secure Decentralised Generalised Transaction Ledger}.
  Ethereum project yellow paper  \textbf{151}(2014),  1--32 (2014)

\end{thebibliography}

\appendix
 


\newpage
\section{Appendix - Supplementary Material}\label{appendix}

\subsection{Proofs of Section~\ref{sec:safetyCLLB}: Preservation and Progress for $\Blang$ }

\queueprf*



\begin{proof}
To check that the rule is admissible, directly derive the
conclusion using $k$ applications of [TCut$\otimes$] or [TCut$\oplus$].
To check inversion, we proceed by
induction in the derivation of
$\cuti{\ou{x}:A \;[\overline{c_k}] \; y:B}{P}{Q} \vdashB \Delta; \Gamma$.
 
  (Case  [TCutB]) We have $k=0$ so $q=\nilm$, $A$ positive and $B$ negative. 

 (Case  [TCut-$\one$]) 
 Not applicable, since $c_k\neq \checkmark$.
 
(Case [TCut-$\otimes$])
We have $\cuti{\ou{x}{:}A \;[\overline{c_{k-1}} @ \mathsf{clos}(z_k;R_k)   ] \; y{:}B}{P}{Q} \vdashB \Delta; \Gamma$ derived from
$\cuti{\ou{x}:T_k{\otimes} A \;[  \overline{c_{k-1}} ]\; y:B}{\potimes{x}{z_k}{R_k}{P}}{Q} \vdashB \Delta; \Gamma$.
By i.h. $\potimes{x}{z_k}{R_k}{P}\vdashB \Sigma_{k},x{:}T_k\otimes A$ and
$Q \vdash \Delta_Q, y{:}B$ where $B=\mathbf{E}_{k-1};{\dual { (T_k\otimes {A}) }}$
and 
$\Delta_i \vdashB c_i : E_i $ for $i\leq k-1$.
By inversion ([T$\otimes$])
$P \vdashB \Delta_P,x:A$,
 $\Sigma_{k} = \Delta_k,\Delta_P$ and 
 $R_k \vdashB \Delta_k, z_k\ass T_k$, so
 $\Delta_k \vdashB c_k \ass E_k$ where $E_k=\dual{T_k}\parl \Box $.
We conclude $P\vdashB \Delta_{P},x:A$,
$Q \vdash \Delta_Q, y\ass B$,
$B= \mathbf{E}_k ;\dual{A}$
and 
$\Delta_i \vdash c_i \ass E_i $ for $i\leq k$.

(Case  [TCut-$\oplus$]) 
We have $\cuti{\ou{x}{:}A_{\labl{l}} \;[\overline{c_{k-1}} @ \labl{l}   ] \; y{:}B}{P}{Q} \vdashB \Delta; \Gamma$ derived from
$\cuti{\ou{x}: \oplusm{\ell}{L}{A_\ell} \;[  \overline{c_{k-1}} ]\; y:B}{{\labl{l}\ x;P}}{Q} \vdashB \Delta, \Delta'; \Gamma$.
We proceed by the i.h. as with [T$\otimes$] above, considering
$B=\mathbf{E}_{k-1};{\dual { (\oplusm{\ell}{L}{A_\ell}) }}$
and
$\Delta_k \vdashB c_k : E_k$ where $E_k=\withm{\ell}{L}{E_\ell} $ and $E_{\labl{l}} = \Box$,
so that $B= \mathbf{E}_k ;\dual{A}$
and 
$\Delta_i \vdash c_i : E_i $ for $i\leq k$.

\end{proof}
\hide{
\begin{lemma}[Queue prefix]\label{lemma:queue-prf} 
The proof rule below is admissible and invertible:
$$
\frac{
P \vdash \Delta_P,x{:}A
\quad
Q \vdash \Delta_Q, y{:}B
\quad q = \overline{c_k}
\quad B=\dual T_1{\parl}{...}{\parl} \dual T_k{\parl} \dual{A}
\quad
\Delta_i \vdash c_i{:} T_i {\parl} ()
	\quad {-}B	} 
{\cuti{\ou{x}{:}A \;[q] \; y{:}B}{P}{Q}
	\vdash \Delta_P,\Delta_Q, \Delta_1,...,\Delta_k} 
$$
\hide{
Let $\cut{\ou{x}:A \;[\overline{c_k}] \; y:B}{P}{Q} \vdash \Delta; \Gamma$,
with $c_i = \mathsf{clos}(z_i;R_i;\Delta_{R_i})$ where $0\leq k$. Then 

(1)
$\cut{\ou{x}:T_1\otimes \ldots \otimes T_k \otimes A \;[\nilm] \; y:B}{s_1;\cdots ; s_k;P} {Q}\vdash \Delta; \Gamma $ and $s_i = \pi(c_i)$.

(2) 
$\Delta=\Delta_{R_1},\ldots, \Delta_{R_k},\Delta'$ and 
$R_i\vdash \Delta_{R_i},z_i:T_i$ and $B=\dual T_1\parl\cdots \parl \dual T_k \parl \dual A$.
}
\end{lemma}
\begin{proof}
To check that the rule is admissible, directly derive the
conclusion using $k$ applications of [TCut-$\otimes$].
To check inversion, we proceed by
induction in the derivation of
$\cuti{\ou{x}:A \;[\overline{c_k}] \; y:B}{P}{Q} \vdash \Delta; \Gamma$.
 
  (Case of [TCutB]) We have $k=0$ so $q=\nilm$, $A$ positive and $B$ negative. 

 (Case of [TCut-$\one$]) Not applicable, since $c_k\neq \checkmark$.

(Case of [TCut-$\otimes$])
We have $\cut{\ou{x}:A \;[\overline{c_{k-1}} @ \mathsf{clos}(z_k;R_k;\Delta_{R_k})   ] \; y:B}{P}{Q} \vdash \Delta; \Gamma$ derived from
$\cut{\ou{x}:T_k{\otimes} A \;[  \overline{c_{k-1}} ]\; y:B}{\potimes{x}{z_k}{R_k}{P}}{Q} \vdash \Delta, \Delta'; \Gamma$.
By i.h. $\potimes{x}{z_k}{R_k}{P}\vdash \Sigma_{k},x:T_k\otimes A$,
$Q \vdash \Delta_Q, y:B$ where $B=\dual T_1\parl\ldots \parl \dual {(T_k\otimes {A})}$
and 
$\Delta_i \vdash c_i : T_i\parl () $ for $i\leq k-1$.
By inversion ([T$\otimes$])
$P \vdash \Delta_P,x:A$,
 $\Sigma_{k} = \Delta_k,\Delta_P$ and 
 $R_k \vdash \Delta_k, z_k:T_k$, so
 $\Delta_k \vdash c_k : T_k\parl () $.
We conclude $P\vdash \Delta_{P},x:A$,
$Q \vdash \Delta_Q, y:B$,
$B=\dual T_1\parl\ldots \parl \dual T_k\parl \dual{A}$
and 
$\Delta_i \vdash c_i : T_i\parl () $ for $i\leq k$.

(Case of [TCut-$\oplus$]) Similar to [T$\otimes$].

\hide{
By i.h.
$\cut{\ou{x}:T_1\otimes \ldots T_k \otimes A \;[\nilm] \; y:B}{s_1;\ldots ;s_{k-1};s_k;P} {Q}\vdash \Delta; \Gamma $.

(2) By (1) and inversion on the typing derivation of  
$s_1;\cdots; s_k;P \vdash x:A,\Delta_R$.
}

\end{proof}
}

\queuefull*

\begin{proof}
We have
$\cuti{\ou{x}:\emptyset \;[\overline{c_k}] \; y:B}{\zero}{Q} \vdashB \Delta; \Gamma$
derived from [TCut$\one$] (a) or [TCut$!$] (b). In case (a)
$\cuti{\ou{x}:\one \; [{\overline{c_k}}]\;  y:B}{\pone{x}}{Q} \vdashB \Delta; \Gamma$
with $c_k=\checkm$. We conclude by Lemma \ref{lemma:queue-prf} .
 In case (b)
$\cuti{\ou{x}:!A \; [{\overline{c_k}}]\;  y:B}{\pbang{x}{z}{R}}{Q} \vdashB \Delta; \Gamma$ with $c_k=\closB(z,R)$.
By Lemma \ref{lemma:queue-prf}, $\pbang{x}{z}{R}\vdash z:!A;\Gamma$, so $\Gamma; \vdash c_k=\closB(z,R):?\dual A$ and we conclude 
with $B = \mathbf{E}_{k-1};C$ and
 $C=?\dual A$.
\end{proof}

\queuenonempty*
\begin{proof}
If $A$ is void then immediate by Lemma \ref{lemma:queue-full}.
If $A$ is negative, suppose $q=\nilm$. By Lemma~\ref{lemma:queue-prf}  $k=0$ and  $\dual A = B$. But then $B$ positive, contradiction. So $q\neq\nilm$.
\end{proof}

\typepreservation*
\begin{proof}
We verify that rules for  $\stackrel{\mathsf{\tiny B}}{\equiv}$ (Fig.~\ref{fig:equivB})  (resp.  $\stackrel{\mathsf{B}}{\to}$ (Fig.~\ref{fig:redB})) are type preserving.

(Case [fwdp])

(Case [fwdp])
$P = \cuti{\ou z\ass A \; [q_1] \; x\ass \dual B}{Q'}{\cuti{\ou y\ass B\; [q_2]\; w\ass C}{\fwd{x}{y}}{P'}}\vdashB\Delta;\Gamma$.
If $q_2=\nilm$ and $B = \dual{C}$,
hence $ Q=\cuti{\ou z \ass A\; [q_1]\; w\ass C}{Q'}{P'} \vdashB\Delta;\Gamma$.

Otherwise $q_2\neq\nilm$. Let
$F_2 = \cuti{\ou y\ass B\; [q_2]\; w\ass C}{\fwd{x}{y}}{P'}
$ where  $F_2 \vdashB \Delta_2, x\ass \dual B$ and $\Delta =\Delta_1, \Delta_2$,
and $q_2=\overline{c_k}$.
By Lemma \ref{lemma:queue-prf},
$\fwd{x}{y} \vdashB x\ass \dual B, y\ass B$, 
and $C = \mathbf{E}_k ; \dual B$, $\Gamma;\Delta_i\vdash c_i\ass E_i$,
$P\vdash w\ass C,\Delta_P;\Gamma$ and $\Delta_2 = \Delta_1,\ldots,\Delta_k,\Delta_P$.

Let 
$F_1 =  \cuti{\ou z \ass A\; [q_1]\; x\ass \dual B}{Q'}{F_2}$
 with $F_1 \vdashB \Delta_1, x\ass B;\Gamma$ and
$q_1=\overline{d_l}$.
By Lemma \ref{lemma:queue-prf},
$Q' \vdashB\Delta_Q,z\ass A;\Gamma$, $\Gamma;\Delta'_j \vdash d_j:F_j$ and $\dual B = \mathbf{F}_l ; \dual A$
and $\Delta_1 = \Delta'_1,\ldots,\Delta'_l,\Delta_Q$.
By Lemma \ref{lemma:queue-prf} we get 
$\cuti{z\ass A \; [ q_2@q_1 ] \; w\ass C}{Q'}{P'}\vdashB \Delta;\Gamma$.

(Case of [$\otimes$]) Let
$\cuti{\ou x:A\otimes C \;[q] \;y:B}{\potimes{x}{z}{P}{Q}}{R}\vdashB \Delta;\Gamma$.
By Lemma \ref{lemma:queue-prf},
$\potimes{x}{z}{P}{Q} \vdashB\Delta_P,x:A$ and
$R \vdash \Delta_R, y:B$ and
$\Delta_i\vdash c_i:E_i$ for $1\leq i\leq k$,
and $ B = \mathbf{E}_k;\dual (A\otimes C)$.
By inversion [T$\otimes$], $\Delta_P=\Delta'_P,\Delta''_P$
and $Q\vdashB \Delta''_P, x:C;\Gamma$ and $P\vdashB \Delta'_P, z:A, \Gamma$.
So $\Gamma;\Delta'_P\vdash c_{k+1} =  \clos(z,P):\dual A\parl \dual C$.
Let $E_{k+1}=\dual A \parl \Box$ so $B=\mathbf{E}_{k+1};\dual C$.
By Lemma \ref{lemma:queue-prf}, $\cuti{\ou x : C\;[q@\mathsf{clos}(z,P)] |\; y:B}{Q}{R} 
\vdashB \Delta;\Gamma$.

(Case [$\parl$])
Let
$\cuti{\ou x:A \; [ \overline{c_k}  ] \; y:B}{P}{
\pparl{y}{w}{Q}} \vdashB\Delta$ where $\overline{c_k} = \mathsf{clos}(z,R)@q$.
By Lemma \ref{lemma:queue-prf},
$P \vdash\Delta_P,x:A$ and
$\pparl{y}{w}{Q} \vdashB \Delta_Q, y:B$ and
$\Delta_1\vdash \mathsf{clos}(z,R):T_1\parl \Box$ and
$\Delta_i\vdashB c_i:E_i$ for $2\leq i\leq k$,
and $ B = \dual T_1\parl E_2[..E_k[\dual A]..]$.
We then have $R\vdashB \Delta_1,z:T_1$.
By inversion on [T$\parl$],  
$Q \vdash \Delta_Q,w:\dual T_1,y:C$ where $C = E_2[ .. E_k[\dual A]..]$.
We consider the case where $T_1$ and $A$ are positive, 
and check that
$\cuti{\ou x:A\;[q]\;y:C}{P}{(\cuti{\ou z_1:T_1\;[\nilm]\; w:\dual T_1}{R_1}{Q})}\vdashB \Delta$
by [TCutB] on  $R_1$ and $Q$ and Lemma \ref{lemma:queue-prf}.
Other cases are handled in a similar way.
\end{proof}
\begin{lemma}[Barbs inversion]\label{lemma:barbs} 
	Let $P\vdash \Delta;\Gamma$ and $\obs{P}{x}$.
\begin{enumerate}

\item If $x\in \Gamma$ then $P\equiv \cut{*}{\pcopy{x}{y}{Q}}{R}$, Otherwise $x:A\in\Delta$ and
\item
    $P\equiv \cut {*}{\fwd{x}{y}}{Q} $, or
\item
	If $A=\one$ then $P\equiv\pone{x}$.
\item
	If $A=B\otimes C$ then $P\equiv\potimes{x}{y}{Q}{R}$.
\item
	If $A=\bot$ then $P\equiv\cut{*}{\pbot{x}{Q}}{R}$.
\item
	If $A=B\parl C$ then $P\equiv \cut{*}{\pparl{x}{y}{Q}}{R}$.
\item
	If $A=?B$ then $P\equiv \cut{*}{\pwhy{x}{Q}}{R}$.
\item
	If $A=!B$ then $P\equiv \cut{*}{\pbang{x}{y}{Q}}{R}$.
\end{enumerate}

\end{lemma}

\liveness*
\begin{proof}
By induction on the derivation for $P \vdash \Delta; \Gamma$, and case analysis on the last typing rule. 

(Case of [TcutB]) we have $P = \cuti{ y:\dual A\;[\nilm]\;\ou z:A}{P_1}{P_2}\vdashB \Delta',\Delta;\Gamma$, derived from
$P_1\vdashB \Delta',y:\dual A;\Gamma$ and
$P_2\vdashB \Delta,z:A;\Gamma$, where $A$ is positive.
 By the i.h. we conclude that $P_2\to$  or $\obs{P_2}{x_2}$. 
 So if $P_2\to$, then $P\to$ .
 Otherwise, $\obs{P_2}{x_2}$. 
 If $x_2\neq z$ then $\obs{P}{x_2}$ with $x_2\in\Delta$.
 
 Otherwise $x_2=z$. By Lemma \ref{lemma:barbs} (1,2,3), 
 either $P_2 \equiv \cut{*}{\fwd{z}{v}}{R}$ (a), or
  since $A$ is 
 positive,   
$P_2\equiv\pone{z}$ (b) or $P_2\equiv\potimes{z}{w}{Q}{R}$ (c).
  
 In case (a) 
 if $v\not\in \Delta$,
 then $P$ reduces by [fwdp] (both $x_2$ and $v$ are bound in $P$ by cuts),
 otherwise  $v\in\Delta$ and $\obs{P_2}{v}$ and $\obs{P}{v}$.
 In case (b) $P\to$ by [$\one$] and 
 in case (c) by [$\otimes$]).
 
 (Case of [TCut$!$])
Let $\cutBi{x:A}{y}{R}{Q} \vdashB \Delta; \Gamma$
derived from $R \vdashB y:A; \Gamma$ and  $Q \vdashB \Delta; \Gamma, x:\dual A$. By i.h. either $Q\to$ or $\obs{Q}{z}$. If $z\neq x$ then $P\to$
or $\obs{P}{z}$. If $z=x$ then by Lemma \ref{lemma:barbs}
$Q\equiv \cut{*}{\pcopy{x}{y}{Q}}{R}$ and $P\to$ by [call].

(Case of [TCut$\otimes$])
We have $P=\cuti{\ou{x}:A \;[q  ] \; y:B}{P_1}{P_2} \vdashB \Delta; \Gamma$ 
where $q=r @ \mathsf{clos}(z_k;R_k)  $
derived from
$P_1 \vdashB x:A,\Delta_{P_1}; \Gamma$.
By  i.h. $P_1\to$ or $\obs{P_1}{x_1}$. 
If $x\neq x_1$ then $\obs{P_1}{x_1}$ and $\obs{P}{x_1}$.
If $A$ is positive, by the same reasoning as above for $P_2$
we conclude that $P\to$ or $\obs{P}{w}$.

Otherwise, $A$ is negative.
By  i.h. $P_2\to$ or $\obs{P_2}{z}$. 
If $z\neq y$ then $\obs{P}{x}$.
Otherwise, $\obs{P_2}{y}$.
By Lemma \ref{lemma:queue-prf}, $B = \dual T\parl C$.
By Lemma \ref{lemma:barbs} (2,6), 
either $P_2 \equiv \cut{*}{\fwd{y}{v}}{Q'}$ (a) or
$P_2\equiv \cut{*}{\pparl{y}{w}{Q'}}{Q''}$ (b).
For case (a) we conclude as in [TCutB] above that $P\to$ or $\obs{P}{v}$,
for case (b) $P\to$ by [$\parl$].

(Case of [TCut$\oplus$]) Similar to [TCut$\otimes$]

(Case of [TCut$\one$])
Let $P = \cuti{\ou{x}:\emptyset\; [q] \; y:B}{\zero}{P_2} \vdashB \Delta; \Gamma$ 
where $q=r @ \checkm$
derived from
$\cuti{\ou{x}:\one \; [{r}]\;  y:B}{\pone{x}}{P_2} \vdashB \Delta; \Gamma$.
By i.h. $P_2\to$  or $\obs{P_2}{z}$. 
If $z\neq y$ then $\obs{P}{x}$.
Otherwise, $\obs{P_2}{y}$.
By Lemma \ref{lemma:queue-full}, either $B = \dual T\parl C$ (a) or $B = \bot$ (b).
If $P_2 \equiv \cut{*}{\fwd{y}{v}}{Q'}$, then as for [TCutB]
we conclude that $P\to$ or $\obs{P}{v}$.
For (a), by Lemma \ref{lemma:barbs} (6), $P_2\equiv \cut{*}{\pparl{y}{w}{Q'}}{Q''}$,
and $P\to$ by [$\parl$].
For (b), by Lemma \ref{lemma:barbs} (5), $P_2\equiv\cut{*}{\pbot{y}{Q'}}{Q''}$,
and $P\to$ by [$\bot$].
\end{proof} 

\subsection{Proofs of Section~\ref{sec:CLLCLLB}: Correspondence between $\CLL$ and $\Blang$}

We give some examples of commutations used in the proofs, those for other pairs
$\tobx{} / \tobn{}$ are checked in a similar way.

\begin{lemma} \label{lemma:read-top}
Commutation (pos-neg-promote):
$$
\begin{array}{lll}
\cuti{\ou x \; [c@q] \; y}{\potimes{x}{z}{R}{P}}{\pparl{y}{w}{Q}} \tobx{}\\
\cuti{\ou x \; [c@q@v] \; y}{{P}}{\pparl{y}{w}{Q}} \tobn{}\\
\cuti{\ou x \; [q@v] \; y}{{P}}{\cuti{u [] w}{S}{Q}}
\vspace{8pt}\\
\cuti{\ou x \; [c@q] \; y}{\potimes{x}{z}{R}{P}}{\pparl{y}{w}{Q}} \tobn{}\\
\cuti{\ou x \; [q] \; y}{\potimes{x}{z}{R}{P}}{\cuti{s [] u}{R}{Q}} \tobx{}\\
\cuti{\ou x \; [q@v] \; y}{{P}}{\cuti{u [] w}{S}{Q}}
\end{array}
$$
\end{lemma} \label{lemma:commute-same}
\begin{lemma} \label{comm-read-write}
Commutation (pos-neg-seq-commute):
$$
\begin{array}{lll}
E_1[\cuti{c \; [q_1] \; \ou a}{X}{\potimes{a}{z}{R}\pparl{b}{w}{Q}}] \tobx{}\\
E_2[\cuti{\ou d \; [c@q_2] \; b}{{Y}}{\pparl{b}{w}{Q}} \tobn{}
E_3[\cuti{\ou x \; [q@v] \; b}{Y}{\cuti{u [] w}{S}{Q}}]
\vspace{8pt}\\
E_1[\cuti{c \; [q_1] \; \ou a}{X}{\potimes{a}{z}{R}\pparl{b}{w}{Q}}]\equiv\\
E_2[\cuti{\ou d \; [c@q_2] \; b}{X}{\pparl{b}{w}{\potimes{a}{z}{R}{Q}}}] \tobn{}
\\
F_3[\cuti{c \; [q_1] \; \ou a}{Z}{\cuti{u [] w}{S}{\potimes{a}{z}{R}{Q}}} \tobx{}
\\
E_3[\cuti{\ou x \; [q@v] \; b}{Y}{\cuti{u [] w}{S}{Q}}]
\end{array}
$$
\end{lemma}

\begin{lemma} \label{fwd-neg-comm}
Commutation (fwd-neg-comm):
$$
\begin{array}{lll}
E[\cuti{\ou z \; [q_1] \; x}{P}{\cut{\ou y\; [c@q_2]\; b}{\fwd{x}{y}}{\pparl{b}{w}{Q}}}]
\toba{} \\
E[\cuti{\ou z\; [c@q_2 @ q_1]\; b}{P}{\pparl{b}{w}{Q}}] \tobn{}\\
E[\cuti{\ou z\; [q_2 @ q_1]\; b}{P}{\cuti{u [] w}{S}{Q}}]\\
\\
E[\cuti{\ou z \; [q_1] \; x}{P}{\cut{\ou y\; [c@q_2]\; b}{\fwd{x}{y}}{\pparl{b}{w}{Q}}}]
\tobn{} \\
E[\cuti{\ou z \; [q_1] \; x}{P}{\cut{\ou y\; [q_2]\; b}{\fwd{x}{y}}{\cuti{u [] w}{S}{Q}}}]
\toba{} \\
E[\cuti{\ou z\; [q_2 @ q_1]\; b}{P}{\cuti{u [] w}{S}{Q}}]
\end{array}
$$
\end{lemma} \label{lemma:commute-same}

\commutemain*




\begin{proof}
(1) Either  (a) the reductions are in the same cut, or 
(b)  the reductions are in different cuts.
For (a), if reductions match we conclude.
Otherwise they commute (by Lemma~\ref{lemma:read-top}
) 
so  
$P_1  \tobn  S   \tobx{} P_2$ for some $S$.
For (b), if reductions are independent (different threads), commute, and we conclude.
If the reductions are dependent (same thread), they commute (by Lemma~\ref{comm-read-write}) and we conclude.

(2) We consider two cases: either  (a) the reductions are in the same cut, or 
(b)  the reductions are in different cuts.
For (a), the reductions commute (by Lemma \ref{fwd-neg-comm}),
so we have 
$P  \tobn{}  S   \toba Q$ for some $S$.
For (b), the reductions are independent and commute.

(3) By (1) and (2).

(4) Assume $P_1 \tobx{} N$. Since the reductions in $S \tor{} P_2$ must act on the same cut with
an initially empty queue,
and all reductions $N \tobaxm{\epsilon} S$ are on empty cuts, 
the reduction $P_1 \tobx{} N$ must act on a different cut of
all those involved and thus commutes with $N \tobaxm{\epsilon} S \tor{} P_2$,
and we conclude (b).
If $P_1 \toba {} S$ then either this redex generates one empty cut,
so $P_1 \tobax{\epsilon} S$, and we conclude (a), or ,as in (b), it
must be independent of $N \tobaxm{\epsilon} S \tor{} P_2$.

\end{proof}

\simtrivial*
\begin{proof}
Each cut reduction of $\CLL$ is simulated 
by two reduction steps of $\mathsf{B}$ in sequence: 
[$\otimes \parl$] by [$\otimes$] followed by [$\parl$];
[$\one\bot$] by [$\one$] followed by [$\bot$];
[$!?$] by [$!$] followed by [$?$]. 
In a closed $\CLL$ process a [fwd] reduction  has the form 
$\cuti{x:A}{Q}{{\cut{y:{A}}{\fwd{x}{y}}{P}}}
\to \cuti{x:A}{Q}{\subs{x}{y}P} 
$, which, for $A{+}$,
reduces by [fwdB] as 
$
\cuti{\ou x\ass{A} \; [\nilm] \; z\ass\dual A}{Q}{\cut{\ou w\; [\nilm]\; y}{\fwd{z}{w}}{P}} \to_a \cuti{\ou x\ass A\; [\nilm]\; y \ass \dual A}{Q}{P}
$. 
\end{proof}

\simul*





\begin{proof}

By induction on $P (\toban{})^*  P'$.
 
(Base) We have $P \equiv P' \tobn Q $, hence (1).

(Inductive) Case $P  \toban{} P' \tobanm{} \tobn Q $.  Let (r0) $P  \toban{} P'$.

By i.h. for $P'$ 
there is $R'$ so that
(c1) $ P' \tobn R'$ and $R' \tobanm{} Q$, or
(c2) $ P' \tor{} R'$ and $R' \tobanm{}  {\;}  Q$, or
 (c3) $ P' \tobaxm{\epsilon} \tor{} R' $ and $R' \tobanm{} Q$ .

 Case (c1). We have $P  \toban{} P' \tobn R'$ and $R' \tobanm{}  {\;}  Q$.
 
By Lemma \ref{lemma:commute-main} (3) on  $P  \toban{} P' \tobn R'$, either 
$P  \tor{} R'$ and we conclude (2) ($R=R'$) or
 there is $R$ such that $P  \tobn{} R \toban R'$ and we conclude (1).
 
 
  Case (c2). We have $P  \toban{} P' \tor{} R'$ and $R' \tobanm{}  {\;}  Q$.
  By Lemma \ref{lemma:commute-main} (4) either  
  (a) $P \tobax{\epsilon} P' \tor{} R'$ 
  or (b) $P \tor{} R'' \toban{} R'$ for some $R''$.
In case (a) we conclude (3) ($R=R'$), 
in case (b) we conclude (2) ($R=R''$).

Case (c3). We have 
$ P  \toban{} P' \tobaxm{\epsilon} \tor{} R' $ and $R' \tobanm{} Q$.
  By Lemma \ref{lemma:commute-main} (4) either  
  (a) $P \tobax{\epsilon} P' $ 
  or (b) $P \tor{} R'' \toban R'$.
  In case (a) 
  $P \tobaxm{\epsilon}\tor{} R'$
  we conclude (3) ($R=R'$), in case (b) we conclude (2)
  ($R=R''$).
\end{proof}

\hide{For $A{-}$ we have $\cuti{x:A}{Q}{{\cut{y:{A}}{\fwd{x}{y}}{P}}}
\equiv \cuti{x:\dual A}{\cuti{y:{\dual A}}{P}{\fwd{y}{x}}}{Q} $
which 
reduces by [fwdB] as 
$
\cuti{\ou y:\dual{A} \; [\nilm] \; z:A}{P}{\cuti{\ou w\; [\nilm]\; x}{\fwd{z}{w}}{Q}} \to_a \cuti{\ou y\; [\nilm]\; x}{P}{Q}
\equiv \cut{x:A \; [\nilm]\; \ou y:\dual A}{Q}{P}
$.}

\hide{
2. The actions enabled on $P^\dagger$
and $P$ are the same, so $P$ imitates, in one step, the
$\CLL$ reduction corresponding to the sequential application of
the given positive and negative rule in $\Blang$. 
}

\hide{

3.  Assume
$P^\dagger \tobxz{} P' \tobnx{m} Q$ (a). 
Let us call \emph{main cut} the cut $m$ in which the last
reduction step $P' \tobnx{m} Q$ acts.
To enable this last (negative) step in a reduction sequence starting from $P^\dagger$, there must exist the matching
prior positive reduction step in the main cut within $P^\dagger \tobxz{} P'$.
Iterating Lemma~\ref{lemma:comut}(1), 
we commute reductions in (a) so that  $P^\dagger \tobxz{} C \tobxz{} P' \tobn  Q$ (b) for some $C$,
where all reductions in $P^\dagger \tobxz{m}  C$ and $P' \tobnx{m}  Q$ occur in the main cut $m$ and no reduction in $C \tobxz{} P'$ occurs in the main cut.
By iterating Lemma~\ref{lemma:comut}(2), 
we can commute reductions in (b) so that (c)
$P^\dagger \tobxz{} C' \tobn D \tobxz{} Q$ and
all reductions in $P^\dagger \tobxz{m} \; C'$ and
$C' \tobnx{m} D$
are on the main cut.
By iterating Lemma~\ref{lemma:comut}(3), 
we we can commute reductions in (c) so that (d)
$P^\dagger  \tobx{m} \tobnx{m}  Q' \tobxz{} Q$.
Hence by picking $R$ such that $R^\dagger = Q'$,
by 2 we conclude
$P \to R$, and $R^\dagger \tobxz{} \; Q$.

4. Similar to 1 for the [fwdB] reduction case.
}
\cllb*
\begin{proof}
1. Iterating Lemma~\ref{lemma:sim-trivial}. 
2.
We first check (A) if $P^\dagger \tobanm{} \tobn Q$ then either there is $R$ such that
$P\to R$ and $R^\dagger \tobanm{} Q$ or
there is $R$ such that
$P\to R$ and $R^\dagger \tobanm{} \tobn Q$. 
Subproof: Assume $P^\dagger \tobanm{} \tobn Q$.
By Lemma~\ref{lemma:simul}  either (2)
$ P^\dagger \tor{} R'$ and $R' \tobanm{}  {\;}  Q$ for some $R'$, or (3)
$ P^\dagger \tobaxm{\epsilon} \tor{} R' $ and $R' \tobanm{} Q$ for some $R'$
((1) cannot apply, since in $P^\dagger$ all queues are empty).
In case (a) we have $P \to R$ and $R' \tobanm{} Q$ with $R'=R^\dagger$ for
some $R$ since all cuts are empty in $R'$.
In case (b) $ P^\dagger \tobaxm{\epsilon} P' \tor{} R' $ and $R' \tobanm{} Q$.
Then $P \Rightarrow Q$ by [fwd] and $Q \to R'$,
for some $Q$ where $Q^\dagger = P'$. So $P\Rightarrow R'$ and 
 $R'=R^\dagger$ for
some $R$ since all cuts are empty in $R'$. To conclude (2), we iterate (A).

\end{proof}

\subsection{Proofs of Section~\ref{sec:soundcore}: Correctness of the core SAM}

\begin{lemma}[Sanity]
\label{lemma:sanity}

\noindent
1. Let $(P,H) \Mapsto (P',H')$. Then if $\fn{P}\subseteq dom(H)$ then
 $\fn{P'}\subseteq dom(H')$.

\noindent

\end{lemma}

Lemma~\ref{app:soundstep} (2) below expresses a useful invariant of session records in 
machine states encoding well-typed processes:
The right (reader) endpoint is of negative type. The left (writer) endpoint is of positive type unless $P=\zero$ or $x\in\fn{P}$.

\begin{lemma}[Soundness]\label{app:soundstep} Let $P\vdashB\emptyset$. 

\noindent
1. If $enc(P) \Mapsto D \expa C$  then there
is $Q$ such that $P \to\cup\equiv Q$ and
$C = enc(Q)$. 

\noindent
2. Let $P(H)$ denote property: 
If $\srecs{x\ass A}{q}{P}{y\ass B} \in H$ then $B-$ and if ($x \in \fn{P}$ or $P=\zero$)
then $A-$ else $A+$.
 
Then if $enc(P)=(\mathcal A,H) \Mapsto (Q,H')$ and $P(H)$ then $P(H')$.
\end{lemma}

\begin{proof} 

(1) Let $enc(P) \Mapsto D \expa C$.
Let $enc(P) = (\mathcal{A},H)$ for some action $\mathcal{A}$.
Then  $P \equiv E[\mathcal{A}]$
for some cut context $E[-]$. 
We consider each reduction $\Mapsto$.

(Case of [Sfwd]) 
Let $enc(P) =(\fwd{x}{y}, H[\srecs{z\ass A}{q_1}{P_1}{x}][\srecs{y}{q_2}{P_2}{w}]) \Mapsto D$ where $D=(P_2, H[\srecs{z}{q_2 @ q_1}{P_1}{w}]) \expa C$.
By (2) since $P_1(z)$, then $A^-$.
Therefore $P \equiv E[\fwd{x}{y}] \equiv E'[\cuti{\ou z \; [q_1]\; x]}{P_1}{
\cuti{\ou y \; [q_1]\; w]}{\fwd{x}{y}}{P_2}
}]$.
Then $P\to Q \equiv E'[\cuti{\ou z \; [q_2@q_1]\; w]}{P_1}{P_2}]$ and 
 $enc(Q) = C$.

(Case of [S$\one$]) 
Let $enc(P) =(\pone{x}, H[\srecs{x}{q}{R}{y}]) \Mapsto 
(R, H[\srecs{x}{q @ \checkmark}{\zero}{y}]) = D \expa C$.
Therefore, $P \equiv E[\pone{x}] \equiv E'[\cuti{\ou x \; [q]\; y]}{\pone{x}}{R}]$
and $P\to Q = E'[\cuti{\ou x \; [q@\checkmark]\; y]}{\zero}{R}]$
and $enc(Q) = C$.

(Case of [S$\bot$]) 
Let $enc(P) = (\pbot{y}{R}, {H}[\srecs{ x}{\checkmark}{\zero}{y}])  \Mapsto D$ and    $D= (R, {H})$.

Hence, $P \equiv E[\pbot{y}{R}] \equiv E'[\cuti{\ou x:\emptyset \; [\checkmark]\; y]}{\zero}
{\pbot{y}{R}}]
\to 
Q = E'[R]$   and 
$D \expa C$
and $enc(Q) = C$.

(Case of [S$-$]) 
Let  $ enc(P) = (\mathcal A^-(x), 
\mathcal{H}[\srecs{x\ass A}{{q}}{R}{y}]) \Mapsto D$ with

$D=
(R, H [\srecs{x}{{q}}{\mathcal A^-(x)}{y}])$.
By typing $A-$ and $x\in\fn{\mathcal A^-(x)}$. Hence (2).

Hence, $P \equiv E[A^-(x)] \equiv $

$E'[\cuti{\ou x \; [q]\; y]}{A^-(x)} {R}] 
\equiv  Q$ and
$D \expa C$
and $enc(Q) = C$.

(Case of [S$\otimes$]) 
Let $ enc(P) = ( \potimes{x}{z}{R}{U}, H[\srecs{x}{q}{S}{y}])  \Mapsto D$ 

where $D= (U, H[\srecs{x}{q @ \clos(z,R)}{S}{y}])$.

So $P \equiv E[\potimes{x}{z}{R}{U}] \equiv
E'[\cuti{\ou x \; [q]\; y]}{\potimes{x}{z}{R}{U}}{S}]
\to\\
E'[\cuti{\ou x \; [q @ \clos(z,R)]\; y]}{U}
{S}] 
\equiv Q$ 
and
$D \expa C$
and $enc(Q) = C$.

(Case of [S$\parl$]) 
Let $ enc(P) = ( \pparl{y}{w\ass -}{U}, H[\srecs{x}{\clos(z,R)@q }{S}{y}]) \Mapsto D$ 
where $D= (R, H[\srecs{z}{\nilm}{U}{w}][\srecss{x}{q }{S}{y}])$ and $S(x)$ with $x:A$ and
$A-$ or $S=\zero$.

\noindent
So, $P \equiv E[\pparl{y}{w\ass -}{U}] \equiv
E'[\cuti{\ou x \; [ \clos(z,R)@ q]\; y]}{\pparl{y}{w\ass -}{U}}
{S}] \to Q $.
If  $q\neq \nilm$ then
$Q = E'[ \cuti{\ou x\;[q]\;y}{S}{\cuti{ \ou z\;[\nilm]\; w}{R}{U}}]$.
%
%
If $q= \nilm$, then
$Q = E'[ \cuti{ x\;[\nilm]\;\ou y}{S}{\cuti{ \ou z\;[\nilm]\; w}{R}{U}} ]$.
In both cases, $B \expa C$
and $enc(Q) = C$. 

(Case of [S$\oplus$]) 
Let $ enc(P) = ( \labl{l};R, H[\srecs{x}{q}{S}{y}]) \Mapsto D =
(R, H[\srecs{x}{q @ \labl{l}}{S}{y}])$.
So $P \equiv E[\labl{l};R] \equiv E'[\cuti{\ou x \; [q]\; y]}{\labl{l};R}{S}]
\to 
E'[\cuti{\ou x \; [q @ \labl{l}]\; y]}{R}
{S}]\equiv Q$, and
$D \expa C$
and $enc(Q) = C$.

(Case of [S$\with$])  Similar to [S$\parl$].

\end{proof}

\ready*
\begin{proof}
The property trivially holds for $(P,\emptyset) $ for
$\fn{P}=\emptyset$ and $H=\emptyset$.
We proceed by transition induction, assuming that  
$S=(P,H) $ is ready, and $\mathcal S \Mapsto \mathcal 
S'= (P,H')$, we check that $S'$ is ready.
 
 (Case of [Sfwd])
 Let $R = \fwd{x}{y}$
and
$S = (R,H') \Mapsto  S' =
( P, H'')$, where
$H'=H[\srecs{z}{q_1}{Q}{x} ][\srecs{y}{q_2}{P}{w} ]$
and
$H''=H[\srecs{z}{q_2@q_1}{Q}{w} ]$.
By i.h, $R$ is $H'$-ready and $H'$ is ready.
Thus $z:A$ is negative or void , $Q(z)$ is
$H'$-ready, and
$P(w)$ is $H',\{w\}$-ready.
Then $Q(z)$ is $H''$-ready, and
$H''$ is ready. Then
$P(w)$ is $H''$-ready and $S'$ is ready.

(Case of [SCut])
Let $R = \cuti{\ou x:T \; [q] \; y:B}{P}{Q}$
and
$S = (R,H) \Mapsto  S' =
( P, H')$, where
$H'=H[\srecs{x}{\nilm}{Q}{y} ]$ and $T$ positive.
By i.h., $R$ is $H$-ready hence $P(x)$ is $H,\{x\}$-ready
and $Q(y)$ is $H,\{y\}$-ready. Then $P(x)$ is $H'$-ready, and $H'$ is ready. We conclude that $S'$ is ready.

(Case of [S$\one$])
Let $R = \pone{x}$
and 
$S = (R, H') \Mapsto 
(Q, H'') = S'$
where 
$H'=H[\srecs{x:A}{q}{Q}{y}]$
and
$H''=H[\srecs{ x:\emptyset}{q @ {\checkmark}}{\zero}{y}]$.
By the i.h, $H'$ is ready and $Q(y)$ is $H',\{y\}$-ready.
Hence $Q(y)$ is $H''$-ready, since $x:\emptyset$. 
Clearly, $H''$ is ready, since $\zero(x)$ is $H''$-ready.
We conclude $S'$ ready.

(Case of [S$\bot$])
Let $R = \pbot{y}{U}$, $H'= H[\srecs{x}{{\checkmark}}{\zero}{y}]$
and  $S = (R, 
H') \Mapsto 
(U, H) =  S'$.
By the i.h, 
$R$ is $H'$-ready and $H'$ is ready.
Therefore $U$ is $H$-ready and $H'$-ready.
We conclude $\mathcal S'$ ready.

(Case of [S$-$]) 
Let $R=\mathcal A^-(x)$ and $S = (R,H') \Mapsto (Q,H'') = S'$
where $H'=H[\srecs{x{:}A}{{q}}{Q}{y}]$ and $H''=H [\srecs{x}{{q}}{\mathcal A^-(x)}{y}]$ with $A-$. By i.h., $\mathcal A^-(x)$ is $H'$-ready and $H'$ is ready. Hence, $Q(y)$ is 
$H',\{y\}$-ready. We have $\mathcal A^-(x)$ is $H''$-ready ($x$ positive endpoint).
Thus 
$H''$ is ready, $Q$ is $H''$-ready (since $A$ is negative), and $S'$ is ready.

(Case of [S$\otimes$]) 
Let $R=\potimes{x}{z}{V}{U}$ and $S=(R,H') \Mapsto (U,H'') = S'$
where $H'=H[\srecs{x}{q}{Q}{y}]$ and $H''=H[\srecs{x}{q @ \clos(z,V)}{Q}{y}]$.
By i.h, $H'$ is ready and $R$ is $H'$-ready. 
$V$ is $H',\{z\}$-ready and  $U$ is $H'$-ready.
Then $V$ is $H'',\{z\}$-ready and  $U$ is $H''$-ready.
Also $Q(y)$ is $H',\{y\}$-ready, and $Q(y)$ is $H'',\{y\}$-ready.
We conclude $H''$-ready and thus $S'$ ready.

(Case of [S$\parl$]) 
Let $R=\pparl{y}{w:-}{U}$ and $S=(R,H') \Mapsto (V,H'') = S'$
where $H'=H[\srecs{x}{\clos(z,V)@q }{Q}{y}]$ and $H''= H[\srecs{z}{\nilm}{U}{w}][\srecss{x}{q }{Q}{y}]$.
By i.h, $V$ is $H'$-ready since $z$ is positive, and thus $H''$-ready.
By typing we have $Q(x)$ and thus $Q(x)$ is $H'$-ready.
Also $R$ is $H'$-ready, so $U$ is $H',\{w\}$-ready. 

If $q\neq \nilm$ then
 $H''= H[\srecs{z}{\nilm}{U}{w}][\srecs{x}{q }{Q}{y}]$.
 Then $U$ is $H'',\{w\}$-ready,
   $Q(x)$ is $H''$-ready, $H''$
is ready, and $S'$ is ready.

If $q= \nilm$ then $H''= H[\srecs{z}{\nilm}{U}{w}][\srecs{y}{\nilm }{Q}{x}]$. Then $U$ is $H'',\{w\}$-ready
, $Q(x)$ is $H'',\{x\}$-ready ($x$ changed polarity),
$H''$ is ready, and $S'$ is ready.
The case $w-$ is handled in the same way.

(Case of [S$\oplus$])  
Let $R=\labl{l}\ x;U$ and $S=(R,H') \Mapsto (U,H'') = S'$
where $H'=H[\srecs{x}{q}{Q}{y}]$ and $H''=H[\srecs{x}{q @ \labl{l}}{Q}{y}]$.
By i.h, $H'$ is ready and $R$ is $H'$-ready. Therefore
 $U$ is $H''$-ready.
We conclude that $S'$ is ready.

(Case of [S$\with$]) 
Let $R=\pcasem{y}{\ell}{L}{P_\ell}$ and $S=(R,H') \Mapsto (P_{\labl{l}},H'') = S'$
where $H'=H[\srecs{x}{\labl{l}@q }{Q}{y}]$ and $H''= H[\srecss{x}{q }{Q}{y}]$.
By i.h, $H'$ is ready, $R$ is $H'$-ready, so $Q(x)$ is $H'$-ready, and $P_{\labl{l}}$ is $H'$-ready.

If $q\neq \nilm$ then
 $H''= H[\srecs{x}{q }{Q}{y}]$
is ready and $P_{\labl{l}}$ is $H''$-ready since $Q(x)$ is 
$H'$-ready,  so $S'$ is ready.

If $q= \nilm$ then $H''= H[\srecs{y}{\nilm }{Q}{x}]$. Notice that $Q$ is $H'',\{x\}$-ready  ($x$ changed polarity).
Then $H''$ is ready.
 $P_{\labl{l}}$ is $H''$-ready, so $S'$ is ready.
\end{proof}

\samprogress*
\begin{proof}
Since $P$ is live then $enc(P) =  S = ({\mathcal A},H)$ for some action $\mathcal A$.

(Case of [fwdB]) $\mathcal A = \fwd{x}{y}$.
Wlog, assume that $x{:}A$ is negative, so $y{:}\dual A$ is positive.
 By typing, $x$ and $y$ must be endpoints of different cuts (and session records).
Case $H=H'[\srecs{z}{q_1}{U(z)}{x}][\srecs{y}{q_2}{V(w)}{w}]$.
We have
$S\Mapsto S'$ 
by [Sfwd].
Case $H=H'[\srecs{x}{q_1}{U(z)}{z}][\srecs{y}{q_2}{V(w)}{w}]$.
Then $\fwd{x}{y} = \mathcal A^-(x)$, and  $S\Mapsto S'$ by [S$-$].
\hide{
---

Wlog, assume that $x{:}A$ is positive, so $y{:}\dual A$ is negative.
 By typing, $x$ and $y$ must be endpoints of different cuts (and session records).
Case $H=H'[\srecs{x}{q_1}{U}{z}][\srecs{w}{q_2}{V}{y}]$.
We have
$S\Mapsto S'$ 
by [Sfwd].
Case $H=H'[\srecs{x}{q_1}{U}{z}][\srecs{y}{q_2}{V}{w}]$.
Then $\fwd{x}{y} = \mathcal A^-(y)$, and  $S\Mapsto S'$ by [S$-$].
}

(Case of [$\one$]) 
$\mathcal A = \pone{x}$ so $S\Mapsto S'$ 
by [S$\one$].

(Case of [$\bot$]) 
$\mathcal A = \pbot{y}{R'}$. 
If $H=H'[\srecs{y{:}A}{q}{Q}{x}]$,
then $S\Mapsto S'$ by [S$-$].
Otherwise, let $H=H'[\srecs{x\ass A}{q}{\mathcal A}{y}]$. 
By Lemma~\ref{lemma:ready}, $\mathcal A$ is $H'$-ready,
so $A$ is negative or void.  By Lemma~\ref{lemma:queue-non-empty}  $q\neq\nilm$. By Lemma~\ref{lemma:queue-full}, $q=\checkmark$. So $S\Mapsto S'$ by [S$\bot$].  

(Case of [$\otimes$])  $\mathcal A = \potimes{x}{z}{R}{Q}{}$ so $S\Mapsto S'$ by [S$\otimes$].

(Case of [$\parl$])  $\mathcal A = \pparl{y}{w}{U}$.
If $H=H'[\srecs{y{:}A}{q}{Q}{x}]$,
then $S\Mapsto S'$ by [S$-$].
Otherwise, $H=H'[\srecs{x:A}{q}{Q}{y}]$.
By Lemma~\ref{lemma:ready}, $\mathcal A$ is $H'$-ready,
so $A$ is negative or void.  By Lemma~\ref{lemma:queue-non-empty}  $q\neq\nilm$. By Lemma~\ref{lemma:queue-prf} or Lemma~\ref{lemma:queue-full}, we must have $q=\clos(z,R)@q'$.
So $S\Mapsto S'$ by [S$\parl$].

(Case of [$\oplus$])  By [S$\oplus$].

(Case of [$\with$])  
$\mathcal A = \pcasem{y}{\ell}{L}{P_\ell}$.
If $H=H'[\srecs{y{:}A}{q}{Q}{x}]$,
then $S\Mapsto S'$ by [S$-$].
Otherwise, $H=H'[\srecs{x{:}A}{q}{Q}{y}]$.
By Lemma~\ref{lemma:ready}, $\mathcal A$ is $H'$-ready,
so $A$ is negative or void.  By Lemma~\ref{lemma:queue-non-empty}  $q\neq\nilm$.
By Lemma~\ref{lemma:queue-prf} or Lemma~\ref{lemma:queue-full}, we must have $q=\labl{l}@q'$ so $S\Mapsto S'$ by [S$\with$].
\end{proof}

\subsection{Proofs and Additional Definitions of
  Section~\ref{sec:sam-mix}: The SAM for full $\CLL$}

To alleviate the proofs we assume that only the unrestricted names are recorded in environment, for the linear names we use names as heap locations, and $\alpha$-conversion as needed to generate fresh names.
\begin{definition}[Encode]
\label{def:encodeexp}
Given $P\vdashB\emptyset;\emptyset$
we define $enc(P)=C$ as $enc(\emptyset,P,\emptyset) \stackrel{\mathsf{cut}*}{\Mapsto} C$ where $enc(\E,P,H) \expa C$ is defined by the rules
$$
\begin{array}{clll}
\displaystyle
\frac{
enc(P(x),H[\srecs{ x\ass A}{q}{Q}{y\ass B}])
\expa  C}{
enc(\cuti{\ou x\ass A [q]\; y\ass B]}{P}{Q}, H)\expa  C}
& \mbox{\rm ($A+$ )}
\vspace{6pt}\\
\displaystyle
\frac{
enc(Q(y),H[\srecs{ x\ass A}{q}{P}{y\ass B}] )
\expa  C
}{
enc(\cuti{\ou x\ass A\; [q]\; y\ass B ]}{P}{Q}, H)\expa  C}
\quad \quad & \mbox{\rm ($A-$ or $P=\zero$ )}
\vspace{6pt}\\
\displaystyle
\frac{
enc(\E\subs{\closB(y,\E,R)}{x}, P),H )
\expa  C
}{
enc(\E,\cutBi{x}{y}{R}{P}, H)\expa  C}
\quad \quad & 
\vspace{6pt}\\
enc(\mathcal{A},H) 
\expa 
(\mathcal{A},H) & \mbox{\rm ($A \in \mathcal A$)}
\end{array}
$$
\end{definition}

\begin{lemma}[Soundness-!]\label{app:soundstep} Let $P\vdashB\emptyset;\emptyset$. 

\noindent
1. If $enc(P) \Mapsto D \expa C$  then there
is $Q$ such that $P \to\cup\equiv Q$ and
$C = enc(Q)$. 

\noindent
2. Let $P(H)$ denote property: 
If $\srecs{x\ass A}{q}{P}{y\ass B} \in H$ then $B-$ and if ($x \in \fn{P}$ or $P=\zero$)
then $A-$ else $A+$.
 
Then if $enc(P)=(\E,\mathcal A,H) \Mapsto (\E,Q,H')$ and $P(H)$ then $P(H')$.
\end{lemma}

\begin{proof} 

(1) Let $enc(P) \Mapsto D \expa C$.
Let $enc(P) = (\mathcal{A},H)$ for some action $\mathcal{A}$.
Then  $P \equiv E[\mathcal{A}]$
for some cut context $E[-]$. 
We consider each reduction $\Mapsto$.

(Case of [S$!$]) 
Let $enc(P) =(\E,\pbang{x}{z}{U}, H[\srecs{x}{q}{R}{y}]) \Mapsto D$ where
$D = (R, H[\srecs{x}{q @ \closB(z,\E,U)}{\zero}{y}])  \expa C$.

Therefore, $P \equiv E[\pbang{x}{z}{U}] \equiv E'[\cuti{\ou x \; [q]\; y]}{
\pbang{x}{z}{U}}{R}]$
and $P\to Q = E'[\cuti{\ou x \; [q@\closB(z,\E,U)]\; y]}{\zero}{R}]$
and $enc(Q) = C$.

(Case of [S$?$]) 
Let $enc(P) = (\E,\pwhy{y}{R}, {H}[\srecs{ x}{\closB(z,\E,U)}{\zero}{y}])  \Mapsto D$ and    $D= (\E'\subs{y}{\closB(z,\E,U)}, R, {H}) \expa C$.

Hence, $P \equiv E[\pwhy{y}{R}] \equiv E'[\cuti{\ou x:\emptyset \; [\closB(z,\E,U)]\; y]}{\zero}
{\pwhy{y}{R}}]
\to 
Q = E'[R]$   and 
$D \expa C$
and $enc(Q) = C$.

(Case of [Scall]) 
Let $ enc(P) = (\E,\pcopy{y}{w\ass -}{U}, H) \Mapsto D$ 
where $D= (R, H[\srecs{z}{\nilm}{U}{w}]) \expa C$.
Let $\E(y) = \closB(z,\E,R)$.

\noindent
So, $P \equiv E[\pcopy{y}{w\ass -}{U}] \equiv
E'[\cutBi{y}{z}{R}{\pcopy{y}{w\ass -}{U}}
] \to Q $.

Then
$Q = E'[\cutBi{x}{y}{P}{\cuti{z [\nilm] w}{R}{{U}}} ]$
and $enc(Q) = C$. 
\end{proof}

\begin{lemma}[Readiness-!]\label{lemma:ready-fullsam}
Let $P\vdash\emptyset;\emptyset$ and $(\emptyset,P,\emptyset) \stackrel{*}{\Mapsto} S$.
Then $\mathcal S$ is ready.
\end{lemma}
\begin{proof}
The property trivially holds for $(\emptyset,P,\emptyset) $ for
$\fn{P}=\emptyset$ and $H=\emptyset$.
We proceed by transition induction, assuming that  
$S=(\E,P,H) $ is ready, and $\mathcal S \Mapsto \mathcal 
S'= (\E',P,H')$, we check that $S'$ is ready.

(Case of [SCut!])
Let $R = \cutBi{x}{y}{R}{Q}$
and
$S = (\E,Q,H) \Mapsto  S' =
( \E' P, H)$, where
$E'=\E\subs{\closB(y,\E,R)}{x}$.
By i.h., $P$ is $H$-ready but then $Q$ is $H$ since $x$ is unrestricted.
We conclude that $S'$ is ready. $R$ is also $H,\{y\}$-ready, since by typing
the only free linear  name in $R$ is $y$.

(Case of [S$!$])
Let $R = \pbang{x}{z}{U}$
and 
$S = (\E, R, H') \Mapsto (\G, Q, H'') = S'$
where 
$H'=H[\srecs{x:!A}{q}{\G,Q}{y}]$
and
$H''=H[\srecs{ x:\emptyset}{q @ \closB(z,\E,U)}{\emptyset,\zero}{y}]$.

By the i.h, $H'$ is ready and $U(z)$ is $H',\{z\}$-ready.

Hence $Q(y)$ is $H''$-ready, since $x:\emptyset$. 

Clearly, $H''$ is ready, since $\zero(x)$ is $H''$-ready.
We conclude $S'$ ready.

(Case of [S$?$])
Let $R = \pwhy{y}{Q}$, $H'= H[\srecs{x}{{\closB(z,\E,U)}}{\emptyset,\zero}{y}]$
and  $S = (\E,R, 
H') \Mapsto 
(\E\subs{\closB(z,\E,U)}{y}, Q, H) =  S'$.
By the i.h, 
$R$ is $H'$-ready and $H'$ is ready.
Therefore $Q$ is $H$-ready and $H'$-ready.
Also $U$ is $H',\{z\}$-ready (no linear names).
We conclude $ S'$ ready.

(Case of [Scall]) 
Let $R=\pcopy{y}{w:-}{U}$ and $S=(\E,R,H) \Mapsto (\F,V,H') = S'$
where $\E(y)= \closB(z,\F,V)$ and 
 $H'= H[\srecs{z}{\nilm}{\E,U}{w}]$.

By i.h, $U$ is $H,\{w\}$-ready and 
$V$ is $H$-ready since $z$ is positive.
Hence $U$ is $H',\{w\}$-ready, so $H'$ is ready,
and $V$ is $H'$-ready.
$U$ is $H',\{w\}$-ready. 
We conclude $ S'$ ready.

\end{proof}

\begin{theorem}[Progress-!]
Let $P\vdash\emptyset;\emptyset$ and $P$ live. Then $enc(P) = S \Mapsto S'$.
\end{theorem}
\begin{proof}
Since $P$ is live then $enc(P) =  S = (\E,{\mathcal A},H)$ for some action $\mathcal A$. The interesting case is:

(Case of [$?$])  $\mathcal A = \pwhy{y}{U}$.
If $H=H'[\srecs{y{:}A}{q}{\F,Q}{x}]$,
then $S\Mapsto S'$ by [S$-$].
Otherwise, $H=H'[\srecs{x:A}{q}{Q}{y}]$.
By Lemma~\ref{lemma:ready-fullsam}, $\mathcal A$ is $H'$-ready,
so $A$ is negative or void.  By Lemma~\ref{lemma:queue-non-empty}  $q\neq\nilm$. By Lemma~\ref{lemma:queue-full}, we must have $q=\closB(z,\G,R)$.
So $S\Mapsto S'$ by [S$?$].
\end{proof}

\subsection{Proofs of Section~\ref{ref:mix}: Concurrent Semantics of
  Cut and Mix}
\label{app:mix}

Here we sketch the basic structure of progress for the SAM with 
concurrent cut and sequential threads.  First we define the encoding
of $\Blang$ processes with annotated concurrent cuts (such annotation
is silent for any purpose other than the concurrent
execution strategy in the SAM). For simplicity we address core SAM (no exponentials) and concurrent cuts.
The proof can be easily extended to any other pair of dual types types,
and already allow the concurrent non-deterministic sychronisation of
different threads.
\begin{figure}
$$
\begin{array}{clll}
\displaystyle
\frac{
enc(P,H)
\expa  (A,H')
\quad
enc^c(T,H') \expa  (T',H'')
}{
enc^c(\{P\}\uplus T,H') \expa  (\{A\}\uplus T',H'')}
\quad \quad & \rm{[thr]}
\vspace{6pt}\\
\displaystyle
\frac{
enc^c(T,H) 
\expa  C
}{
enc^c(\{\zero\}\uplus T,H) \expa  C
}
\quad \quad & \rm{[\zero]}
\vspace{6pt}\\
\displaystyle
\frac{
enc^c(\{P,Q\}\uplus T,H) 
\expa  C
}{
enc^c(\mixP \{\mix{P}{Q}\} \uplus T,H) \expa  C
}
\quad \quad & \rm{[mix]}
\vspace{6pt}\\
\displaystyle
\frac{
enc^c(\{P,Q\}\uplus T,H[\srecc{x}{q}{y}]) 
\expa  C
}{
enc^c(\cutp{x:A^+ [q] y:B}{P}{Q}\uplus T,H) \expa  C
}
\quad \quad & \rm{[pcut]} 
\end{array}
$$
\caption{Encoding of annotated $\Blang$  of into concurrent SAM}
\end{figure}
\begin{figure}
$$
\begin{array}{ll}
(\pone{x}, H[\srecc{x}{q}{y}]) \Mapsto 
(\zero, H[\srecc{x}{q @ {\checkmark}}{y}])\quad\quad\;&  \text{\rm[S$\one$c]} 
\vspace{6pt}\\
(\pbot{y}{ P}, 
H[\srecs{x}{{\checkmark}}{y}]) \Mapsto 
(P, H) &  \text{\rm[S$\bot$c]}
\vspace{6pt}\\
(\potimes{x}{z}{R}{Q}, H[\srecc{x}{q}{y}]) \Mapsto 
(Q, H[\srecc{x}{q @ \clos(z,R)}{y}]) & \text{\rm[S$\otimes$c]}
\vspace{6pt}\\
(\pparl{y}{w{:}-}{Q}, H[\srecc{x}{\clos(z,R)@q }{y}{}]) \Mapsto 
(R, H[\srecs{z}{\nilm}{ Q}{w}][\srecss{x}{q }{y}{y}])
& \text{\rm[S$\parl-$c]}
\vspace{6pt}\\
(\fwd{x}{y}, H[\srecc{z}{q_1}{x}[\srecc{y}{q_2}{w}]]) \Mapsto 
(\zero, H[\srecc{z}{q_2@q_1}{w}])\quad\quad\;&  \text{\rm[SfwdPc]} 
\\
\\
\srecss{a}{q}{}{b} \deff \mbox{if }(q = \nilm)\mbox{ then }
\srecc{b}{q}{}{a}\mbox{ else }
\srecc{a}{q}{}{b}
\end{array}
$$
\caption{Transition rules for concurrent actions (sample)}
\end{figure}
\begin{figure}
$$
\begin{array}{c}
\displaystyle
\frac{
(P, H) \Mapsto (P', H') 
}{
(P\uplus T, H) \Mapsto (P'\uplus T, H') 
}\; \text{\rm[Srun]}
\quad
\displaystyle
\frac{
(T, H) \Mapsto (T', H') 
}{
(\zero\uplus T, H) \Mapsto (T', H') 
}\quad \text{\rm[S$\zero$p]}
\vspace{6pt}\\
(\cutp{x:A^+}{P}{Q}\uplus T, H) \Mapsto 
(\{P,Q\subs{y}{x}\}\uplus T\}, H[\srecc{x}{\nilm}{y}] ) \quad \text{\rm[SCutp]}
\vspace{6pt}\\
((\mix{P}{Q})\uplus T, H) \Mapsto 
(\{P,Q\}\uplus T\}, H[\srecc{x}{\nilm}{y}] ) \quad \text{\rm[SMixp]}
\end{array}
$$
\caption{Transition rules for concurrent actions}
\end{figure}

\noindent 
We show that SAM execution 
on well-typed processes simulates $\Blang$ reduction. 
\samcsound*

\begin{proof}
Induction of $enc(P) \expa B$.
\end{proof}

To prove progress along the lines of Theorem~\ref{theorem:progress}
we introduce a simplified notion of observation (cf. Fig.~\ref{fig:obs}) for 
our scenario of concurrent
actions on heap allocated concurrent sessions, as follows. 

\hide{
$\obs{\pone{x}}{x}^H$ iff $\srecc{x}{q}{y}\in H$.

$\obs{\pbot{y}{Q}}{y}^H$ iff $\srecc{x}{q}{y}\in H$.
}

\vspace{6pt}
$\obs{P}{x}^H$ iff $\obs{P}{x}$ and ($\srecc{x}{q}{y}\in H$ or $\srecc{x}{q}{y}\in H$).

\begin{lemma}[Liveness-c]
\label{livenessc} Let  $P\vdash \Delta$ and $enc^c(P,H) \expa (T, H')$.
For every non-empty $S\subseteq T$
 there is $P_i\in S$ such that (a) $(P_i,H')\Mapsto C$ or (b) 
  $\obs{P_i}{y}^{H'}$.
\end{lemma}

\begin{proof}
By induction on the derivation of $enc(P,H) \expa (T, H)$.

(Case of [thr]) 
By Theorem~\ref{samprogress} (Progress), we have $(P,H)\Mapsto (A,H')$ 
for any action on a sequential cut. Otherwise $P$ is an action in a 
concurrent cut, and we have (b). By i.h., the property holds for $T',H''$.
Hence it holds for $(\{A\}\uplus T',H'')$.

(Case of [mix]) Let $enc^c(\{P,Q\}\uplus T,H[\srecc{x}{q}{y}]) \Mapsto (T',H')$

By i.h. (a) holds for $T',H'$ and we conclude.

(Case of [cutp]) Let $R = \cutp{\ou x\ass A\; [q] \;y\ass B}{P}{Q}\uplus T$ and $enc(R,H) \Mapsto (T',H')$.
Let $U_Q=T\downarrow \fn{Q}$. By i.h., there is $Q_i\in U_P$ 
such that (a) $(Q_i,H')\Mapsto $ 
or (b)  
$\obs{Q_i}{z}^{H'}$.
In case (a) we conclude (a). Otherwise assume (b).
If $z\neq y$ then
(a) holds for $T',H'$. 
If $z=y$ then (by typing) $B$ is negative
and $Q_i = {\mathcal A}^{B}(y)$ is a negative 
action on $y$ of type $B$ and $\srecc{w}{q}{y}\in H$.

If $q\neq\nilm$ then we have $(Q_i, H) \Mapsto$ by [S$B$c] (we denote 
by [S$B$c] the negative reduction rule associated to the negative type $B$).

Otherwise $q=\nilm$ and $A = \dual B$ so $P\neq\zero$.

Let $U_P=T\downarrow \fn{P}$. 
By i.h., there is $P_i\in U_P$ 
such that $(P_i,H)\Mapsto $ 
or (b) 
and $\obs{P_j}{u}^{H'}$.
In case (a) we conclude. Otherwise assume (b).
If $u\neq x$ then
(a) holds for $T',H'$. 
So $u=x=w$. Now, if $A$ is negative or void, then by Lemma~\ref{lemma:queue-non-empty} $q\neq\nilm$, contradiction. So $A$ is positive and $P_j={\mathcal A}^A(x)$ is a positive
action on $x$ of type $A$. So we
have $(P_j, H'[\srecc{x}{q}{y}]) \Mapsto (\zero,H'[\srecc{x}{q@c}{y}])$ by [S$A$c]
and $(Q_i, H'[\srecss{x}{q}{}{y}])  \Mapsto C$ by [S$B$c].
\end{proof}

\hide{
\begin{proof}
By induction on the derivation of $enc(P,H) \expa (T, H)$.

(Case of [thr]) 
By Theorem~\ref{samprogress} (Progress), we have $(P,H)\Mapsto (A,H')$ 
for any action on a sequential cut. Otherwise $P$ is an action in a 
concurrent cut, and we have (b). By i.h., the property holds for $T',H''$.
Hence it holds for $(\{A\}\uplus T',H'')$.

(Case of [mix]) Let $enc^c(\{P,Q\}\uplus T,H[\srecc{x}{q}{y}]) \Mapsto (T',H')$

By i.h. (a) holds for $T',H'$ and we conclude.

(Case of [cutp]) Let $R = \cutp{x:\one [q] y:\bot}{P}{Q}\uplus T$ and $enc(R,H) \Mapsto (T',H')$.
Let $U_Q=T\downarrow \fn{Q}$. By i.h., there is $Q_i\in U_P$ 
such that (a) $(Q_i,H')\Mapsto $ 
or (b)  
$\obs{Q_i}{z}^{H'}$.
In case (a) we conclude (a). Otherwise assume (b).
If $z\neq y$ then
(a) holds for $T',H'$. 
If $z=y$ then (by typing) $Q_i = \pbot{z}{Q'}$ and $\srecc{w}{q}{y}\in H$.

If $q\neq\nilm$ then we  have $(Q_i, H) \Mapsto$ by [S$\bot$c].

If $q=\nilm$ then $A=\dual B$ so $P\neq\zero$.

let $U_P=T\downarrow \fn{P}$. By i.h., there is $P_i\in U_P$ 
such that $(P_i,H)\Mapsto $ 
or (b) 
and $\obs{P_j}{u}^{H'}$.
In case (a) we conclude. Otherwise assume (b).
If $u\neq x$ then
(a) holds for $T',H'$. 
If $u=x=w$ then (by typing) $P_j = \pone{z}$ and then we
have $(P_j, H'[\srecc{x}{\nilm}{y}]) \Mapsto (\zero,H'[\srecc{x}{\checkmark}{y}]$ by [S$\one$c]
and $(Q_i, H'[\srecc{x}{\nilm}{y}]) \Mapsto C$ by [S$\bot$c].
\end{proof}
We conclude:
\samcprogress*
\begin{proof}
By Lemma~\ref{livenessc} since $\fn{P}=\emptyset$.
\end{proof}

---}


\end{document}